\documentclass[10pt,reqno,a4]{amsart}

\usepackage[latin1]{inputenc}
\usepackage{multirow}
\usepackage{graphicx,amssymb,amsmath,amscd,amsfonts,color, amsthm,enumerate}
\usepackage{hyperref,commath}

\title[]{The random matrix regime of Maronna's M-estimator for observations corrupted by elliptical noises }
\author[Kammoun et al]{Mohamed-Slim Alouini and Abla Kammoun}
\date{\today}

\newtheorem{theorem}{Theorem}[section]
\newtheorem{remark}{Remark}[section]
\usepackage{appendix}
\newtheorem{lemma}[theorem]{Lemma}
\newtheorem{corollary}[theorem]{Corollary}

\newtheorem{proposition}[theorem]{Proposition}
\newtheorem{assumption}{Assumption A-\hspace{-0.15cm}}

\newcommand{\bdm}{\begin{displaymath}}
\newcommand{\edm}{\end{displaymath}}
\newcommand{\bea}{\begin{eqnarray*}}
\newcommand{\eea}{\end{eqnarray*}}

\numberwithin{equation}{section}
\theoremstyle{remark}

\DeclareMathOperator*{\tr}{Tr}

\newcommand{\asto}{\overset{\rm a.s.}{\longrightarrow}}

\begin{document}

\maketitle
\begin{abstract}
This article studies the behavior of the Maronna robust scatter estimator $\hat{C}_N\in \mathbb{C}^{N\times N}$ of a sequence of observations $y_1,\cdots,y_n$ which is composed of a $K$ dimensional signal drown in a heavy tailed noise,  i.e $y_i=A_N s_i+x_i$ where $A_N \in \mathbb{C}^{N\times K}$ and $x_i$ is drawn from elliptical distribution. 
  In particular, we prove that as the population dimension $N$, the number of observations $n$ and the rank of $A_N$ grow to infinity at the same pace and under some mild assumptions,  the robust scatter matrix can be characterized by a random matrix $\hat{S}_N$ that follows a standard random model. Our analysis can be very useful for many applications of the fields of statistical inference and signal processing.
\end{abstract}

\section{Introduction}
Estimation of covariance matrices is at the heart of the theory of multivariate statistical analysis \cite{murihead}. Its importance can be seen from its broad range of applications including financial data analysis, statistical signal processing, and wireless communication.    A natural way to estimate covariance matrices is represented by the sample covariance matrix. Given $n$ observations  $y_1,\cdots,y_n$, of size $N$, independent, and identically distributed, (i.i.d) then the sample covariance matrix is given by $\frac{1}{n}\sum_{i=1}^n y_iy_i^{*}$.  The popularity of the sample covariance matrix essentially comes from its low-complexity and the existence of a good understanding of its behaviour in two asymptotic regimes: $n$ goes to infinty while $N$ is fixed when $N$ and $n$ go to infinity with the same pace.
Recent advances in the theory of large random matrices have made it clear that in the  second asymptotic regime, the sample covariance matrix is no longer consistent. Conventional estimation methods that are based on the use of the sample covariance matrix are thus inefficient when the number of observations and their dimension become commensurate and large.
Such scenario naturally arises in current array processing applications where the trend is to employ large antenna arrays. 
Based on a deep understanding of the behaviour of the sample covariance matrix, a new wave of detection methods \cite{CAR09,BIA09,nadler-10} and subspace estimation techniques \cite{mestre-08,vallet-10,hachem-13} has recently emerged. Although consistent, these methods are bound by the fact that they still fundamentally rely on the sample covariance matrix, their consistency being obtained by resorting to a deep analysis of its asymptotic behaviour. Nevertheless, the use of the sample covariance matrix can lead to poor performances, especially when observations are drawn from heavy tailed distributions or contain outliers.  
In such situations, the use of robust covariance estimators has been aknowledged as an efficient solution to combat the presence of outliers.  Although references to robust techniques are traced back to the eighties with the works of Huber \cite{huber-81} and Maronna \cite{maronna-76}, the study of their performance has been often restricted to the conventional regime where the number of observations is  too large as compared to their dimensions. It was only recenlty that new tools have been developed in  \cite{couillet-13,couillet-13a,couillet-pascal-2013} which allow to analyse the behaviour of robust Maronna's scatter estimators. The main contributors are Couillet et al. who established that the robust scatter estimator can be well-approximated in the asymptotic regime by a random matrix that follows a standard random model. One of the key advantages of this result, is that it allows to bring back the asymptotic analysis of robust-scatters to that of an other random object for which an important load of results already exist.

Despite their high value, these works have been derived only for the case of pure noise observations. While the case of a low rank signal observations can be dealt with  by resorting to easy adaptations of the approach of \cite{couillet-pascal-2013}, handling high-rank signal observations is much more challenging. Building on the tools developed in \cite{couillet-pascal-2013}, we propose in this work to analyse this difficult scenario. We show that in this case the adaption of the method in \cite{couillet-pascal-2013} is not immediate and necessitates the development of additional appropriate tools. Some of the required results that were of independent interest were submitted in an other work which can be found in \cite{smallest_eigenvalue}.

{\it Notations:} In the remainder of this work, we shall denote $\lambda_1(X)\leq \cdots\leq \lambda_N(X)$ the real eigenvalues of $n\times n$ Hermitian matrix $X$. The  notation  $\|.\|$ will refer to the spectral norm of matrices and Euclidean norm for vectors, while $^*$ sill stand for the complex conjugate operator. The derivative of a differentiable function $f$ will be denoted by $f^{'}$.
 %In the first regime, it is known that this conventional estimation method is consistent in that it tends to the population covariance matrix when $n$ tends to infinity. %This behavior completely changes when $n$ and $N$ are of the same order of magntiude.    
\section{Assumptions and Main results} 
We start by introducing the data model under study. We consider $n$ sample vectors $y_1,\cdots,y_n\in\mathbb{C}^{N}$ satisfying:
$$
y_i=A_N s_i +x_i, i=1,\cdots,n,
$$
where $A_N$ is a $N\times K$ deterministic matrix and $x_1,\cdots,x_n$ are random vectors defined as:
$$
x_i=\sqrt{\tau_i}w_i
$$
with the scalars $\tau_1,\cdots,\tau_n\in\mathbb{R}_{+}$. 
Let $\overline{N}=K+N$.
We denote by $c_N=\frac{N}{n}$   and considers the following assumptions:
\begin{assumption}
\label{ass:regime}
For each $N$, $c_N<1$, $\overline{c}_{{N}}\geq 1$, and 
$$
c_{-} < \lim\inf_n c_N < \lim\sup_n c_N <c_{+},
$$
with $0<c_{-}<c_{+}<1$.
\end{assumption}
This paper studies the asymptotic behaviour of the Maronna's M-robust scatter estimator in the regime  of Assumption \ref{ass:regime}. We recall that the Maronna's M-robust estimator which we denote by $\hat{C}_N$ is given by the unique solution in $Z$ of the following equation:
\begin{equation}
Z=\frac{1}{n}\sum_{i=1}^n u\left(\frac{1}{N}y_i^*Z^{-1}y_i\right) y_iy_i^*.
\label{eq:Z}
\end{equation}
where function $u(\cdot)$ satisfies the following properties:
\begin{assumption}
\begin{enumerate}[i)]
\item Function $u(\cdot):[0,\infty)\to[0,\infty)$ is non-negative continuous and non-increasing,
\item The function $\phi(\cdot):x\mapsto xu(x)$ is increasing, bounded and continuously differentiable with $\lim_{x\to\infty}\phi(x)\triangleq \phi_{\infty} >1$ and $\phi'>0$.
\item $\phi_{\infty}< c_{+}^{-1}$.
\end{enumerate}
\label{ass:u}
\end{assumption}
and the scalars $\tau_i$ are such that:
\begin{assumption}
\begin{enumerate}[i)]
\item The random empirical measure $\nu_n=\frac{1}{n}\sum_{i=1}^n \delta_{\tau_i}$ converges weakly to $\nu$ which satisfies $\int x\nu(x)=1$,
\item There exists $\epsilon <1-\phi_{\infty}^{-1}<1-c_{+}$ and $m>0$ such that for all large $n$ a.s., $\nu\left(\left[0,m\right)\right) <\epsilon$.
\end{enumerate}
\label{ass:measure}
\end{assumption}
The conditions in Assumption \ref{ass:u}  are the same as those considered in \cite{couillet-pascal-2013}. It is worth observing that Assumption \ref{ass:u}-$ii)$ is different from the one considered by Maronna in \cite{maronna-76}, in that $\phi$ is not allowed to be constant on any open interval. However, Assumption \ref{ass:u}:-$iii)$ is much more adapted to the high-dimensional regime than Assumption $(D)$ p.53  of \cite{maronna-76}, which requires that $\phi_{\infty} >N$.

Assumption \ref{ass:measure} is different from the original assumption in \cite{couillet-pascal-2013} as we assume here the weak convergence of the empirical measure $\nu_n$. However, one can easily see by the Portmanteau lemma that Assumption \ref{ass:measure} will bring about the same useful requirements, namely the a.s. tightness of $\left\{\nu_n\right\}_{n=1}^{\infty}$, i.e., for each $\eta>0$, there exists $M>0$ such that with probability one, $\nu_n\left(\left[M,\infty\right)\right)<\eta$, along with the absence of a heavy mass concentrating close to zero ($\nu_n\left[0,m\right)<\epsilon$ for $n$ large enough a.s.).

The statistical hypothesis on $y_1,\cdots,y_n$ is detailed below:
\begin{assumption}
\begin{enumerate}[i)]
\item $w_1,\cdots,w_n\in\mathbb{C}^{{N}}$ are independent invariant complex zero-mean vectors with for each $i$, $\|w_i\|^2=N$ and are independent of $\tau_1,\cdots,\tau_n$,
\item $s_i\sim\mathcal{CN}(0,I_K),i=1,\cdots,n$ are independent standard Gaussian distributed vectors.
\item Define ${B}_N=A_NA_N^{*}$, then $\lim\sup\left\|B_N\right\|<\infty$ and $\lim\inf \frac{1}{N}\tr B_N >0$.
\end{enumerate}
\label{ass:statistical}
\end{assumption}
In addition to the above assumptions, the following hypothesis might be required:
\begin{assumption}
For each $a>b>0$, a.s.,
$$
\lim\sup_{t\to\infty} \frac{\lim\sup_n\nu_n(t,\infty)}{\phi(at)-\phi(bt)}=0.
$$
\label{ass:unbounded}
\end{assumption}

\begin{theorem}[Uniqueness]
Let Assumptions \ref{ass:regime}-\ref{ass:statistical} hold true. Then, for all large $n$, \eqref{eq:Z} admits a unique solution $\hat{C}_N$. Moreover, $\hat{C}_N$ is the limit of of the sequence $Z^{(t)}$ given by:
$$
Z^{(t+1)}=\frac{1}{n}\sum_{i=1}^n u\left(\frac{1}{N}y_i^{*}\left(Z^{t}\right)^{-1}y_i\right)y_iy_i^*,
$$
where $Z^{(0)}\succeq 0$.
\label{th:uniqueness}
\end{theorem}
\begin{theorem}
Let Assumptions \ref{ass:regime}-\ref{ass:unbounded} hold. Let $\hat{C}_N$ be given by Theorem \ref{th:uniqueness} when uniquely defined. Then,
$$
\|\hat{C}_N-\hat{S}_N\|\asto 0
$$
where
$$
\hat{S}_N=\frac{1}{n}\sum_{i=1}^n v(\delta_i)y_iy_i^*
$$
and $\delta_1,\cdots,\delta_n$ are the unique positive solutions in $x_1,\cdots,x_n$ to the following system of equations
\begin{equation}
x_i=\frac{1}{N}\tr \left(B_N+\tau_i I_N\right)\left(\frac{1}{n}\sum_{j=1}^n \frac{v(x_j)(B_N+\tau_jI_N)}{1+c\psi(x_j)}\right),
\label{eq:x}
\end{equation}
with the functions $v:x\mapsto \left(u\circ g^{-1}\right)(x)$, $\psi(\cdot):x\mapsto xv(x)$ and $g(\cdot):\mathbb{R}_{+}\to \mathbb{R}, x\mapsto x/(1-c_N\phi(x))$.
\label{th:asymptotic}
\end{theorem}
\begin{corollary}
Let Assumptions \ref{ass:regime}-\ref{ass:unbounded} hold true. Let $\hat{C}_N$ be the  solution of \eqref{eq:Z} when uniquely defined. Assume further that the empirical distribution $F^{B_N}$ converges in distribution to $F^{B}$, a cumulative distribution function and $c_N\to c$. Set $\chi_{\infty}$ and $\gamma_{\infty}$ the unique solutions to the following system of equations:
\begin{align*}
\chi_{\infty}&=\int \frac{y}{\int_0^{+\infty}\frac{v(\chi_{\infty}+t\gamma_{\infty})(y+t)}{1+c\psi(\chi_{\infty}+t\gamma_{\infty})}\nu(dt)}F^{B}(dy) \\
\gamma_{\infty}&=\int \frac{1}{\int_0^{+\infty}\frac{v(\chi_{\infty}+t\gamma_{\infty})(y+t)}{1+c\psi(\chi_{\infty}+t\gamma_{\infty})}\nu(dt)}F^{B}(dy).
\end{align*}
Then,
$$
\left\|\hat{C}_N-S_N\right\|\asto 0
$$
where $S_N=\frac{1}{n}\sum_{i=1}^n v(\chi_{\infty}+\tau_i\gamma_{\infty})y_iy_i^*$.
\end{corollary}
\begin{proof}
Let $\delta_1,\cdots,\delta_n$ be the solution of the system of equations \eqref{eq:x}. Let $T_N$ be given by:
$$
T_N=\left(\frac{1}{n}\sum_{j=1}^n \frac{v(\delta_j)(B_N+\tau_j I_N)}{1+c_N\psi(\delta_i)}\right)^{-1}.
$$
Let $\hat{\chi}_N=\frac{1}{N}\tr B_N T_N$ and $\hat{\gamma}_N=\frac{1}{N}\tr T_N$. Then,
$$
\delta_j=\hat{\chi}_N+\tau_j \hat{\gamma}_N, \hspace{0.1cm} j=1,\cdots,n.
$$
Noticing that $\hat{\chi}_N$ and $\hat{\gamma}_N$ satisfy:
\begin{align}
\hat{\chi}_N&=\frac{1}{N}\tr B_N\left(\frac{1}{n}\sum_{j=1}^n \frac{v(\hat{\chi}_N+\tau_j \hat{\gamma}_N)(B_N+\tau_j I_N)}{1+c_N \psi(\hat{\chi}_N+\tau_j\hat{\gamma_N})}\right)^{-1} \label{eq:chi}\\
\hat{\gamma}_N&=\frac{1}{N}\tr \left(\frac{1}{n}\sum_{j=1}^n \frac{v(\hat{\chi}_N+\tau_j \hat{\gamma}_N)(B_N+\tau_j I_N)}{1+c_N \psi(\hat{\chi}_N+\tau_j\hat{\gamma_N})}\right)^{-1}, \label{eq:gamma}
\end{align}
it is not difficult to see that solving the system of the $n$ equations in \eqref{eq:x} can be reduced to determining the solutions of a two equations system, whose solutions are $\hat{\chi}_N$ and $\hat{\gamma}_N$.
The control of $\delta_j$ in Lemma \ref{lemma:control_boundedness} allow us to ensure that $\hat{\chi}_N$ and $\hat{\gamma}_N$ are uniformly bounded for enough large $n$ a.s. Hence, there exists a subsequence over which $\hat{\gamma}_N$ and $\hat{\chi}_N$ converge to $\hat{\gamma}_{\infty}$ and $\hat{\chi}_\infty$. Taking the limits of both sides of \eqref{eq:chi} and \eqref{eq:gamma}, we obtain
\begin{align}
\chi_{\infty}&=\int \frac{y}{\int_0^{+\infty} \frac{v(\chi_{\infty}+t\gamma_{\infty})(y+t)}{1+c\psi(\chi_{\infty}+t\gamma_{\infty})}\nu(dt)}F^{B}(dy) \label{eq:chi_infty}\\
\gamma_{\infty}&=\int \frac{1}{\int_0^{+\infty} \frac{v(\chi_{\infty}+t\gamma_{\infty})(y+t)}{1+c\psi(\chi_{\infty}+t\gamma_{\infty})}\nu(dt)}F^{B}(dy) \label{eq:gamma_infty}.
\end{align}
Such limits are unique since the solutions of the systems of equations \eqref{eq:chi_infty} and \eqref{eq:gamma_infty} are unique in case they exist. The existence and unicity of the solutions of   \eqref{eq:chi_infty} and \eqref{eq:gamma_infty} essentially relies on showing that the following function
\begin{align*}
h&:\mathbb{R}_{+}^2 \to \mathbb{R}_{+}^2 \\
(x_1,x_2)& \mapsto (h_1(x_1,x_2),h_2(x_1,x_2))\triangleq\left(\int \frac{y}{\int_0^{+\infty}\frac{v(x_1+tx_2)(y+t)}{1+c\psi(x_1+tx_2)}\nu(dt)}F^{B}(dy),\int \frac{1}{\int_0^{+\infty}\frac{v(x_1+tx_2)(y+t)}{1+c\psi(x_1+tx_2)}\nu(dt)}F^{B}(dy)\right)
\end{align*}
is a standard interference function \cite{yates}, i.e it satisfies the three conditions of positivity, monotonicity and scalability that have been used in the proof of Theorem \ref{th:uniqueness}.
\end{proof}
\section{Numerical analysis}
In order to assess the accuracy of our results, we represent in Fig. \ref{fig:result}, the empirical estimate of the mean squared error (MSE) between the robust scatter estimate and $\hat{S}_N$ with respect to $N$  
$$
{\rm MSE}=\mathbb{E}\left\|\hat{S}_N-\hat{C}_N\right\|^2
$$
when $n=3N$ and $B_N=A_NA_N^{*}$ with $A_N$ is $N\times \frac{N}{2}$ having independent standard Gaussian entries with zero mean and variance $\frac{1}{K}$. We set $u(t)=\frac{1+\alpha}{t+\alpha}$, and $\alpha=0.5$. We note that the MSE decreases with $N$, thereby supporting the convergence of $\hat{C}_N$ to $\hat{S}_N$.
\begin{figure}[h]
\includegraphics[scale=0.5]{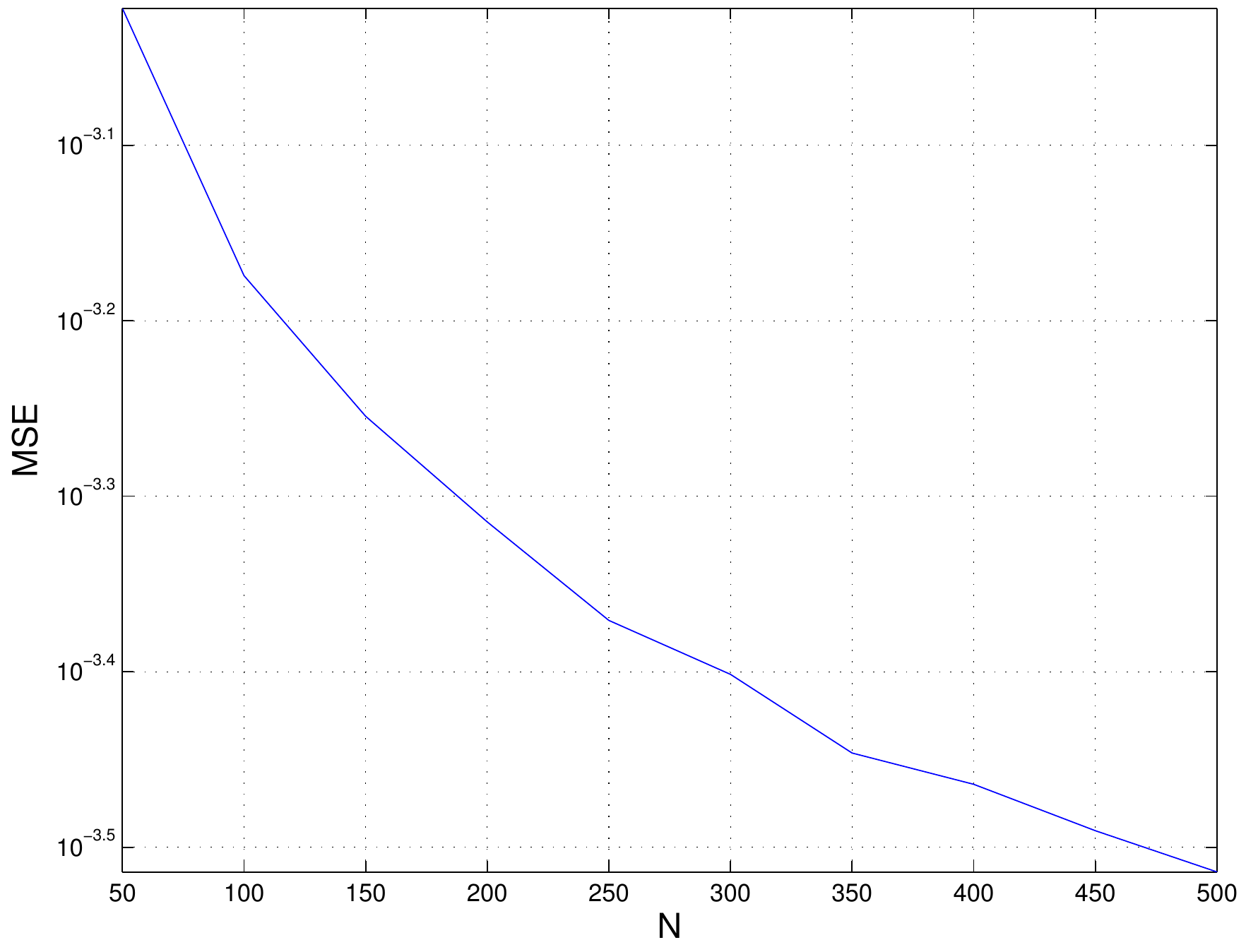}
\caption{MSE with respect to $N$.}
\label{fig:result}
\end{figure}
\section{Proofs}
\subsection{Heuristic Analysis}
The study of the asymptotic behaviour of robust scatter matrices  requires careful attention. The difficulty essentially lies in the rank-1 matrices present in the sum of \eqref{eq:Z} being dependent through  $\hat{C}_N$. At first sight, this observation might make us think that the asymptotic analysis of $\hat{C}_N$ is out of the framework of the standard random matrix theory. However, a careful investigation of the expression of $\hat{C}_N$ can lead us to replace $\hat{C}_N$ by a random object, whose analysis using the theory of random matrices is quite standard.

Hereafter, we develop some heuristics that will lead to determine the asymptotic random equivalent of $\hat{C}_N$. We believe that beyond their interest to the considered scenario, these heuristics can facilitate the understanding of the asymptotic behaviour of robust estimation techniques in the regime where the number of observations is of the same of order of the size of the population covariance matrix.

Building on the ideas of \cite{couillet-pascal-2013}, we will first start by deriving a new rewriting of $\hat{C}_N$ that will also  be extensively used in section \ref{sec:rigorous} devoted to the exposition of the rigorous proofs.  Let $\hat{C}_{(i)}$ be the matrix $\hat{C}_N$ where we remove $\frac{1}{n} u(\frac{1}{N}y_i^*\hat{C}_N^{-1}y_i)y_iy_i^*$, i.e.,
$$
\hat{C}_{(i)}=\hat{C}_N-\frac{1}{n}u\left(\frac{1}{N}y_i^*\hat{C}_N^{-1}y_i\right)y_iy_i^*.
$$
Applying the identity:
$$
(A-t zz^*)^{-1}z=\frac{A^{-1}z}{1-t z^*A^{-1}z}
$$
for any invertible $A$, vector $z$ and scalar $t$ such that $A-t zz^*$  is invertible, we obtain:
\begin{align*}
\frac{1}{N}y_i^*\hat{C}_{(i)}^{-1}y_i&=\frac{\frac{1}{N}y_i^*\hat{C}_N y_i}{1-\frac{1}{n}y_i^*\hat{C}_N^{-1}y_i u(\frac{1}{N}y_i^*\hat{C}_N^{-1}y_i)}\\
&=\frac{\frac{1}{N}y_i^*\hat{C}_N y_i}{1-c_N \phi\left(\frac{1}{N}y_i^*\hat{C}_N^{-1}y_i\right)}\\
&=g_N\left(\frac{1}{N}y_i^*\hat{C}_N^{-1}y_i\right),
\end{align*}
where $g_N:\left[0,\infty\right)\to\left[0,\infty\right), x\mapsto \frac{x}{1-c_N\phi(x)}.$ As $\phi$ is increasing and $\phi_{\infty}<c_N^{-1}$, function $g_N$ is positive increasing and maps $\left[0,\infty\right)$ to $\left[0,\infty\right)$. It is therefore invertible with inverse denoted by $g_N^{-1}$. We have thus:
$$
\frac{1}{N}y_i^*\hat{C}_N^{-1}y_i=g_N^{-1}\left(\frac{1}{N}y_i^*\hat{C}_{(i)}^{-1}y_i\right).
$$
We can therefore express $\hat{C}_N$ as:
\begin{align*}
\hat{C}_N&=\frac{1}{n}\sum_{j=1}^n \left(u\circ g_N^{-1}\right)\left(\frac{1}{N} y_j^*\hat{C}_{(j)}^{-1}y_j\right)y_jy_j^*\\
&=\frac{1}{n}\sum_{j=1}^n v\left(\frac{1}{N} y_j^*\hat{C}_{(j)}^{-1}y_j\right)y_jy_j^*
\end{align*}
with $v=u\circ g_N^{-1}$  positive and non-increasing.

This new rewriting of $\hat{C}_N$ is of fundamental importance.  It has two major advantages. First, it reveals that $\hat{C}_N$ is uniquely determined by $q_j=\frac{1}{N} y_j^*\hat{C}_{(j)}^{-1} y_j, j=1,\cdots,n$. This can be seen by noticing that a solution $\hat{C}_N$ to \eqref{eq:Z} exists and is unique if and only if  the following system of equations in $x_1,\cdots,x_n$:
$$
x_j=y_j^*\left(\frac{1}{n}\sum_{i=1,i\neq j}^n v(x_i) y_iy_i^*\right)^{-1}y_j
$$
admits a unique positive solution  $q_1,\cdots,q_n$.  The estimation of the $N\times N$ robust scatter matrix is then reduced to determining the solutions of a $n$ system of equations. 
The second advantage of this new rewriting is that it can provide, based on some heuristics, interesting insights about the asymptotic behaviour of $\hat{C}_N$. In effect, it is not difficult to understand that $y_i$ is weakly dependent on $\hat{C}_{(i)}$, since $\hat{C}_{(i)}$ depends on $y_i$ only through the terms $\frac{1}{N}y_j^*\hat{C}_N^{-1}y_j, j\neq i$. Standard results from random matrix theory will thus lead to $q_i=\frac{1}{N}y_i^*\hat{C}_{(i)}^{-1}y_i \sim \frac{1}{N}\tr (B_N+\tau_i I_N)\hat{C}_{(i)}^{-1}$, which tends to imply that $q_i$ scales with $\tau_i$. Assume that  $q_i, i=1,\cdots,n$ can be approximated by $\delta_i$ where $\delta_i$ does not depend on the random vector $w_i$. Then, because of rank-1 perturbation arguments leading to replace $\hat{C}_{(i)}^{-1}$ with $\hat{C}_N^{-1}$, we have:
$$
\frac{q_i}{\frac{1}{N}\tr (B_N+\tau_i I_N)\hat{C}_{(i)}^{-1}} \sim \frac{\delta_i}{\frac{1}{N}\tr (B_N+\tau_i I_N) \hat{C}_N^{-1}}.
$$
On the other hand, from the asymptotic equivalence between $q_i$ and $\delta_i$, we expect $\hat{C}_N$ to be asymptotically equivalent to $\frac{1}{n}\sum_{i=1}^n v(\delta_i) y_iy_i^*$. As we will see later, without inducing a major error, one can assume that $y_i$ are Gaussian. The asymptotic behaviour of  $\frac{1}{n}\sum_{i=1}^n v(\delta_i) y_iy_i^*$  can be thus studied using results from \cite{wagner}.  If Theorem 1 in \cite{wagner} is applicable, then $\delta_i$ should satisfy:
\begin{equation}
\begin{aligned}
1&\sim\frac{1}{N}\tr \frac{(B_N+\tau_i I_N)}{\delta_i}\hat{C}_N^{-1} \\
& \sim \frac{1}{N}\tr \frac{(B_N+\tau_i I_N)}{\delta_i}\left(\frac{1}{n}\sum_{j=1}^n \frac{v(\delta_j)(B_N+\tau_j I_N)}{1+e_j}\right)^{-1} \label{eq:approximate},
\end{aligned}
\end{equation}
where $e_1,\cdots,e_n$ are the fixed point solutions to the following system of equations:
$$
e_i=\frac{v(\delta_i)}{n}\tr (B_N+\tau_i I_N)\left(\frac{1}{n}\sum_{j=1}^n \frac{v(\delta_j)(B_N+\tau_j I_N)}{1+e_j}\right)^{-1}.
$$
Multiplying both sides of \eqref{eq:approximate}, we thus get:
$$
e_i\sim c_N\delta_i v(\delta_i)=c_N \psi(\delta_i).
$$
Plugging the above equations into \eqref{eq:approximate}, we obtain that $\delta_1,\cdots,\delta_n$ are solutions to the following system of equations:
$$
\delta_i=\frac{1}{N}\tr (B_N+\tau_i I_N)\left(\frac{1}{n}\sum_{j=1}^n \frac{v(\delta_j)(B_N+\tau_j I_N)}{1+c_N\psi_N(\delta_j)}\right), \hspace{0.3cm} i=1,\cdots,n.
$$
\subsection{Rigorous Proofs}
\label{sec:rigorous}
The main differences of our work with respect to the one in \cite{couillet-pascal-2013} lies in the considered data model. While \cite{couillet-pascal-2013} assumes purely noise observations drawn from elliptical distributions, we consider in the present work, sequence of time observations that are given by the sum of a heavy-tailed noise  and  a Gaussian distributed vector  modeling the "signal" part of the observations. In practice, the estimation of the covariance matrix of the available observations can help infer precious information on the  signal of interest. 
 From a theoretical standpoint, if the useful data live in a low-dimensional space,   the same approach  considered in \cite{couillet-pascal-2013} can be pursued with only minor changes.  Although less popular, high rank data models, occurring when $K$ scales with $N$,  are more attractive for several applications of array processing concerning distributed source localization \cite{valaee}. They are also more difficult to handle, since the use of  the approach of \cite{couillet-pascal-2013} poses many technical difficulties, when $B_N$ is allowed to be of high rank . This can be easily seen by noticing that our heuristic computations involve solving a system of $n$ equations while  those of \cite{couillet-pascal-2013} requires only solving the fixed point of a single equation.
One can easily convince oneself that in the context of interest, it is much more difficult to get insights into the behaviour of the $n$ solutions of the underlying system. 
%A large effort is then required in order to adapt the use of techniques of \cite{couillet-pascal-2013} to our particular context.
Before delving into the core of the proof, we need first to introduce in the sequel some preliminary results that will help adapt the techniques of \cite{couillet-pascal-2013} to our particular context.

\subsubsection{Preliminary Results}

\begin{paragraph}{{\bf Function $u$ and Related Functions}}
The robust-scatter estimator is parametrized by function $u$, which significantly impacts its performance. This intuition is further confirmed by  theoretical analysis, showing that a number of sequence of functions in relation with $u$ naturally arise. This section aims at presenting the list of these functions along with some of their most important properties. We first summarize in the following table some of the results that has been established in \cite{couillet-pascal-2013}.
\begin{table}[ht]
\caption{Properties of Functions}
\begin{center}
\begin{tabular}{|c|c|}
\hline
Functions & Properties \\
    \hline
    \multirow{3}{*}{$u(x)$}& Non-increasing\\
    &Positive \\
    & Continuous \\
    \hline
\multirow{4}{*}{$\phi(x)\triangleq xu(x)$}& Increasing\\
    &Positive \\
    & Continuous \\
& Bounded with $\phi_{\infty} <c_{+}^{-1}$.\\
\hline
\end{tabular}\hfill
\begin{tabular}{|c|c|}
\hline
Sequence of Functions & Properties \\
    \hline
   \multirow{4}{*}{$g_N(x)\triangleq \frac{x}{1-c_N\phi(x)}$}& Increasing\\
   &Positive \\
    & Continuous \\
& Unbounded \\
    \hline
\multirow{3}{*}{$v_N(x)\triangleq u\circ g_N^{-1}(x)$}& Non-Increasing\\
    &Positive \\
    & Continuous \\
\hline
\multirow{3}{*}{$\psi_N(x)\triangleq xv(x)$}& Increasing\\
    &Positive \\
    & Continuous \\
& Bounded with $\psi_{\infty}=\frac{\phi_{\infty}}{1-c_N\phi_{\infty}}$\\
\hline
\end{tabular}
\end{center}
\label{tab:functions}
\end{table}

In addition to the aforementioned properties, we need to prove the following results, which will be used in our proofs. 
\begin{lemma}
Let $u(\cdot)$ and $\phi(\cdot)$ be two functions satisfying assumption \ref{ass:u}. Then, we have, for all $x,y\geq0$,
$$
\frac{\phi(x)-\phi(y)}{x-y} \leq u(x)\leq u(0).
$$
In other words, $\phi(\cdot)$ is Lipschitz with Lipschitz constant $u(0)$.
\label{lemma:lipschitz}
\end{lemma}
\begin{proof}
We have:
\begin{align*}
\frac{\phi(x)-\phi(y)}{x-y}&=\frac{xu(x)-yu(y)}{x-y} \\
&=\frac{xu(x)-yu(x)+yu(x)-yu(y)}{x-y}.
\end{align*}
Since $u(\cdot)$ is non-increasing: 
$$
y\frac{u(x)-u(y)}{x-y} \leq 0.
$$
Therefore,
$$
\frac{\phi(x)-\phi(y)}{x-y}\leq \frac{xu(x)-y u(x)}{x-y}=u(x)\leq u(0).
$$
\end{proof}
\begin{lemma}
Let $u$ and $\phi$ be two functions satisfying assumption \ref{ass:u}. Then, we have, for all $x\geq 0$:
$$
\phi(x)\leq u(0)x.
$$
Moreover, for all $x\in[0,m]$,
$$
u(m)x \leq \phi(x) \leq u(0)x.
$$
\label{lemma:boundedness}
\end{lemma}
\begin{proof}
The first statement follows from the previous lemma by setting $y=0$. 
%The proof follows since $u(x)$ is non-increasing, and thus $\phi(x)=xu(x) \leq x u(0)$. 
To prove the second, notice that when $x\leq m$, $u(x) \geq u(m)$, thereby showing that $\phi(x) \geq xu(m)$ whenever $x\in\left[0,m\right]$.
\end{proof}
\begin{remark}
As it has been proven in \cite{couillet-pascal-2013}, functions $x\mapsto \psi(x)$ and $x\mapsto v(x)$ share respectively the same properties as $x\mapsto \phi(x)$ and $x\mapsto u(x)$. As a consequence, we can prove that $x\mapsto \psi(x)$ is Lipschitz with constant lipschitz $v(0)=u(0)$. The constant Lipschitz being independent on $n$, we conclude that $\left(\psi_N\right)$ form an equicontinuous family of functions and as such converge uniformly on $[0,\infty)$. 
Moreover,
$$
\psi(x)\leq v(0)x
$$
and $\psi(x)\geq v(m)x$ whenever $x\in\left[0,m\right]$.
\label{remark:psi}
\end{remark}
\begin{lemma}
\label{lemma:g_N}
Let $g_N(\cdot):x\mapsto \frac{x}{1-c_N \phi(x)}$. Denote by $g_N^{-1}(\cdot)$ the inverse function corresponding to $g_N$. Then, for all $y\geq z\geq 0$, we have:
$$
g_N^{-1}(y)-g_N^{-1}(z)\leq (y-z)(1-c_N \phi(g_N^{-1}(y)))
$$
In particular, $g_N^{-1}$ is lipschitz on $\left[0,\infty\right)$ with constant lipschitz $1$. %whereas  $g_N$ is Lipschitz on every compact of $\left[0,+\infty\right)$. 
Besides, functions $\left(x\mapsto \frac{\psi_N(x)}{1+c_N\psi_N(x)}\right)_{N=1}^{\infty}$ are Lipschitz and converge uniformly on $\left[0,\infty\right)$.
\end{lemma}
\begin{proof}
Let $y\geq z\geq 0$. From the relation $g_N^{-1}(y)=y-c_N\phi(g_N^{-1}(y))$, we have:
\begin{align*}
g_N^{-1}(y)-g_N^{-1}(z)&=y-z+c_Nz\phi(g_N^{-1}(z))-c_Ny\phi(g_N^{-1}(y))\\
&=y-z +c_Nz\left(\phi(g_N^{-1}(z))-\phi(g_N^{-1}(y))\right)+c_N(z-y)\phi\left(g_N^{-1}(y)\right).
%&\stackrel{(a)}\leq (y-z)-c_N\phi\left(g_N^{-1}(z)\right)\left(g_N^{-1}(y)-g_N^{-1}(z)\right)+c_N(z-y)\phi\left(g_N^{-1}(y)\right)
\end{align*}
Since $g_N^{-1}$ is increasing,  $\phi(g_N^{-1}(z))-\phi(g_N^{-1}(y))\leq 0$. Hence,
%Let $g_{N}^{-1}$ be the inverse of $g_N$. Then, for any $y\geq 0$,
%$$
%g_N^{-1}(y)=(1-c_N \phi(g_N^{-1}(y)))y.
%$$
%Consider $y_1$ and $y_2$ on $[0,\infty)$ with $y_1\leq y_2$. Then,
%\begin{align*}
%g_N^{-1}(y_2)-g_N^{-1}(y_1)&=y_2-y_1 -c_N y_2 \phi\left(g_N^{-1}(y_2)\right)+c_N y_1 \phi\left(g_N^{-1}(y_1)\right) \\
%&\stackrel{(a)}{\leq} y_2-y_1
%\end{align*}
%where $(a)$ follows from lemma \ref{lemma:lipschitz}. %since $g_N^{-1}$ and $y\mapsto yg_N^{-1}(y)$ are increasing. 
%Hence,
$$
g_N^{-1}(y)-g_N^{-1}(z)\leq (y-z)(1-c_N\phi(g_N^{-1}(y))) \leq (y-z).
$$
Finally, after simple calculations, we can prove that:
$$
\frac{\psi_N(x)}{1+c_N\psi_N(x)} =\phi\circ g_N^{-1}(x).
$$
Therefore, $x\mapsto
\frac{\psi_N(x)}{1+c_N\psi_N(x)}$ is Lipschitz with constant lipschitz equal to $u(0)$. This constant being independent on $n$, the sequence of functions $\frac{\psi_N(x)}{1+c_N\psi_N(x)}$ converge uniformly on $\left[0,\infty\right)$.
%Similarly, let  $a>0$ and $0\leq x_1\leq x_2\leq a$. We have:
%\begin{align*}
%g_N(x_2)-g_N(x_1) &=\frac{x_2}{1-c_N\phi(x_2)}-\frac{x_1}{1-c_N\phi(x_1)}\\
%&=\frac{x_2-x_1 +c_N x_2(\phi(x_2)-\phi(x_1))+c_N\phi(x_2)(x_1-x_2)}{(1-c_N\phi(x_2))(1-c_N\phi(x_1))}\\
%&\leq \frac{(x_2-x_1)(1+c_{+}au(0))}{(1-c_{+}\phi_{\infty})^2}
%\end{align*}
% Finally, to prove that $\psi_N$ is Lipschitz, first note that $\phi\circ g_N^{-1}(x) \in \left(0,\phi_\infty\right)$. Since $g_N$ is Lipschitz on every compact, we have for any $0\leq x_1\leq x_2$,
%\begin{align*}
%\left|\psi_N(x_1)-\psi_N(x_2)\right|& \leq \frac{\left|\phi\circ g_N^{-1}(x_1)-\phi\circ g_N^{-1}(x_2)\right|(1+c_{+}\phi_{\infty}u(0))}{(1-c_{+}\phi_{\infty})^2} \\
%&\leq \frac{u(0)(1+c_{+}\phi_{\infty}u(0))\left|x_1-x_2\right|}{(1-c_{+}\phi_{\infty})^2}.
%\end{align*} 
%Functions $\left(\psi_N\right)_{N=1}^{\infty}$ are thus lipschitz with a constant lipschitz independent of $n$. It forms an equicontinuous family of functions and as such converges uniformly on $[0,\infty)$.
\end{proof}
\end{paragraph}
\begin{paragraph}{\bf Useful results}
As previously stated, the difficulty of studying the robust-scatter estimator lies in the control of the asymptotic behaviour of $q_i$ and $\delta_i$. The proof of Theorem \ref{th:uniqueness} and Theorem \ref{th:asymptotic}  will require us to show that $q_i$ and $\delta_i$ scale with $\tau_i$ and to control quadratic forms involving matrix $\frac{1}{n}\sum_{i=1}^n f(\tau_i)y_iy_i^*$ where $f$ is a certain functional. To this end, we develop in this section two key results that will underlie the 
proof of the main theorems. 

\begin{proposition}
  \label{prop:inequality}
  Let $(B_N)$ be a sequence of $N\times N$ hermitian positive matrices satisfying Assumption \ref{ass:statistical}-iii).
  In addition, let $\tau_i,i=1,\cdots,n$ be positive random variables satisfying assumption \ref{ass:statistical}-ii). Consider  $(f_N)$  a sequence of  piece-wise continuous positive  bounded functions defined on $[0,\infty)$ that has at least one subsequence converging uniformly. We assume that  functions $t\mapsto f_N(t)$ satisfy the following additional properties:
\begin{itemize}
\item Function $t\mapsto f_N(t)$ grows at most linearly, i.e there exists $\alpha,\beta>0$ such that:
   \begin{align*}
  \sup_{N} f_N(t) \leq \alpha \hspace{0.5cm} &  \forall t\geq 0, \\
\sup_N f_N(t)\leq \beta t  \hspace{0.5cm}&\forall t\geq 0.
   \end{align*}
   \item    $\int f_N(t)\nu(dt)=1$ 
\item  $\liminf_N \inf_{t\in \left[m,+\infty\right)} f_N(t)>0$.
\end{itemize}
  Then the following equation in $x$:
  \begin{equation}
 \int \frac{F^{B_N}(dy)}{\int \frac{y+t}{t+x}f_N(t)\nu(dt)} =1 
  \label{eq:eta}
  \end{equation}
  admits a unique positive solution which we denote by $\eta_N$.  Then, there exists a sequence $(r_N^{-})$ with $\lim\inf r_N^{-} > 0$ such that:
  $$
r_N^{-} \frac{1}{N}\tr B_N \leq \eta_N \leq \frac{1}{N}\tr B_N.
$$
%with $\lim\inf r_N^{-} >0$, %$r_N^{-}=\frac{m^2a_{m,N}^3}{\beta\left(1+\frac{\|B_N\|}{m}\right)^3(a_{m,N}\|B_N\|+1)} $,
Moreover, we have:
%If $\eta>0$ or $t\mapsto f(t)$ is zero in a neighborhood of zero, then, we have:
    \begin{equation}
c_N-|\epsilon_{n,j}|\leq  \frac{1}{N}\tr \frac{B_N+\tau_j I}{\tau_j+\eta_N} \left(\frac{1}{N}\sum_{i=1}^n {f_N(\tau_i)}\frac{B_N+\tau_i I}{\tau_i+\eta_N}\right)^{-1} \leq c_N+|\epsilon_{n,j}|,%{\frac{1}{n}\sum_{i=1}^{n}\frac{\tau_i}{\tau_i+\eta}} +\epsilon_n
\label{eq:fundamental}
  \end{equation}
  where $\max_{1\leq j\leq n}\left|\epsilon_{n,j}\right|$ converges almost surely to zero.
  \end{proposition}
\begin{proof}
 We start by showing that \eqref{eq:eta} admits a unique solution $\eta_N$. % For that first decompose the limit distribution $G$ as:
%  $G(y)=r\delta_0(y) +(1-r)H(y)$ where $H$ is defined within the interval $[a,b]$. We thus get:
 % $$
 % \int \frac{F^{B}(dy)}{\int \frac{y+t}{t+x}{f(t)}\nu(dt)} = \frac{r}{\int \frac{t}{t+x}{f(t)}\nu(dt)} +(1-r)\int_{a}^b \frac{H(y)dy}{\int \frac{y+t}{t+x}{f(t)}\nu(dt)}
 % $$
  It is clear that function $h_N: x\mapsto \int \frac{F^{B_N}(dy)}{\int \frac{y+t}{t+x}f_N(t)\nu(dt)} $ is  increasing and  continuous on $(0,+\infty)$ with the limit at $x\to 0^{+}$ less than $1$, while the limit when $x\to +\infty$ is $+\infty$. Therefore, there exists a unique $\eta_N$ that satisfies \eqref{eq:eta}. It is easy to check that $\eta_N$ is less than the maximum eigenvalue of $B_N$. Therefore, we can restrict the domain of $h_N$ to the set $\left[0,\|B_N\|\right]$. 
Since $y\mapsto \frac{1}{\int \frac{y+t}{t+x}f_N(t)\nu(dt)}$ is convex, applying the Jensen inequality, we obtain:
$$
\int \frac{F^{B_N}(dy)}{\int \frac{y+t}{t+x}f_N(t)\nu(dt)} \geq \frac{1}{\int \frac{\frac{1}{N}\tr B_N +t}{t+x}f_N(t)\nu(dt)}.
$$
Setting  $x=\eta_N$, the above inequality becomes:
\begin{equation}
1\geq \frac{1}{\int \frac{\frac{1}{N}\tr B_N +t}{t+\eta_N}f_N(t)\nu(dt)}.
\label{eq:true}
\end{equation}
Therefore, $\eta_N\leq \frac{1}{N}\tr B_N$ because otherwise, \eqref{eq:true} would not hold.

The proof of the lower-bound inequality is more delicate. Let $m$ be as in Assumption \ref{ass:measure}-ii) and denote by $h_{m,N}$ the following map:
$$
h_{m,N}:[0,\|B_N\|]\to \mathbb{R}^{+}, x\mapsto \int \frac{1}{y\int_m^{\infty} \frac{f_N(t)}{t+x}\nu(dt)+\int_0^{\infty}\frac{tf_N(t)}{t+x}\nu(dt)}F^{B_N}(dy).
$$ 
Functions $h_N$ and $h_{m,N}$ are both increasing, while $h_{m,N}(x)\geq h_N(x) \ \ \forall x\in[0,\|B_N\|]$. Furthermore, we can easily check that:
$$
\lim_{x\to 0^{+}}h_{m,N}(x)<1 \hspace{0.1cm}\textnormal{and} \hspace{0.1cm} \lim_{x\to +\infty} h_{m,N}(x)=+\infty.
$$
Therefore, there exists  $\eta_{N,m}$ solution in $x$ to the equation $h_{m,N}(x)=1$. Moreover, we have $h_N(\eta_{N,m})\leq 1$, and thus $\eta_N \geq \eta_{N,m}$. On the other hand, $h_{m,N}$ is differentiable with derivative $h_{m,N}^{'}(x)$ given by:
$$
h_{m,N}^{'}(x)=\int \frac{y\int_m^{+\infty}\frac{f_N(t)}{(t+x)^2}\nu(dt)+\int_0^{+\infty}\frac{tf_N(t)}{(t+x)^2}\nu(dt)}{\left(y\int_m^{+\infty} \frac{f_N(t)}{t+x}\nu(dt)+\int_0^{+\infty}\frac{tf_N(t)}{t+x}\nu(dt)\right)^2}F^{B_N}(dy).
$$
Let $a_{m,N}=\int_m^{+\infty} \frac{f_N(t)}{t}\nu(dt)$.
Hence, if $0\leq x\leq \|B_N\|$,
\begin{align*}
h_{m,N}^{'}(x) & \leq \int \frac{y\int_m^{+\infty}\frac{f_N(t)}{(t+x)^2}\nu(dt)+\int_0^{+\infty}\frac{tf_N(t)}{(t+x)^2}\nu(dt)}{\left(\int_m^{+\infty}\frac{tf_N(t)}{t+x}\nu(dt)\right)^2}F^{B_N}(dy)\\
 & \leq \int \frac{y\int_m^{+\infty}\frac{f_N(t)}{m t}\nu(dt)+\int_0^{+\infty}\frac{f_N(t)}{t}\nu(dt)}{\left(\frac{m^2}{\|B_N\|+m} a_{m,N}\right)^2}F^{B_N}(dy)\\
&\leq \frac{\beta\left(\frac{\|B_N\|}{m}+1\right)^3}{m^2a_{m,N}^2}.
\end{align*}
The mean value theorem implies that:
\begin{equation}
\frac{1-h_{m,N}(0)}{\eta_{N,m}} \leq \frac{\beta\left(\frac{\|B_N\|}{m}+1\right)^3}{m^2a_{m,N}^2},
\label{eq:mean_value}
\end{equation}
where
$$
h_{m,N}(0)=\int \frac{1}{y\int_m^{+\infty} \frac{f_N(t)}{t}+1}F^{B_N}(dy) =\frac{1}{N}\tr \left(a_{m,N} B_N+ I_N\right)^{-1}.
$$
As a consequence,
\begin{align}
1-h_{m,N}(0)& =\frac{1}{N}\tr I_N-\frac{1}{N}\tr \left(a_{m,N}B_N+I_N\right)^{-1} \nonumber\\
&=\frac{a_{m,N}}{N}\tr B_N\left(a_{m,N}B_N + I_N\right)^{-1} \geq \frac{a_{m,N}\frac{1}{N}\tr B_N}{a_{m,N}\|B_N\|+1}. \label{eq:mean_value_bis}
\end{align}
Combining \eqref{eq:mean_value} and \eqref{eq:mean_value_bis}, we therefore get:
$$
\eta_N \geq \frac{m^2a_{m,N}^3 \frac{1}{N}\tr B_N}{(a_{m,N}\|B_N\|+1)\left(1+\frac{\|B_N\|}{m}\right)^3\beta}.
$$
Note that it is easy to prove that  $r_N\triangleq \frac{m^2a_{m,N}^3 }{(a_{m,N}\|B_N\|+1)\left(1+\frac{\|B_N\|}{m}\right)^3\beta} \leq 1$, since $ma_{m,N}\leq 1$ and $\frac{a_{m,N}}{\beta}<1$.  Moreover, $\lim\inf r_N >0$ as $r_N \geq \frac{m^2 a_{m,N}^3}{\beta\left(\beta \|B_N\|+1\right)^4}$ and $\lim\inf a_{m,N} \geq \liminf_N \inf_{x\in \left[m,\infty\right)} f_N(x) \int_{m}^{+\infty} \frac{1}{t}\nu(dt) >0$.  % For that, first note that for any positive $x$, function  $y\mapsto \frac{H(y)}{\int \frac{y+t}{t+x} \nu(dt)}$ is integrable. Moreover, for any $y$, the map $x\mapsto \frac{H(y)}{\int \frac{y+t}{t+x}\nu(dt)}$ is continuous in $x$.  

We will now proceed proving the inequalities in  \eqref{eq:fundamental}.
  Let $\lambda_1^{N}\leq \cdots\leq \lambda_N^{N}$ be the eigenvalues of $B_N$.  
   We have:
  \begin{align}
&	  \frac{1}{N}\tr \frac{B_N+\tau_j I}{\tau_j+\eta_N} \left(\frac{1}{N}\sum_{i=1}^n f_N(\tau_i)\frac{B_N+\tau_i I}{\tau_i+\eta_N}\right)^{-1} =\frac{1}{N}\sum_{k=1}^N \frac{\lambda_k^{N}+\tau_j}{(\tau_j+\eta_N)\frac{1}{N}\sum_{i=1}^nf_N(\tau_i) \frac{\lambda_k^{N}+\tau_i}{\tau_i+\eta_N}} \nonumber \\			&=\frac{c_N}{\tau_j+\eta_N}\int \frac{y F^{B_N}(dy)}{\frac{1}{n}\sum_{i=1}^n f_N(\tau_i)\frac{y+\tau_i}{\tau_i+\eta_N}} +\frac{c_N \tau_j}{(\tau_j+\eta_N)} \int \frac{F^{B_N}(dy)}{\frac{1}{n}\sum_{i=1}^n f_N(\tau_i)\frac{y+\tau_i}{\tau_i+\eta_N}} \nonumber \\													& =\frac{c_N}{\tau_j+\eta_N} \int \frac{ydF^{B_N}}{\int \frac{y+t}{t+\eta_N}f_N(t)\nu(dt)} +\frac{c_N\tau_j}{\tau_j+\eta_N} \int \frac{dF^{B_N}}{\int \frac{y+t}{t+\eta_N}f_N(t)\nu(dt)}+\frac{c_N}{\tau_j+\eta_N} \epsilon_{n,1} +\frac{c_N\tau_j}{\tau_j+\eta_N}\epsilon_{n,2},  \label{eq:develop}  %c_N \int \frac{\tau_j dF^{B^{N}}}{(\tau_j+\eta)\frac{1}{n}\sum_{i=1}^n \frac{y+\tau_i}{\tau_i+\eta}}\\
     %   &=c_Nr\int \frac{(y+\tau_j) dF^{B^{[a,b]}}}{(\tau_j+\eta)\frac{1}{n}\sum_{i=1}^n \frac{y+\tau_i}{\tau_i+\eta}} +c_N(1-r)\frac{\tau_j}{(\tau_j+\eta)\frac{1}{n}\sum_{i=1}^n \frac{\tau_i}{\tau_i+\eta}}+%+c_N \int \frac{\tau_j dF^{B}}{(\tau_j+\eta)\frac{1}{n}\sum_{i=1}^n \frac{y+\tau_i}{\tau_i+\eta}} +\epsilon_{n,j,1} +\epsilon_{n,j,2}
  \end{align}
  where
\begin{align}
\epsilon_{n,1}&=\int\frac{ y \left(\int \frac{y+t}{t+\eta_N}f_N(t)\nu(dt)-\frac{1}{n}\sum_{i=1}^n f_N(\tau_i)\frac{y+\tau_i}{\tau_i+\eta_N}\right)}{\frac{1}{n}\sum_{i=1}^n f_N(\tau_i)\frac{y+\tau_i}{\tau_i+\eta_N}\int \frac{y+t}{t+\eta_N}f_N(t)\nu(dt)}F^{B_N}(dy)\\
\epsilon_{n,2}&=\int\frac{ \int \frac{y+t}{t+\eta_N}f_N(t)\nu(dt)-\frac{1}{n}\sum_{i=1}^n f_N(\tau_i)\frac{y+\tau_i}{\tau_i+\eta_N}}{\frac{1}{n}\sum_{i=1}^n f_N(\tau_i)\frac{y+\tau_i}{\tau_i+\eta_N}\int \frac{y+t}{t+\eta_N}f_N(t)\nu(dt)}F^{B_N}(dy).
\end{align}
%  \begin{align}
%	  \epsilon_{n,1}&=\int \frac{y(dF^{B^N}-dF^{B})}{\frac{1}{n}\sum_{i=1}^n f(\tau_i)\frac{y+\tau_i}{\tau_i+\eta}} +\int \frac{y\left(\int f(\tau_i) \frac{y+t}{t+\eta} \nu(dt)\right)-\frac{1}{n}\sum_{i=1}^n f(\tau_i)\frac{y+\tau_i}{\tau_i+\eta}}{\left(\frac{1}{n}\sum_{i=1}^n f(\tau_i)\frac{y+\tau_i}{\tau_i+\eta}\right)\int f(t)\frac{y+t}{t+\eta}\nu(dt)} F^{B}(dy)\\
%	  \epsilon_{n,2}&=\int \frac{dF^{B^N}-dF^B}{\frac{1}{n}\sum_{i=1}^n f(\tau_i) \frac{y+\tau_i}{\tau_i+\eta}} +\int \frac{\int f(t) \frac{y+t}{t+\eta}\nu(dt)-\frac{1}{n}\sum_{i=1}^n f(\tau_i)\frac{y+\tau_i}{\tau_i+\eta}}{\left(\frac{1}{n}\sum_{i=1}^n f(\tau_i)\frac{y+\tau_i}{\tau_i+\eta}\right)\int f(t)\frac{y+t}{t+\eta}\nu(dt)} F^B(dy)
%  \end{align}
  As $\eta_N$ is the unique solution of \eqref{eq:eta}, $\int \frac{yF^{B_N}(dy)}{\int \frac{y+t}{t+\eta_N}f_N(t)\nu(dt)}$ can be further simplified as:
\begin{align}
\int \frac{yF^{B_N}(dy)}{\int \frac{y+t}{t+\eta_N}f_N(t)\nu(dt)} & = \int \frac{y\int\frac{f_N(t)}{t+\eta_N}\nu(dt)+ \int\frac{tf_N(t)}{t+\eta_N}\nu(dt)-\int \frac{tf_N(t)}{t+\eta_N}\nu(dt)}{\int \frac{f_N(t)}{t+\eta_N}\nu(dt)\int \frac{y+t}{t+\eta_N}f_N(t)\nu(dt)}F^{B_N}(dy) \nonumber \\
																																							    &\stackrel{(a)}{=}\frac{\int f_N(t)\nu(dt)}{\int \frac{f_N(t)}{t+\eta_N}\nu(dt)} -\frac{\int \frac{tf_N(t)}{t+\eta_N} \nu(dt)}{\int \frac{f_N(t)}{t+\eta_N}\nu(dt)}  \nonumber\\
  &=\eta_N, \label{eq:eta_bis}
%  %{\int \frac{1}{t+\eta}\nu(dt)\int \frac{y+t}{t+\eta}\nu(dt)}
 \end{align}
where $(a)$ follows due to the fact that $\int\frac{1}{\int \frac{y+t}{t+\eta_N}f_N(t)\nu(dt)}F^{B_N}(dy)=1$.
  Substituting \eqref{eq:eta_bis} into \eqref{eq:develop}, we get:
  $$
  \frac{1}{N}\tr \frac{B_N+\tau_j I}{\tau_j+\eta_N} \left(\frac{1}{N}\sum_{i=1}^n f(\tau_i)\frac{B_N+\tau_i I}{\tau_i+\eta_N}\right)^{-1} ={c_N} +\frac{c_N}{\tau_j+\eta_N} \epsilon_{n,1} +\frac{c_N\tau_j}{\tau_j+\eta_N} \epsilon_{n,2}.
     $$
  Therefore, the result immediately follows once we prove that $\max_{1\leq j\leq n}\frac{1}{\tau_j+\eta_N}\left|\epsilon_{n,1}\right|$ and $\left|\epsilon_{n,2}\right|$ converge almost surely to zero. We will only control $\frac{1}{\tau_j+\eta_N}\epsilon_{n,1}$. The control of $\epsilon_{n,2}$ can be obtained using the same arguments.
We have:
\begin{align*}
\frac{\left|\epsilon_{n,1}\right|}{\tau_j+\eta_N}&\leq \sup_y\left|\int \frac{y+t}{t+\eta_N}f_N(t)\nu(dt)-\frac{1}{n}\sum_{i=1}^n f_N(\tau_i)\frac{y+\tau_i}{\tau_i+\eta_N}\right|\\
&\times \frac{1}{\eta_N}\int \frac{yF^{B_N}(dy)}{\left(\frac{1}{n}\sum_{i=1}^n f_N(\tau_i)\frac{y+\tau_i}{\tau_i+\eta_N}\right)\int \frac{y+t}{t+\eta_N}f_N(t)\nu(dt)}\\
&\leq \left(\lambda_N^N \left|\int \frac{f_N(t)}{t+\eta_N}\nu(dt)-\frac{1}{n}\sum_{i=1}^n \frac{f_N(\tau_i)}{\tau_i+\eta_N}\right| +\left|\int \frac{tf_N(t)}{t+\eta_N}\nu(dt)-\frac{1}{n}\sum_{i=1}^n\frac{f_N(\tau_i)\tau_i}{\tau_i+\eta_N}\right| \right) \\
&\times \frac{1}{\eta_N}\int \frac{y F^{B_N}(dy)}{\left(\frac{1}{n}\sum_{i=1}^n \frac{f_N(\tau_i)(y+\tau_i)}{\tau_i+\eta_N}\right)\int\frac{y+t}{t+\eta_N}f_N(t)\nu(dt) }.
\end{align*}

Since $\frac{f_N(t)}{t+\eta_N} \leq \frac{\beta t}{t+\eta_N} \leq \beta$ and $\frac{tf_N(t)}{t+\eta_N} \leq \frac{t\alpha}{t+\eta_N} \leq \alpha$, sequences  $\int \frac{f_N(t)}{t+\eta_N}\nu_N(dt)-\int \frac{f(t)}{t+\eta_N}\nu(dt)$ and $\int \frac{tf_N(t)}{t+\eta_N}\nu_N(dt)-\int \frac{tf_N(t)}{t+\eta_N}\nu(dt)$ are bounded. One can extract a subsequence $(n)$ such that:
$
\int \frac{f_N(t)}{t+\eta_N}\nu_N(dt) - \int \frac{f_N(t)}{t+\eta_N}\nu(dt)
$ and  $\int \frac{tf_N(t)}{t+\eta_N}\nu_N(dt)-\int \frac{tf_N(t)}{t+\eta_N}\nu(dt)$ converge.
Over this subsequence, $t\mapsto \frac{tf_N(t)}{t+\eta_N}$ converge uniformly to $t\mapsto \frac{tf^*(t)}{t+\eta^*}$ and as such:
$$
\int \frac{tf_N(t)}{t+\eta_N}\nu_N(dt) -\int \frac{tf^*(t)}{t+\eta^*}\nu(dt)\asto 0.
$$
Moreover, since $\lim\inf\eta_N\geq \lim\inf r_N^{-}\frac{1}{N}\tr B_N >0$, $\eta^*\neq 0$. The sequence  of functions $t\mapsto \frac{f_N(t)}{t+\eta_N}$ converge also uniformly to $t\mapsto \frac{f^*(t)}{t+\eta^*}$, thereby yielding:
%To handle the term $\int \frac{1}{t+\eta_N}f_N(t)\nu_N(dt)$, two cases need to be considered. If $\eta^*=0$, then we can show the uniform convergence of  $\frac{f_N(t)}{t+\eta_N}$ to $\frac{f^*(t)}{t+\eta^*}$, thereby establishing that: 
$$
\int \frac{f_N(t)}{t+\eta_N}\nu_N(dt) - \int \frac{f^*(t)}{t+\eta^*}\nu(dt)\asto 0.
$$
It remains thus to check that $ \frac{1}{\eta_N}\int \frac{y dF^{B_N}}{\left(\frac{1}{n}\sum_{i=1}^n \frac{f_N(\tau_i)(y+\tau_i)}{\tau_i+\eta_N}\right)\int\frac{y+t}{t+\eta_N}f_N(t)\nu(dt) }$ is almost surely bounded. For that, first note that since $
\eta_N =\int \frac{ydF^{B_N}}{\int \frac{y+t}{t+\eta_N}f_N(t)\nu(dt)}.
$, we have:

\begin{align*}
\frac{1}{\eta_N}\int \frac{y}{\frac{1}{n}\sum_{i=1}^n f_N(\tau_i)\frac{y+\tau_i}{\tau_i+\eta_N}\int \frac{y+t}{t+\eta_N}f_N(t)\nu(dt)}dF^{B_N} \leq \frac{1}{\frac{1}{n}\sum_{\substack{i=1 \\ \tau_i \geq m}}^n \frac{f_N(\tau_i)m}{m+\eta_N}} 
\end{align*} 
As $\eta_N \leq \lambda_N^{N}$ and $\lim\inf_N \inf_{t\in\left[m,+\infty\right)} f_N(t)>0$, $\frac{1}{n}\sum_{\substack{i=1 \\ \tau_i \geq m}}^n \frac{f_N(\tau_i)m}{m+\eta_N}$ is almost surely bounded away from zero, thereby implying the desired result. The control of $\epsilon_{n,2}$ could be done using the same arguments.  
%This can be treated using the weak convergence of $F^{B^{N}}$ to $F^{B}$ and $\nu_n$
%to $\nu$.
%We will only treat the term $\epsilon_{n,1}$. The term $\epsilon_{n,2}$ can be controlled similarly. 
\end{proof}
The second ingredient that will be of extensive use in the proof of the theorem is provided by the following key-lemma.  
\begin{lemma}
\label{lemma:quadratic_alpha}
Let Assumption \ref{ass:regime}-\ref{ass:statistical}  hold true.
Let $(f_N)$ be a sequence functions satisfying the conditions of proposition \ref{prop:inequality}.  Denote by $\eta_N$  the unique solution in $x$ to the following equation:
$$
\int \frac{F^{B_N}}{\int_0^{\infty} \frac{y+t}{t+x}f_N(t)\nu(dt)} =1.
$$
 %Let $\alpha_1,\cdots,\alpha_n$ a set of positive scalars. Assume that  there exists positive scalars, $\beta_N^{-}=\mathcal{O}(\frac{1}{N}\tr B_N)$,  such that for all $1\leq i \leq n$, we have:
%$$
%\tau_i+\beta_N^{-}\leq \alpha_i \leq \tau_i+\frac{1}{N}\tr B_N $$
Consider $e_1,\cdots,e_n$  the unique solutions to the following system of equations:
\begin{equation}
e_k=\frac{f_N(\tau_k)}{n}\tr \frac{B_N+\tau_k I_N}{\tau_k+\eta_N}\left(\frac{1}{n}\sum_{i=1}^n \frac{f_N(\tau_i)(B_N+\tau_i I_N)}{(\tau_i+\eta_N)(1+e_i)}\right)^{-1}.
\label{eq:e_k}
\end{equation}
Then, the following statements hold true:
\begin{enumerate}[i)]
\item
$
\max_{1\leq j\leq n} \left|\frac{f_N(\tau_j)}{N(\tau_j+\eta_N)}y_j^{*}\left(\frac{1}{n}\sum_{i\neq j} \frac{f_N(\tau_i)}{\tau_i+\eta_N}y_iy_i^*\right)^{-1}y_j-c_N^{-1} e_j\right|\xrightarrow[a.s]{}0.
$
\item If $T_N=\left(\frac{1}{n}\sum_{i=1}^n \frac{f_N(\tau_i)(B_N+\tau_iI_N)}{(\tau_i+\eta_N)(1+e_i)}\right)^{-1}$.
Then, we have:
\begin{equation}
\max_{1\leq j\leq n}\left|\frac{1}{N(\tau_j+\eta_N)}y_j^*\left(\frac{1}{n}\sum_{i\neq j}\frac{f_N(\tau_i)}{\tau_i+\eta_N}y_iy_i^*\right)^{-1}y_j-\frac{1}{N}\tr \frac{\left(B_N+\tau_j I_N\right)}{\tau_j+\eta_N} T_N\right|\xrightarrow[]{a.s.}0.
\label{eq:T_2}
\end{equation}
\item If $\|f_N\|_{\infty} < \frac{1}{c_N}$, then, there exists $\epsilon_n\downarrow 0$ such that for $n$ large enough, a.s.
$$
\max_{1\leq k\leq n} e_k \leq \frac{c_N\|f_N\|_{\infty}}{1-\|f_N\|_{\infty}c_N} + \epsilon_n.
$$
\end{enumerate}
\end{lemma} 
\begin{proof}
The proof of the first two items  is based on  Lemma \ref{app:convergence_quadratic} and Lemma \ref{lemma:convergence_deterministic_equivalent} in Appendix \ref{app:technical}. For these Lemmas to be applicable, we need to check that $\lim\inf_N\min_{1\leq j\leq n}\lambda_1(\frac{1}{n}\sum_{i\neq j} \frac{f_N(\tau_i)}{\tau_i+\eta_N}y_iy_i^*)>0$. To this end, first note that:
$$
\min_{1\leq j\leq n}\lambda_1\left(\frac{1}{n}\sum_{i\neq j} \frac{f_N(\tau_i)}{\tau_i+\eta_N}y_iy_i^*\right)\geq \min_{1\leq j \leq n}\lambda_1\left(\frac{1}{n}\sum_{i\neq j,\tau_i\geq m} \frac{f_N(\tau_i)}{\tau_i+\eta_N}y_iy_i^*\right).
$$ 
The right-hand side of the above equality is almost surely bounded above zero since $f_N(\tau_i)\frac{(B_N +\tau_i I_N)}{\tau_i+\eta_N} \succeq \inf_{t\in\left[m,\infty\right)}f_N(t) \frac{m}{m+\eta_N}$ and $\eta_N$ is almost surely bounded by proposition \ref{prop:inequality}. We conclude thus by resorting to Lemma \ref{lemma:smallest_eigenvalue} in Appendix \ref{app:technical}.

 In order to prove the last statement, let $j_0$ be the index of the maximum element in $\left\{e_1,\cdots,e_n\right\}$. We therefore have:
\begin{align*}
e_{j,0}& \leq \|f_N\|_{\infty}(1+e_{j_0}) \frac{1}{n}\tr \frac{(B_N+\tau_{j_0}I_N)}{\tau_{j_0}+\eta_N}\left(\frac{1}{n}\sum_{i=1}^n \frac{f_N(\tau_i)(B_N+\tau_i I_N)}{\tau_i+\eta_N}\right)^{-1}\\
&\stackrel{(a)}{\leq}  \|f_N\|_{\infty}(1+e_{j_0})\left(c_N+\epsilon_n\right),
\end{align*}
where $(a)$ follows from proposition \ref{prop:inequality}. Besides, scalars $e_1,\cdots,e_n$ being the limits of almost surely bounded random quantities are bounded. Therefore,
$$
e_{j_0}\leq \frac{\|f_N\|_{\infty}c_N}{1-\|f_N\|_{\infty}c_N} +\epsilon_n^{'}
$$
where $\epsilon_n^{'}\downarrow 0$. 
\end{proof}
\end{paragraph}

\subsubsection{Proof of the Main Theorems}

With the above preliminary results at hand, we are now in position to provide the proofs of Theorem \ref{th:uniqueness} and Theorem \ref{th:asymptotic}.

\begin{paragraph}{{\bf Proof of Theorem \ref{th:uniqueness}: Asymptotic Existence of the Robust Scatter Estimator}}

Theorem \ref{th:uniqueness} establishes the existence of the robust scatter estimate for large $n$ and $N$. In particular, it implies that for each realization, there exists $n_0$ and $N_0$ large such that  for all $n$ and $N$ greater than $n_0$ and $N_0$, equation \eqref{eq:Z} admits a unique solution.
Although we believe that a stronger result showing the existence of the robust scatter estimate for well-behaved set of samples can be established using probably the same kind of techniques as in \cite{chitour-14}, we have chosen in this paper to show Theorem \ref{th:uniqueness} under the setting of the asymptotic regime. The reason is  that the techniques used in that proof will be key to understanding some aspects of the asymptotic behaviour of the robust scatter estimate, thereby paving the way towards the proof of Theorem \ref{th:asymptotic}. 

The proof of Theorem \ref{th:uniqueness} follows the same lines as in \cite{couillet-pascal-2013}. Define $h=(h_1,\cdots,h_n)$ with:
\begin{align*}
h_j:\mathbb{R}_+^{n}&\to\mathbb{R}_+ \\
(x_1,\cdots,x_n)&\mapsto \frac{1}{N} y_j^*\left(\frac{1}{n}\sum_{\substack{i=1\\i\neq j}}v(x_i) y_iy_i^* \right)^{-1} y_j.
\end{align*}
As it has been already mentioned, in order to prove that $\hat{C}_N$ is uniquely defined for $n$ large enough a.s., it suffices to show that the system of equations in $x_1,\cdots,x_n$
$$
x_j=h_j(x_1,\cdots,x_n), j=1,\cdots,n
$$
admits a unique solution $q_1,\cdots,q_n$ a.s. for $n$ large enough. To this end, we will show that $h$ is a standard interference function, i.e, it satisfies the following three conditions:
\begin{enumerate}[a)]
\item Positivity: For each $q_1,\cdots,q_n\geq 0$, and each $i$, $h_i(q_1,\cdots,q_n) >0$,
\item Monotonicity: For each $q_1\geq q_1^{'},\cdots,q_n\geq q_n^{'}$ and each $i$, $h_i(q_1,\cdots,q_n)\geq h_i(q_1^{'},\cdots,q_n^{'})$,
\item Scalability:  For all $\alpha >1$, and $q_1,\cdots,q_n\geq 0$, $\alpha h_i(q_1\cdots,q_n) > h_i(\alpha q_1,\cdots,\alpha q_n)$.
\end{enumerate}
 Item $a)$ can be easily shown by noticing that matrix $\frac{1}{n}\sum_{i=1,i\neq j} v(q_i) y_iy_i^*$ is invertible almost surely and is positive definite, while the monotonicity follows immediately from the fact that $h_j$ is non-decreasing of each $q_i$. As for the scalability, we can assume without loss of generality that there exists $q_i >0$ as the results holds trivially when $q_1=\cdots=q_n=0$. With this assumption at hand, we  rewrite $h_j(q_1,\cdots,q_n)$ as:
$$
h_j(q_1,\cdots,q_n)=\frac{1}{N} y_j^*\left(\frac{1}{n}\sum_{\substack{i=1\\ i\neq j}}^n\frac{\psi_N(q_i)}{q_i}y_iy_i^*\right)^{-1} y_j.
$$

As $\psi_N$ is increasing, $\psi_N(\alpha q) > \psi_N(q)$ for $\alpha >1$ and $q >0$. Hence,
\begin{align*}
h_j(\alpha q_1,\cdots,\alpha q_n) =\frac{\alpha}{N} y_j^*\left(\frac{1}{n}\sum_{\substack{i=1 \\ i\neq j}}\frac{\psi_N(\alpha q_i)}{q_i} y_iy_i^*\right)^{-1}y_j.
\end{align*}
If there exists at least $q_i >0$, we therefore get:
$$
h_j(\alpha q_1,\cdots,\alpha q_n) < \frac{\alpha}{N} y_j^*\left(\frac{1}{n}\sum_{i=1, i\neq j} \frac{\psi_N(q_i)}{q_i}y_iy_i^*\right)^{-1} y_j =\alpha h_j(q_1,\cdots,q_n).
$$
We have thus established that $h$ is a standard interference function. Referring to the results of \cite{yates}, it remains to show that there exists vector $(q_1,\cdots,q_n)$ such that for all $i=1,\cdots,n$, $q_i > h_i(q_1,\cdots,q_n)$ a.s. for $n$ large enough, a statement which is known as the feasibility condition. 

In order to establish the feasibility condition,  let $q_N^{+}$ be chosen so that:
$$
\int_0^{+\infty} \psi_N(q_N^{+} t)\nu(dt) = \frac{1+\kappa}{1-c_N\phi_{\infty}}.
$$
for some sufficiently small $0\leq \kappa\leq 1$ satisfying:
$$
1+\kappa\leq \phi_{\infty}(1-\nu\left\{0\right\}).
$$
This is possible since,
$$
\lim_{q\to+\infty} \int_0^{+\infty} \psi_N(q t)\nu(dt)=\psi_{\infty} (1-\nu\left\{0\right\}) =\frac{\phi_{\infty}(1-\nu\left\{0\right\})}{1-c_N\phi_{\infty}} > \frac{1+\kappa}{1-c_N\phi_{\infty}}.
$$
We will prove that $q_N^{+}$ is a bounded sequence. To this end, we will proceed by contradiction. Assume that there exists a sequence $(n)$ such that $\lim_{n\to+\infty} q_{(n)}=+\infty$. Since the sequence of functions $\psi_N$ converge uniformly, one can extract a subsequence $(p)$ from $(n)$ such that:
$c_{(p)}\to c^{*}$ and $\psi_{(p)}$ converge uniformly to $\psi^{*}$. Therefore, the sequence of functions $\left(f_p:q\mapsto \int_0^{+\infty} \left(1-c_{(p)}\phi_{\infty}\right)\psi_{(p)}(q t)\nu(dt)\right)$ converge uniformly to $f^{*}:q \mapsto \int_{0}^{+\infty} (1-c^{*}\phi_{\infty})\psi^{*}(q t)\nu(dt)$. Hence,
$$
\lim_{n\to+\infty} f_p(q_p) \to \lim_{x\to +\infty} f^*(x) =\psi_{\infty}(1-c^{*}\phi_{\infty})(1-\nu\left\{0\right\})=\phi_{\infty}\left(1-\nu\left\{0\right\}\right) >1+\kappa,
$$ 
which is in contradiction with the fact that:
$$
\int_{0}^{\infty} (1-c_N\phi_{\infty})\psi_N(q_N^{+} t)\nu(dt) =1+\kappa.
$$
Now, consider $\eta_N$ the unique solution of:
$$
1=\int \frac{\int_0^{+\infty}\psi_N(q_N^{+} x)\nu(dx) F^{B_N}(dy)}{\int_0^{+\infty} \psi_N(q_Nt)\frac{y+t}{t+\eta_N}\nu(dt)}.
$$
Such $\eta_N$ exists and is unique by Proposition \ref{prop:inequality}. Set $q_i=q_N^{+}(\tau_i+\eta_N)$. We will prove that this choice of $q_j,j=1,\cdots,n$ guarantees:
$$
\frac{h_j(q_1,\cdots,q_n)}{q_j} \leq 1,
$$
a.s. for $n$ large enough.
We have:
\begin{align*}
\frac{h_j(q_1,\cdots,q_n)}{q_N^{+}(\tau_j+\eta_N)} &=\frac{1}{N(\tau_j+\eta_N)} y_j^*\left(\frac{1}{n}\sum_{i\neq j} \frac{\psi_N(q_N^{+}(\tau_i+\eta_N))}{\tau_i+\eta_N}y_iy_i^*\right)^{-1}y_j\\
&\leq \frac{1}{N(\tau_j+\eta_N)\int_{0}^{+\infty}{\psi_N}(q_N^{+} x)\nu(dx)} y_j^*\left(\frac{1}{n}\sum_{i\neq j} \frac{\overline{\psi}_N(q_N^{+}\tau_i)}{\tau_i+\eta_N}y_iy_i^*\right)^{-1}y_j,
\end{align*}
where $\overline{\psi}_{q_N^{+}}(x)=\frac{\psi_N(q_N^{+} x)}{\int_0^{+\infty}\psi_N(q_N^{+} t)\nu(dt)}$.
From item $ii)$ of Lemma \ref{lemma:quadratic_alpha}, we have:
$$
\max_{1\leq j\leq n} \left|\frac{1}{N(\tau_j+\eta_N)} y_j^*\left(\frac{1}{n}\sum_{i\neq j} \frac{\overline{\psi}_{q_N^{+}}(q_N^{+}\tau_i)}{\tau_i+\eta_N}y_iy_i^*\right)^{-1}y_j-\frac{1}{N}\tr \frac{B_N+\tau_j I_N}{\tau_j+\eta_N}T_N\right|\asto 0.
$$
where $T_N=\left(\frac{1}{n}\sum_{i=1}^n \frac{\overline{\psi}_{q_N^{+}}(\tau_i)(B_N+\tau_i I_N)}{(\tau_i+\eta_N)(1+e_i)}\right)^{-1}$ with $e_1,\cdots,e_n$ are the unique solutions to the following system of equations:
$$
e_k=\frac{\overline{\psi}_{q_N^{+}}(\tau_k)}{n}\tr\frac{\left(B_N+\tau_k I_N\right)}{\tau_k+\eta_N}\left(\frac{1}{n}\sum_{i=1}^n \frac{\overline{\psi}_{q_N^{+}}(\tau_i)\left(B_N+\tau_i I_N\right)}{(\tau_i+\eta_N)(1+e_i)}\right)^{-1}.
$$
Let $j_0$ be the index of the maximum element in $\left\{e_1,\cdots,e_n\right\}$. Then, there exists $\epsilon_n\downarrow 0 $ such that for all $j=1,\cdots,n$
\begin{equation}
\frac{h_j(q_1,\cdots,q_n)}{q_j}\leq \frac{1+e_{j_0}}{\int_0^{+\infty}\psi_N(q_N^{+} x)\nu(dx)}\frac{1}{N}\tr \frac{B_N+\tau_j I_N}{\tau_j+\eta_N}\left(\frac{1}{n}\sum_{i=1}^n \frac{\overline{\psi}_{q_N^{+}}(\tau_i)(B_N+\tau_i I_N)}{\tau_i+\eta_N}\right)^{-1} +\epsilon_n.
\label{eq:hj}
\end{equation}
As $\left\|\overline{\psi}_{q_N^{+}}\right\|_{\infty}=\frac{\psi_{\infty}}{\int_0^{+\infty}\psi_N(q_N^{+} t)\nu(dt)}=\frac{\phi_{\infty}}{1+\kappa} \leq \frac{1}{c_N(1+\kappa)}$, we obtain from item $iii)$ of Lemma \ref{lemma:quadratic_alpha},
\begin{equation}
\max_{1\leq k\leq n}e_k \leq \frac{c_N\left\|\overline{\psi}_{q_N^{+}}\right\|_{\infty} }{1-c_N\left\|\overline{\psi}_{q_N^{+}}\right\|_{\infty}} +\epsilon_n^{'}=\frac{c_N\phi_{\infty}}{1+\kappa-c_N\phi_{\infty}}+|o(1)|.
\label{eq:epsilon}
\end{equation}
where $o(1)$ refers to some sequences converging almost surely to zero as $n$ grow to infinity.
Plugging \eqref{eq:epsilon} into \eqref{eq:hj}, and using the fact that:
$$
\frac{1}{N}\tr\frac{B_N+\tau_j I_N}{\tau_j+\eta_N}\left(\frac{1}{n}\sum_{i=1}^n \frac{\overline{\psi}_{q_N^{+}}(\tau_i)(B_N+\tau_i I_N)}{\tau_i+\eta_N}\right)^{-1} \leq 1+|o(1)|,
$$
we finally get:
$$
\frac{h_j(q_1,\cdots,q_n)}{q_j} \leq \frac{1-c_N\phi_{\infty}}{1+\kappa -c_N\phi_{\infty}} + |o(1)|
$$
thereby establishing that:
$$
h_j(q_1,\cdots,q_n) <q_j 
$$
a.s. for $n$ large enough.
\end{paragraph} 
\begin{paragraph}{{\bf Proof of Theorem \ref{th:asymptotic}: Asymptotic Convergence of the Robust-Scatter Estimator}}
The proof of Theorem \ref{th:asymptotic} heavily relies on the new rewriting of the robust-scatter estimate as:
\begin{equation}
\hat{C}_N=\frac{1}{n}\sum_{i=1}^n v(q_i)y_iy_i^*,
\label{eq:new_rewriting}
\end{equation}
where $q_1,\cdots,q_n$ are the unique solutions of the following system of equations:
$$
q_j=y_j^*\left(\frac{1}{n}\sum_{i=1,i\neq j}^n v(q_i)y_iy_i^*\right)^{-1}y_j,
$$
their existence and uniqueness in the asymptotic regime being established in the proof of Theorem \ref{th:uniqueness}. From the rewriting of $\hat{C}_N$ in \eqref{eq:new_rewriting}, it appears that an in-depth study of the asymptotic behaviour of $q_1,\cdots,q_n$ can be a good starting point. As mentioned in our heuristic analysis, one intuitively expects the $q_1,\cdots,q_n$ to approach in the asymptotic regime $\delta_1,\cdots,\delta_n$, the solutions of the following system of equations: 
$$
\delta_i=\frac{1}{N}\tr (B_N+\tau_i I_N)\left(\frac{1}{n}\sum_{j=1}^n \frac{v(\delta_j)(B_N+\tau_j I_N)}{1+c_N\psi_N(\delta_j)}\right)^{-1}.
$$
This intuition underlies the proof of Theorem \ref{th:asymptotic}. In particular, we will prove that:
$$
f_i=\frac{v(q_i)}{v(\delta_i)}, \hspace{0.5cm}i=1,\cdots,n
$$
satisfy:
\begin{equation}
\max_{1\leq i\leq n}\left|f_i-1\right|\asto 0.
\label{eq:e_i}
\end{equation}
This in particular will allow us to state that $\hat{C}_N$ can be approximated by $\hat{S}_N=\frac{1}{n}\sum_{i=1}^n v(\delta_i) y_iy_i^*$. The importance of this finding lies in the fact that unlike $\hat{C}_N$, $\hat{S}_N$ follows a classical random matrix model, thereby opening up possibilities of exploiting an important load of available results.  
%With this new rewriting at hand, we will focus in the proof of Theorem \ref{th:asymptotic} on an in-depth study of the asymptotic behaviour of $q_1,\cdots,q_n$. 
%While the existence and uniqueness of $q_1,\cdots,q_n$ have been established in the proof of Theorem \ref{th:uniqueness}, 
%While Theorem \ref{th:uniqueness} guarantees the existence and uniqueness of $q_1,\cdots,q_n$, it does not help 
%The analysis in \cite{couillet-pascal-2013} involves solving  one fixed point equation which is independent of 
%A large effort must be devoted in order to understand the behaviour of the deterministic quantities $\delta_i$.  
%Studying the asymptotic behaviour of high-rank data models is however more challenging. 
Prior to proceeding into the proof of the convergence stated in \eqref{eq:e_i}, we first need  to introduce the following key lemmas that allow to identify the intervals within which lie almost surely quantities $q_1,\cdots,q_n$ and $\delta_1,\cdots,\delta_n$. We start by handling terms $\delta_1,\cdots,\delta_n$. We have in particular the following Lemma: 
\begin{lemma}
\label{lemma:control_boundedness}
Let:
$$\begin{array}{l}
	h_j:\mathbb{R}_+^n\to \mathbb{R}_{+} \\
	 (x_1,\cdots,x_n):\mapsto \left\{\begin{array}{ll} 
	\frac{1}{N}\tr \left(B_N+\tau_j I\right)\left(\frac{1}{n}\displaystyle{\sum_{i=1}^n} \frac{\psi(x_i)(B_N+\tau_i I)}{x_i(1+c_N \psi(x_i))}\right)^{-1} & \textnormal{if} \  \exists  \  x_i \ \ \neq 0 \\
	\frac{1}{N}\tr \left(B_N+\tau_j I\right) \left(\frac{1}{n}\displaystyle{\sum_{i=1}^n} v(0)\left(B_N+\tau_i I\right)\right)^{-1} & \textnormal{otherwise.}
	\end{array}
	 \right.\end{array} $$
Then, for all large $n$, there exists a unique vector $(\delta_1,\cdots,\delta_n)\in\mathbb{R}_{+}^n$ such that:
\begin{equation}
h_j(\delta_1,\cdots,\delta_n)=\delta_j, \hspace{0.1cm} \forall \hspace{0.1cm} 1\leq j\leq n.
\label{eq:delta_j}
\end{equation}
Besides, vector $(\delta_1,\cdots,\delta_n)$ is given by:
$$
(\delta_1,\cdots,\delta_n)=\lim_{t\to+\infty} (\delta_1^t,\cdots,\delta_n^t),
$$
with $(\delta_1^0,\cdots,\delta_n^{0})$ chosen arbitrarily in $\mathbb{R}_{+}^n$ and:
$$
\delta_j^{t+1}=h_j(\delta_1^t,\cdots,\delta_n^t).
$$
Moreover, there exists $\delta_N^{+}$ and $\delta_N^{-}$ with $\limsup \delta_N^{+}<\infty$ and $\lim\inf \delta_N^{-}>0$ and $\eta_N^{+},\eta_N^{-}$ such that, almost surely for $n$ large enough, we have:
$$
\delta_N^{-}(\tau_j+\eta_N^{-})\leq \delta_j\leq \delta_N^{+}(\tau_j+\eta_N^{+}), \hspace{0.1cm} 1\leq j\leq n.
$$
Moreover, $\eta_N^{+}$ and $\eta_N^{-}$ satisfy:
$$
\eta_N^{+}=O\left(\frac{1}{N}\tr B_N\right) \hspace{0.1cm} \textnormal{and} \hspace{0.1cm} \eta_N^{-}=O\left(\frac{1}{N}\tr B_N\right).
$$
\end{lemma}
\begin{proof}
Similar to the proof of Theorem \ref{th:uniqueness}, we can show along the same lines that $h=(h_1,\cdots,h_n)$ is a standard interference function.   
%As in the proof of Theorem \ref{th:uniqueness}, we show that $h=(h_1,\cdots,h_n)$ is a standard interference function. Positivity, monotonicity and scalability of $h$ follows using the same arguments as before. As a matter of fact, positivity is obvious since $h$ involves the trace of positive non-null matrices. As for monotonicity it follows from the fact that $x\mapsto \frac{v(x)}{1+c_N\psi(x)}$ is non-increasing. Finally to prove scalability, note that $x\mapsto \frac{\psi(x)}{1+c_N \psi(x)}$ is increasing and as such $\frac{\psi(\alpha x)}{1+c_N \psi(\alpha x)} \geq \frac{\psi(x)}{1+c_N \psi(x)}$ for $\alpha >1$. Hence:
%	\begin{align*}
%	h_j(\alpha \delta_1,\cdots,\alpha \delta_n) &=\frac{1}{N}\tr (B+\tau_j I) \left(\frac{1}{n}\sum_{i=1}^n \frac{\psi(\alpha \delta_i)(B+\tau_i I)}{\alpha \delta_i(1+c_N \psi(\alpha \delta_i))}\right)^{-1} \\
%	& \leq \alpha \frac{1}{N}\tr (B+\tau_j I)\left(\frac{1}{n}\sum_{i=1}^n \frac{\psi( \delta_i)(B+\tau_i I)}{\delta_i(1+c_N \psi(\alpha \delta_i))}\right)^{-1} =\alpha h_j(\delta_1,\cdots,\delta_n)
%	\end{align*}
	It remains to prove the existence of $(\delta_1,\cdots,\delta_n)$ such that $h_j(\delta_1,\cdots,\delta_n) < \delta_j$. 
	To this end, take $\delta_j=\xi(\tau_j+\eta_N)$, where $\eta_N$ is the unique solution to the following equation:
	$$
	\int \frac{F^{B_N}(dy)}{\int \frac{y+t}{t+x}\nu(dt)}=1.
	$$
	Such $\eta_N$ exists according to proposition \ref{prop:inequality}.
	Then:
	\begin{equation}
	\lim_{\xi\to+\infty} \frac{h_j(\delta_1,\cdots,\delta_n)}{\delta_j} = \frac{1+c_N\psi_{\infty}}{\psi_{\infty}}\frac{1}{N}\tr \frac{(B+\tau_j I)}{\tau_j+\eta_N}\left(\frac{1}{n}\sum_{i=1}^n \frac{B+\tau_i I}{\tau_i+\eta_N}\right)^{-1}.
	\label{eq:limit}
	\end{equation}
	Again, the limit in the above equation \eqref{eq:limit} can be controlled using proposition \ref{prop:inequality}, thereby yielding:
	$$
	\lim_{\xi\to+\infty} \frac{h_j(\delta_1,\cdots,\delta_n)}{\delta_j} \leq \frac{1+c_N\psi_{\infty}}{\psi_{\infty}} +\epsilon_n =\frac{1}{\phi_{\infty}} +\epsilon_n 
	$$
	where $\epsilon_n\downarrow 0$ almost surely. As $\phi_{\infty}>1$, one can conclude that there exists $\delta_1,\cdots,\delta_n$, such that for enough large $N$,
	$$
	h_j(\delta_1,\cdots,\delta_n) < \delta_j.
	$$
	We are now in position to prove the uniform boundedness of $\delta_i$. 
	For that, consider $\theta >0$ such that $\theta <\frac{\phi_{\infty}}{2(1-c_{+}\phi_{\infty})}$. Let $M_\theta$ be chosen such that $\nu(M_\theta,+\infty) <\theta$ and $M_\theta$ is greater than the limit support of $\|B_N\|$. Set
	 $\delta_N^{-}$ and $\delta_N^{+}$ so that the following conditions are fulfilled: 
	\begin{align}
	&\int_{0}^{+\infty} \frac{\psi_N(\delta_N^{+}t)}{1+c_N\psi(\delta_N^{+}t)}\nu(dt) >1,\label{eq:fplus}\\
	&\int_{0}^{M_\theta} \psi_N\left(\delta_N^{-}(t+M_\theta)\right)\nu(dt) < \frac{1}{2} \label{eq:fmoins}.
	%%\frac{1+c_N\Psi(2\delta^{-}M_\theta)}{\Psi(2\delta^{-}M_\theta)} &> 2 \\
	%%\frac{1+c_N \Psi(m\delta^{+})}{\Psi(m\delta^{+})} & < (1-\epsilon)
	\end{align}
	 Such choices are possible since 
 \begin{itemize}
\item $\displaystyle{\lim_{\delta^{+}\to+\infty}}\int_{0}^{+\infty}\frac{\psi_N(\delta^{+}t)}{1+c_N\psi_N(\delta^{+}t)}\nu(dt)=\frac{(1-\nu\left\{0\right\})\psi_{\infty}}{1+c_N\psi_{\infty}}=(1-\nu\left\{0\right\})\phi_{\infty}>1 $,
\item $\displaystyle{\lim_{\delta^{-}\to 0^{+}}} \int_{0}^{M_\theta} {\psi\left(\delta^{-}(t+M_\theta)\right)}\nu(dt) =0$,
 %\item $\lim_{\delta^{-}\to 0^{+}}\frac{1+c_N\Psi(\delta^{-}M_\theta)}{\Psi(\delta^{-}M_\theta)} =+\infty$ 
 %\item $1-\epsilon > \frac{1}{\Phi_{\infty}}$ where $x\mapsto \frac{1+c_N\psi(x)}{\psi(x)}$ is decreasing and tends to $\frac{1}{\Phi_{\infty}}$ as $x\to+\infty$.
 \end{itemize}
Moreover, we can check that one can choose $\delta_N^{-}$ and $\delta_N^{+}$ such that $\lim\inf \delta_N^{-} >0$ and $\lim\sup\delta_N^{+} < +\infty$. As a matter of fact, building on the same reasoning used to show that $\lim\sup \alpha_N<+\infty$ in the proof of Theorem \ref{th:uniqueness}, we take $\delta_N^{+}$ and $\delta_N^{-}$  the positive reals that verify:
\begin{align*}
&\int_{0}^{+\infty} \frac{\psi(\delta_N^{+}t)}{1+c_N\psi(\delta_N^{+}t)}\nu(dt)=1+\kappa\\
 &\int_{0}^{M_\theta} {\psi\left(\delta_N^{-}(t+M_\theta)\right)}\nu(dt) =\kappa,
\end{align*}
where $0\leq \kappa<\frac{1}{2}$  satisfies $1+\kappa < (1-\nu\left\{0\right\})\phi_{\infty}$. Assume that $\lim\inf \delta_N^{-} =0$. There exists a sequence $(n)$ such that $\lim_{n\to+\infty} \delta_N^{-}=0$. Since the sequence of functions $\psi_N$ converge uniformly, one can extract a subsequence $(p)$ from $(n)$ such that $c(p)\to c^{*}$ and $\psi_{p}$ converge uniformly to $\psi^{*}$. Therefore, the sequence of functions $f_p:\alpha \mapsto \int_0^{M_\theta}  {\psi\left(\alpha(t+M_\theta)\right)}\nu(dt)$ converge uniformly to $f^*:\alpha \mapsto \int_0^{M_\theta} {\psi\left(\alpha(t+M_\theta)\right)}\nu(dt)$.
Hence:
$$
\lim_{n\to+\infty}f_p(\delta_p^{-}) \to \lim_{x\to 0} f^{*}(x) =0.
$$
which is in contradiction with the fact that $f_p(\delta_p^{-})=\kappa$. The same method can be used to prove that $\lim\sup \delta_N^{+} <\infty$.
Consider now the function $f^{+}:t\mapsto \frac{\psi(\delta_N^{+}t)}{\left(1+c_N\psi(\delta_N^{+}t)\right)}$ in the domain  $t\in\left[0,\infty\right)$. Define $\eta_N^{+}$ the unique solution to the following equation:
$$
1=\int \frac{F^{B_N}(dy)\int_0^{+\infty}f^{+}(x)\nu(dx)}{\int_{0}^{+\infty}\frac{y+t}{t+\eta_N^{+}}f^{+}(t)\nu(dt)}.
$$
Similarly, define on $\mathbb{R}^{+}$ the function $f^{-}:t\mapsto \psi_{\infty}1_{\left\{t\geq M_\theta\right\}} +{\psi(2\delta^{-}(t+M_\theta))} 1_{\left\{t<M_\theta\right\}}$. Let $\eta_N^{-}$ be the unique solution to the following equation:
\begin{equation}
1=\int\frac{F^{B_N}(dy)\int_0^{+\infty}f^{-}(x)\nu(dx)}{\int_0^{+\infty}\frac{y+t}{t+\eta_N^{-}}f^{-}(t)\nu(dt)}.
\label{eq:eta_N_moins}
\end{equation}
Note that from proposition \ref{prop:inequality}, $\eta_N^{+}$ and $\eta_N^{-}$ are well-defined and satisfy:
$$
\eta_N^{+}=\mathcal{O}\left(\frac{1}{N}\tr B_N\right) \hspace{0.1cm}, \eta_N^{-}=\mathcal{O}\left(\frac{1}{N}\tr B_N\right).
$$

Set for all $i$, $\delta_i^{0}=\delta_N^{+}(\tau_i+\eta_N^+)$. Define recursively the sequences:
$$
\delta_{j}^{t+1}=\frac{1}{N}\tr \left(B_N+\tau_j I_N\right)\left(\frac{1}{n}\sum_{i=1}^n \frac{\psi(\delta_i^{t})(B_N+\tau_i I_N)}{\delta_i^{t}(1+c_N \psi(\delta_i^{t}))}\right)^{-1}.
$$
From the previous analysis, $\delta_i=\lim_{t\to+\infty} \delta_i^{t}$.  To prove the uniform boundedness of $\delta_i$, one can proceed by induction on $t$. For $t=0$, the result is true. Let $t\in\mathbb{N}^{*}$ and assume that $\delta_j^{k}\leq \delta_N^{+}(\tau_j+\eta_N^{+})$ holds true for any $k\leq t$ and $j=1,\cdots,n$. We propose to prove it for $k=t+1$. We have
%let $\epsilon >0$ and $m >0$ such that $(1-\epsilon)\phi_{\infty} >1$ and $\nu_n((m,\infty))> 1-\epsilon$ for all $n$ large a.s. (always possible by Assumption\ref{ass:statistical}-iii).
\begin{align*}
\frac{\delta_j^{t+1}}{\delta_N^{+}(\tau_j+\eta^{+})}&= \frac{1}{N}\tr \frac{B+\tau_j I}{\delta_N^{+}(\tau_j+\eta^{+})}\left(\frac{1}{n}\sum_{i=1}^n \frac{\psi(\delta_i^{t})(B+\tau_i I)}{\delta_i^t(1+c_N \psi(\delta_i^{t}))}\right)^{-1}. \\
%& =\frac{N-S}{N}\frac{\tau_j}{\frac{1}{n}\sum_{i=1}^n \frac{\Psi(\delta_i^t)\tau_i}{\delta_i^t(1+c_N \Psi(\delta_i^t))}} +\frac{1}{N}\sum_{k=1}^{S} \frac{\lambda_k+\tau_j}{(\delta^{+}\tau_j+\eta)\frac{1}{n}\sum_{i=1}^n \frac{\Psi(\delta_i^t)(\lambda_k+\tau_i)}{\delta_i^t(1+c_N\Psi(\delta_i^t))}}
\end{align*}
From the induction assumption along with the fact that $x\mapsto \frac{\psi(x)}{x(1+c_N\psi(x))}$ is non-increasing, we obtain:
$$
\frac{\psi(\delta_i^t)}{\delta_i^t\left(1+c_N\psi(\delta_i^t)\right)}\geq \frac{\psi(\delta_N^{+}(\tau_i+\eta_N^{+}))}{\delta_N^{+}(\tau_i+\eta_N^{+})(1+c_N\psi(\delta_N^{+}(\tau_i+\eta_N^{+})))}.
$$
Hence,
\begin{align*}
\frac{\delta_j^{t+1}}{\delta_N^{+}(\tau_j+\eta_N^{+})} &\leq \frac{1}{N}\tr \frac{B_N+\tau_j{I}_N}{\tau_j+\eta_N^{+}}\left(\frac{1}{n}\sum_{i=1}^n \frac{\psi(\delta_N^{+}\tau_i)}{\left(1+c_N\psi(\delta_N^{+}\tau_i)\right)}\frac{B_N+\tau_i{I}_N}{\tau_i+\eta_N^{+}}\right)^{-1}\\
&=  \frac{1}{N}\tr \frac{B_N+\tau_j{I}_N}{\tau_j+\eta_N^{+}}\left(\frac{1}{n}\sum_{i=1}^n f^{+}(\tau_i)\frac{B_N+\tau_i{I}_N}{\tau_i+\eta_N^{+}}\right)^{-1}.
\end{align*}
%Consider the function $f:t\mapsto \frac{\frac{\psi(\delta^{+}t)}{1+c_N\psi(\delta^{+}t)}}{\int_0^{\infty} \frac{\psi(\delta^{+}x)}{1+c_N\psi(\delta^{+}x)}\nu(dx)}$. 
From Remark \ref{remark:psi} along with Lemma \ref{lemma:g_N}, function $t\mapsto \frac{f^{+}(t)}{\int f^{+}(x)\nu(dx)}$ satisfies the assumptions of proposition \ref{prop:inequality}. We have therefore, 
$$
\frac{\int_{0}^{+\infty}f^{+}(x)\nu(dx)}{N}\tr \frac{B_N+\tau_jI_N}{\tau_j+\eta_N^{+}}\left(\frac{1}{n}\sum_{i=1}^nf^{+}(\tau_i) \frac{B+\tau_i I_N}{\tau_i+\eta_N^{+}}\right)^{-1} \leq 1+ \epsilon_{n,j}, \ \ \forall 1\leq j\leq n.
$$
where $\max_{j}\left|\epsilon_{n,j}\right|$ converges to zero almost surely. 
Equation \eqref{eq:fplus} guarantees that  $\int_0^{+\infty}f^{+}(x)\nu(dx) >1$, thereby showing, that almost surely for $n$ large enough:
$$
\frac{\delta_j^{t+1}}{\delta_N^{+}(\tau_j+\eta_N^{+})} \leq 1. 
$$ 

We will now prove the lower-bound inequality. Similarly, consider for all $i$, $\delta_i^0=2\delta_N^{-}(\tau_i+\eta_N^{-})$. The sequence:
\begin{align*}
\delta_j^{t+1}&=\frac{1}{N}\tr (B+\tau_j I)\left(\frac{1}{n}\sum_{i=1}^n \frac{\psi(\delta_i^t)(B+\tau_i I)}{\delta_i^t(1+c_N\psi(\delta_i^t))}\right)^{-1}
\end{align*}
converges to $\delta_i^*$ as $t\to+\infty$. In the same way as for the upper-bound inequality, we will show the result by induction on $t$. For $t=0$, the result is true. Let $t\in\mathbb{N}^+$ and assume that $\delta_j^k \geq \delta_N^{-}(\tau_j+\eta_N^{-})$ holds true for any $k\leq t$ and $j=1,\cdots n$. We propose to prove the result for $k=t+1$. 
Similar to above, using the fact that $x\mapsto \frac{\psi(x)}{x}$ is non-increasing, we have:
\begin{align*}
\frac{\delta_{j}^{t+1}}{\delta_N^{-}\left(\tau_j+\eta_N^{-}\right)} &= \frac{1}{N}\tr \frac{B_N+\tau_j I_N}{\delta_N^{-}\left(\tau_j+\eta_N^{-}\right)} \left(\frac{1}{n}\sum_{i=1}^n \frac{\psi(\delta_i^t)(B_N+\tau_i I_N)}{\delta_i^t(1+c_N\psi(\delta_i^t))}\right)^{-1}  \\
&\geq \frac{1}{N}\tr \frac{B_N+\tau_j I_N}{\left(\tau_j+\eta_N^{-}\right)} \left(\sum_{i=1}^n\frac{\psi(\delta_N^{-}(\tau_i+\eta_N^{-}))(B_N+\tau_i I_N)}{(\tau_i+\eta_N^{-})}\right)^{-1}\\
&=\frac{1}{N}\tr \frac{B_N+\tau_j I_N}{\left(\tau_j+\eta_N^{-}\right)} \left(\sum_{\tau_i\geq M_\theta}\frac{\psi(\delta_N^{-}(\tau_i+\eta_N^{-}))(B_N+\tau_i I_N)}{(\tau_i+\eta_N^{-})}+\sum_{\tau_i\leq M_\theta}\frac{\psi(\delta_N^{-}(\tau_i+\eta_N^{-}))(B_N+\tau_i I_N)}{(\tau_i+\eta_N^{-})}\right)^{-1} \\
&\stackrel{(a)}{\geq}  \frac{1}{N}\tr \frac{B_N+\tau_j I_N}{\left(\tau_j+\eta_N^{-}\right)} \left(\sum_{\tau_i\geq M_\theta} \frac{\psi_{\infty}(B_N+\tau_i I_N)}{(\tau_i+\eta_N^{-})}+\sum_{\tau_i\leq M_\theta}\frac{\psi(\delta_N^{-}(M_\theta+\tau_i))(B_N+\tau_i I_N)}{(\tau_i+\eta_N^{-})}\right)^{-1}\\
%&\geq \frac{1}{N}\tr \frac{B+\tau_j I}{\left(\tau_j+\eta_N^{-}\right)} \left(\sum_{\tau_i\geq M_\theta}^n \frac{\phi_{\infty}(B+\tau_i I)}{(\tau_i+\eta_N^{-})}+\sum_{\tau_i\leq M_\theta}^n\frac{\psi(2\delta^{-}M_\theta)(B+\tau_i I)}{(\tau_i+\eta_N^{-})(1+c_N\psi(2\delta^{-}M_\theta)}\right)^{-1}\\
&=\frac{1}{N}\tr  \frac{B_N+\tau_j I_N}{\left(\tau_j+\eta_N^{-}\right)} \left(\sum_{i=1}^n f^{-}(\tau_i)\frac{B_N+\tau_i I_N}{\tau_i+\eta_N^{-}}\right),
%&=\frac{1}{N}\tr \frac{B+\tau_j I}{\delta^{-}\left(\tau_j+\eta^{-}\right)} \left(\frac{1}{n}\sum_{\tau_i \leq M_\theta}\frac{\Psi(\delta_i^t)(B+\tau_i I)}{\delta_i^t(1+c_N\Psi(\delta_i^t))} +\sum_{\tau_i\geq M_\theta} \frac{\Psi(\delta_i^t)(B+\tau_i I)} {\delta_i^t(1+c_N\Psi(\delta_i^t))} \right)^{-1} \\
%&\geq \frac{1}{N}\tr \frac{B+\tau_j I}{\delta^{-}\left(\tau_j+\eta^{-}\right)}\left(\right)^{-1}\left(\frac{1}{n}\sum_{\tau_i\leq M_\theta} \frac{\Psi(\delta)}\right)
\end{align*}
where $(a)$ follows from the fact that $\eta_N^{-}\leq \|B_N\|\leq M_\theta$.
Again, from remark \ref{remark:psi} and Lemma \ref{lemma:g_N}, function $t\mapsto \frac{f^{-}(t)}{\int f^{-}(x)\nu(dx)}$ satisfies the assumptions of proposition \ref{prop:inequality}. We have therefore,
$$
\frac{\int_0^{+\infty} f^{-}(x)\nu(dx)}{N}\tr \frac{B_N+\tau_jI_N}{\tau_j+\eta_N^{-}}\left(\frac{1}{n}\sum_{i=1}^n f^{-}(\tau_i)\frac{B_N+\tau_i I_N}{\tau_i+\eta_N^{-}}\right)^{-1} \geq 1-|\epsilon_{n,j}|,
$$
where $\max_j\left|\epsilon_{n,j}\right|$ converges to zero almost surely. On the other hand,
  $$\int_0^{+\infty} f^{-}(x)\nu(dx) \leq \theta \frac{\phi_{\infty}}{1-c_{+}\phi_{\infty}}+\int_0^{M_\theta}{\psi(\delta_N^{-}(t+M_\theta))}<1,$$
and hence, almost surely, for enough large $n$, 
$$
\frac{\delta_{j}^{t+1}}{\delta^{-}(\tau_j+\eta_N^{-})} \geq \frac{1}{N}\tr  \frac{B+\tau_j I}{\left(\tau_j+\eta_N^{-}\right)} \left(\sum_{i=1}^n f^{-}(\tau_i)\frac{B_N+\tau_i I_N}{\tau_i+\eta_N^{-}}\right) >1.
$$
\end{proof}
The following  refinement of Lemma \ref{lemma:control_boundedness} will be required in the proof of  the asymptotic convergence of the robust-scatter estimator.
\begin{lemma}
\label{lemma:extension_boundedness}
Let $(\kappa,M_\kappa)$ be couples indexed by $\kappa$ with $0 < \kappa <1$ and $M_\kappa >0$ such that $\nu(M_\kappa,+\infty) < \kappa$. Then, for sufficiently small $\kappa$ the following system of equations:
\begin{equation}
\delta_j=\frac{1}{N}\tr \left(B_N+\tau_j I_N\right)\left(\frac{1}{n}\sum_{\substack{i=1\\ \tau_i \leq M_\kappa}}^n \frac{\psi(\delta_i)(B_N+\tau_i I_N)}{\delta_i(1+c_N \psi(\delta_i))}\right)^{-1}, \ \ \forall 1\leq j\leq n
\label{eq:delta_kappa}
\end{equation}
has a unique vector solution $\left(\delta_1^{\kappa},\cdots,\delta_n^{\kappa}\right)$ for all large $n$ a.s, and there exists $\delta_N^{-,\kappa_0},\delta_N^{+,\kappa_0}$ with $\lim\sup \delta_N^{+,\kappa_0} <\infty$ and $\lim\inf \delta_N^{-,\kappa_0} >0$ and
$ \eta_N^{-,\kappa_0}$, $\eta_N^{+,\kappa_0}$ such that for all $\kappa <\kappa_0$ small:
$$
\delta_N^{-,\kappa}(\tau_i+\eta_N^{-,\kappa})\leq \delta_{i}^{\kappa} \leq \delta_N^{+,\kappa}(\tau_i+\eta_N^{+,\kappa}), \ \ i= 1,\cdots,n
$$
for all large $n$ a.s. Moreover, $\eta_N^{+,\kappa_0}$ and $\eta_N^{-,\kappa_0}$ satisfies:
$$
\eta_N^{+,\kappa_0}=O\left(\frac{1}{N}\tr B_N\right) \hspace{0.2cm} \textnormal{and} \hspace{0.2cm} \eta_N^{-,\kappa_0}=O\left(\frac{1}{N}\tr B_N\right)
$$
\end{lemma}
\begin{proof}
The same proof as that of Lemma \ref{lemma:extension_boundedness} holds by taking $\kappa_0$ smaller than $\theta$ and choosing $\delta_N^{+}$ so that it satisfies:
\begin{align*}
&\int_0^{M_{\kappa_0}} \frac{\psi(\delta_N^{+}t)}{1+c_N\psi(\delta_N^{+}t)} \nu(dt) >1\\
\end{align*}
while  $\delta_N^{-}$ is set in the same way as before. 
\end{proof}
We will now provide similar results for the random quantities $q_1,\cdots,q_n$. In particular, we have the following Lemma:
\begin{lemma}
 Let $q_i\triangleq y_i^* \hat{C}_{(i)}^{-1} y_i, i=1,\cdots,n$.
There exists $q_N^{+}, q_N^{-},\alpha_N^{+},\alpha_N^{-}> 0$ with $\limsup_N q_N^{+} <+\infty$ and $\liminf_N q_N^{-} > 0$ such that, for all large $n$ a.s.,
\begin{equation}
 q_N^{-}(\tau_i+\alpha_N^{-})\leq q_i \leq q_N^{+}(\tau_i  +\alpha_N^{+}), \hspace{0.2cm} i=1,\cdots,n. 
\label{eq:lower_upper_bound_q}
\end{equation}
\end{lemma}
\begin{proof}
The proof is based on the same tools as those used to show Lemma \ref{lemma:control_boundedness}. The single difference is that the random quantities $q_i$ involve quadratic forms which will be treated by resorting to Lemma \ref{lemma:quadratic_alpha}. First recall that $q_1,\cdots,q_n$ are given by:
$$
(q_1,\cdots,q_n)=\lim_{t\to+\infty} (q_1^{t},\cdots,q_n^{t})
$$
with $(q_1^0,\cdots,q_n^{0})$ chosen arbitrarily in $\mathbb{R}_{+}^n$ and:
$$
q_j^{t+1}=\frac{1}{N}y_j^*\left(\frac{1}{n}\sum_{i=1,i\neq j}^n \frac{\psi_N(q_i^{t})}{q_i^t} y_iy_i^*\right)^{-1}y_j.
$$
Similar to the proof of Theorem \ref{th:uniqueness}, consider $(q_N^{+})$ so that $\int_0^{+\infty}\psi_N(q_N^{+}t) \nu(dt)>\frac{1}{1-c_N\phi_{\infty}}$ and $\lim\sup q_N^{+}<+\infty$. Let $\alpha_N^{+}$ be the unique solution of:
$$
1=\int \frac{\int_0^{+\infty} \psi(q_N^{+}x)\nu(dx)}{\int_0^{+\infty}\psi(q_N^{+}t)\frac{y+t}{t+\alpha_N}\nu(dt)}F^{B_N}(dy).
$$
Set $q_i^{+}=\frac{q_N^{+}}{2}(\tau_i+\alpha_N)$. We will prove by induction on $t$ that $q_i^{t}\leq q_N^{+}(\tau_i+\alpha_N)$. For $t=0$, the result holds true. Assume now that for all $k\leq t$:
$$
q_i^{k}\leq q_N^{+}(\tau_i+\alpha_N),
$$
and let us show that $q_i^{t+1} \leq q_N^{+}(\tau_i+\alpha_N)$.
We have:
\begin{align*}
\frac{q_j^{t+1}}{q_N^{+}(\tau_j+\alpha_N)} &=\frac{1}{Nq_N^{+}(\tau_j+\alpha_N^{+})}y_j^*\left(\frac{1}{n}\sum_{i\neq j}\frac{\psi(q_i^{t})}{q_i^{t}}y_iy_i^*\right)^{-1}y_j \\
&\leq \frac{1}{N(\tau_j+\alpha_N^{+})}y_j^*\left(\frac{1}{n}\sum_{i\neq j}\frac{\psi(q_N^{+}\tau_i)}{\tau_i+\alpha_N^{+}}y_iy_i^*\right)^{-1}y_j.
\end{align*}
Let $\overline{\psi}_{q_N^{+}}(x)= \frac{\psi(q_N^{+}x)}{\int_0^{+\infty}\psi(q_N^{+}t)\nu(dt)}$. From item $ii)$ in Lemma \ref{lemma:quadratic_alpha}, we have:
\begin{equation}
\max_{1\leq j\leq n}\left| \frac{1}{N(\tau_j+\alpha_N^{+})\int_0^{+\infty}\psi(q_N^{+}x\nu(dx))}y_j^*\left(\frac{1}{n}\sum_{i\neq j} \frac{\overline{\psi}_{q_N^{+}}(\tau_i)}{\tau_i+\alpha_N^{+}}y_iy_i^*\right)^{-1}y_j-\frac{1}{N\int_0^{+\infty}\psi(q_N^{+}x)\nu(dx)}\tr \frac{B_N+\tau_j I_N}{\tau_j+\alpha_N^{+}}T_N\right|\asto 0,
\label{eq:convergenceqj}
\end{equation}
where $T_N=\left(\frac{1}{n}\sum_{i=1}^n \frac{\overline{\psi}_{q_N^{+}}(\tau_i)(B_N+\tau_i I_N)}{(\tau_i+\alpha_N^{+})(1+e_i)}\right)^{-1}$ with $e_1,\cdots,e_n$ the unique solutions to the following system of equations:
$$
e_k=\frac{\overline{\psi}_{q_N^{+}}(\tau_k)}{n}\tr \frac{B_N+\tau_k I_N}{\tau_k+\alpha_N^{+}}\left(\frac{1}{n}\sum_{i=1}^n \frac{\overline{\psi}_{q_N^{+}}(\tau_i)(B_N+\tau_i I_N)}{(\tau_i+\alpha_N^{+})(1+e_i)}\right)^{-1}, \hspace{0.2cm} k=1,\cdots,n.
$$
The limit of the convergence in \eqref{eq:convergenceqj} can be bounded as:
\begin{align*}
\frac{1}{N\int_0^{+\infty}\psi(q_N^{+}x)\nu(dx)}\tr \frac{B_N+\tau_j I_N}{\tau_j+\alpha_N^{+}}T_N &\leq (1+\max_{1\leq k\leq n}e_k)\frac{1}{N}\tr \frac{B_N+\tau_j I_N}{\tau_j+\alpha_N^{+}}\left(\frac{1}{n}\sum_{i=1}^n \frac{\overline{\psi}_{q_N^{+}}(\tau_i)(B_N+\tau_i I_N)}{\tau_i+\alpha_N^{+}}\right)^{-1} \\
&\stackrel{(a)}{\leq} \frac{1}{\int_0^{+\infty}\psi(q_N^{+}x)\nu(dx)} (1+\max_{1\leq k\leq n}e_k ) +\epsilon_{n,j}
\end{align*}
where $\max_{1\leq j\leq n}|\epsilon_{n,j}|$ converges to zero almost surely (inequality (a) being a by-product  of \eqref{eq:fundamental} in proposition \ref{prop:inequality}).
Finally. from item $iii)$ of Lemma \ref{lemma:quadratic_alpha}. we get:
$$
\frac{1}{\int_0^{+\infty}\psi(q_N^{+}x)\nu(dx)} \frac{1}{N}\tr \frac{B_N +\tau_j I_N}{\tau_j+\alpha_N^{+}}T_N \leq \frac{1}{\int_0^{+\infty}\psi(q_N^{+}x)\nu(dx) -c_N\psi_{\infty}} +\epsilon_{n,j}^{'}
$$
with $\epsilon_{n,j}^{'}$ converging to zero almost surely. Since $q_{N}^{+}$ satisfies:
$$
\int_0^{+\infty}\psi(q_N^{+}x)\nu(dx) > \frac{1}{1-c_N\psi_{\infty}} = 1 +c_N \psi_{\infty},
$$
we obtain that:
$$
\frac{q_j^{t+1}}{q_N^{+}(\tau_j+\alpha_N^{+})}<1
$$
for $n$ large enough a.s.
In order to prove the lower bound in \eqref{eq:lower_upper_bound_q}, the same reasoning as the one used in the previous lemma applies. In particular, it suffices to set $\theta>0$ and $M_\theta$ such that $\theta < \frac{\phi_{\infty}}{2(1-c_{+}\infty)}$, $\nu(M_\theta,+\infty)<\theta$ and $M_\theta \geq \|B_N\|$. Taking $q_N^{-}$ such that:
$$
\int_0^{M_\theta} \psi_N(q_N^{-}(t+M_\theta))\nu(dt) <\frac{1}{2}
$$
and setting $\alpha_N^{-}$ to the unique solution of the following equation:
$$
1=\int \frac{\int_0^{+\infty}f^{-}(x)\nu(dx)}{\int_0^{+\infty}\frac{y+t}{t+\alpha_N^{-}}f^{-}(t)\nu(dt)}F^{B_N}(dy)
$$
with $f^{-}:t\mapsto \psi_\infty1_{t\geq M_\theta} +\psi(2q_N^{-}(t+M_\theta))1_{t\leq M_\theta}$, we can establish by induction on $t$ and using the same steps as in the control of the lower bound of $\delta_i$ that:
$$
\frac{q_i}{q_N^{-}(\tau_i+\alpha_N^{-})} >1, \hspace{0.2cm} i=1,\cdots,n.
$$
\end{proof}
The determination of an interval in which lies all quantities $\delta_1,\cdots,\delta_n$ is of utmost important in that it allows us to control the quadratic forms: $\frac{\psi(\delta_j)}{N\delta_j}y_j^*\left(\frac{1}{n}\sum_{i=1}^n \frac{\psi(\delta_i)}{\delta_i}y_iy_i^*\right)^{-1}y_j$ and  $\frac{\psi(\delta_j^{\kappa})}{N\delta_j^{\kappa}}y_j^*\left(\frac{1}{n}\sum_{i=1,\tau_i\leq M_\kappa}^n \frac{\psi(\delta_i^\kappa)}{\delta_i^\kappa}y_iy_i^*\right)^{-1}y_j$, where $\delta_1,\cdots,\delta_n$ and $\delta_1^{\kappa},\cdots,\delta_n^{\kappa}$ are solutions of equations \eqref{eq:delta_kappa} and \eqref{eq:delta_j}. In particular, we have the following two lemmas which easily follows from item-$i)$ of Lemma \ref{lemma:quadratic_alpha}.
\begin{lemma}
Let Assumptions \ref{ass:regime}-\ref{ass:statistical} hold true. Then,
$$
\max_{1\leq j\leq n}\left|\frac{\psi(\delta_j)}{N\delta_j}y_j^*\left(\sum_{i=1,i\neq j}^n\frac{\psi(\delta_i)}{\delta_i}y_iy_i^*\right)^{-1}y_j-\psi(\delta_j)\right|\asto 0.
$$
\label{lemma:quadratic_form_total}
\end{lemma}
\begin{lemma}
Let $(\kappa,M_\kappa)$ be couples indexed by $\kappa$ with $0<\kappa<1$, and $M_\kappa >0$ such that $\nu(M_\kappa,\infty)<\kappa$. Then, for all $\kappa<\kappa_0$, we have:
$$
\max_{1\leq j\leq n}\left|\frac{\psi(\delta_j^{\kappa})}{N\delta_j^\kappa}y_j^*\left(\sum_{i=1,i\neq j}^n\frac{\psi(\delta_i^\kappa)}{\delta_i^\kappa}y_iy_i^*\right)^{-1}y_j-\psi(\delta_j^\kappa)\right|\asto 0,
$$
where $\delta_i^{\kappa},i=1,\cdots,n$ are defined the solutions of \eqref{eq:delta_kappa}.
\label{lemma:quadratic_form_total_kappa}
\end{lemma}
With these results at hand, we are now in position to prove $f_i= \frac{v(q_i)}{v(\delta_i)}$ satisfies:
$$
\max_{1\leq i\leq n }\left|f_i-1\right|\asto 0.
$$ 
As in \cite{couillet-pascal-2013}, we will distinguish two cases: the case where all $\tau_i$s are bounded and that of unbounded $\tau_i$. The proof is merely based on the same techniques with only some modifications and will be detailed for sake of completeness. 

\begin{paragraph}{\bf All $\tau_i$-s are Bounded}
Assume that there exists a constant $M$ such that $\tau_i\leq M$ for all $i=1,\cdots,n$. Define $f_i=\frac{v(q_i)}{v(\delta_i)}>0$. Without loss of generality, we assume that $f_1\leq \cdots \leq f_n$. We have:
\begin{align}
f_j&=\frac{v(\frac{1}{N}y_j^*\hat{C}_j^{-1}y_j)}{v(\delta_j)} \\
&=\frac{v\left(\frac{1}{N}y_j^*\left(\frac{1}{n}\sum_{i\neq j}v(q_i)y_iy_i^*\right)^{-1}y_j\right)}{v(\delta_j)} \nonumber\\
&=\frac{v\left(\frac{1}{N}y_j^*\left(\frac{1}{n}\sum_{i\neq j}f_iv(\delta_i)y_iy_i^*\right)^{-1}y_j\right)}{v(\delta_j)}\nonumber\\
&\leq \frac{v\left(\frac{1}{N}y_j^*\left(\frac{1}{n}\sum_{i\neq j}f_nv(\delta_i)y_iy_i^*\right)^{-1}y_j\right)}{v(\delta_j)}\nonumber\\
&=\frac{v\left(\frac{\delta_j}{f_n\psi(\delta_j)} \frac{\psi(\delta_j)}{N\delta_j}\left(\frac{1}{n}\sum_{i\neq j}v(\delta_i)y_iy_i^*\right)^{-1}y_j\right)}{v(\delta_j)}.\label{eq:useful_ineq}
\end{align}
In a similar way, we also have:
$$
f_1\geq \frac{v\left(\frac{\delta_j}{f_1\psi(\delta_j)} \frac{\psi(\delta_j)}{N\delta_j}\left(\frac{1}{n}\sum_{i\neq j}v(\delta_i)y_iy_i^*\right)^{-1}y_j\right)}{v(\delta_j)},
$$
From Lemma \ref{lemma:quadratic_form_total}, let $0<\epsilon_n<1$ with $\epsilon_n\downarrow0$ such that for all alrge $n$, a.s. and for all $1\leq j\leq n$:
\begin{equation}
\psi(\delta_j)-\epsilon_n <\frac{\psi(\delta_j)}{N\delta_j}y_j^*\left(\frac{1}{n}\sum_{i\neq j}v(\delta_i)y_iy_i^*\right)^{-1} y_j\leq \psi(\delta_j) +\epsilon_n.
\label{eq:borne}
\end{equation}
In particular, since $v$ is non-increasing, taking $j=n$ in \eqref{eq:useful_ineq} and applying the left-hand inequality in \eqref{eq:borne}, we obtain:
\begin{equation}
f_n < \frac{v\left(\frac{\delta_n}{f_n\psi(\delta_n)}\max\left(\psi(\delta_n)-\epsilon_n,0\right)\right)}{v(\delta_n)}.
\label{eq:f_n}
\end{equation}
Assume now that for some $\ell>0$, $f_n>1+\ell$ infinitely often. Therefore, there exists a sequence $(n)$ over which $f_{(n)}>1+\ell$ for $n$ large enough. We distinguish two cases. First, assume that $\lim\inf \delta_{(n)}=0$. There exists a sequence $(p)$ obtained from a subsequence of $(n)$ over which $\lim_{n\to+\infty} \delta_{(p)}=0$. 

From \eqref{eq:f_n}, we have:
$$
\lim_{n\to+\infty}f_{(p)} \leq \lim_{n\to+\infty} \frac{v\left(\frac{1}{f_{(p)}v(\delta_{(p)})}\max\left(\psi(\delta_{(p)})-\epsilon_{(p)},0\right)\right)}{v(\delta_{(p)})}=1.
$$
which is in contradiction with $f_{(p)}>1+\ell$. Therefore, for \eqref{eq:f_n} to hold, we must have $\lim\inf \delta_n>\delta_{\rm min}$. Since all $\tau_i$-s are bounded, $\left(\delta_n\right)_n$ is also a bounded sequence. One can thus extract a subsequence $(q)$ extracted from $(p)$ over which $\delta_{(q)}\to x>0$   and $c_N\to c$. Let $\psi_c(x)=\lim_{c_N\to c}\psi(x)$ and write \eqref{eq:f_n} in the following equivalent form:
\begin{equation}
(1-\frac{\epsilon_{(q)}}{\psi(\delta_{(q)})})\frac{\psi(\delta_{(q)})}{\psi\left(\frac{\delta_{(q)}}{f_{(q)}}\left(1-\frac{\epsilon_{(q)}}{\psi(\delta_{(q)})}\right)\right)} <1.
\label{eq:contradiction}
\end{equation}
We therefore have:
\begin{align*}
\lim_{n\to+\infty}(1-\frac{\epsilon_{(q)}}{\psi(\delta_{(q)})})\frac{\psi(\delta_{(q)})}{\psi\left(\frac{\delta_{(q)}}{f_{(q)}}\left(1-\frac{\epsilon_{(q)}}{\psi(\delta_{(q)})}\right)\right)} &\geq \lim_{n\to+\infty} (1-\frac{\epsilon_{(q)}}{\psi(\delta_{(q)})}) \frac{\psi(\delta_{(q)})}{\psi\left(\frac{\psi(\delta_{(q)})}{\ell+1}\left(1-\frac{\epsilon_{(q)}}{\psi(\delta_{(q)})}\right)\right)} \\
&=\frac{\psi_c(x)}{\psi_c((1+\ell)^{-1}x)} > 1,
\end{align*}
which is in contradiction with \eqref{eq:contradiction}.
Symmetrically, we obtain that for $\epsilon_n\downarrow0$ and for large $n$ a.s.,
$$
f_1 >\frac{v(\frac{\delta_1}{f_1\psi(\delta_1)})(\psi(\delta_1)+\epsilon_n)}{v(\delta_1)} 
$$
which is equivalent to:
$$
\frac{f_1v(\delta_1)}{v\left(\frac{\delta_1}{f_1\psi(\delta_1)}(\psi(\delta_1)+\epsilon_n)\right)}.
$$
We conclude using the same reasoning as above that for each $\ell>0$ small $f_1\geq 1-\ell$ for all large $n$. a.s. so that finally, we have:
$$
\max_{1\leq i\leq n}\left|f_i-1\right|\asto 0.
$$
The uniform boundedness of $\tau_i$ implies that of $q_i$ and $\delta_i$, thereby ensuring that:
$$
\max_{1\leq i\leq n} |v(q_i)-v(\delta_i)| \asto 0.
$$
Hence, for any $\ell>0$, arbitrarily small, we have for all large $n$, 
$$
(1-\ell)\frac{1}{n}\sum_{i=1}^n \frac{\psi(\delta_i)}{\delta_i} y_iy_i^* \preceq \frac{1}{n}\sum_{i=1}^n v(q_i)y_iy_i^* \preceq (1+\ell)\frac{1}{n}\sum_{i=1}^n \frac{\psi(\delta_i)}{\delta_i} y_iy_i^*,
$$
Since the spectral norm of $\frac{1}{n}\sum_{i=1}^n \frac{\psi(\delta_i)}{\delta_i} y_iy_i^*$ is almost surely bounded and $\ell$ is arbitrary, we conclude that:
$$
\left\|\hat{C}_N-\hat{S}_N\right\|\asto 0.
$$
\end{paragraph}
\begin{paragraph}{\bf Unbounded $\tau_i$}
We now relax the boundedness assumption on the support of the distribution of $\tau_i$. We will follow the same technique used in \cite{couillet-pascal-2013}. Similarly to \cite{couillet-pascal-2013}, let $(\kappa,M_\kappa)$ be couples indexed by $\kappa$ such that for all large $n$, we have $\nu_n\left(M_\kappa,+\infty\right)<\kappa\leq \kappa_0$, for $\kappa_0$ small enough, and $M_\kappa\geq \lim\sup \|B_N\|$. Denote by $\mathcal{C}_{\kappa}=\left\{i,\tau_i\leq M_\kappa\right\}$ with cardinality $|\mathcal{C}_\kappa|$. Then,
$$
\frac{|\mathcal{C}_\kappa|}{n}=1-\nu_n(M_\kappa,\infty) \geq 1-\kappa.
$$
In the sequel, we will differentiate the indexes in $\mathcal{C}_\kappa$ from those in $\mathcal{C}_{\kappa}^{c}$. Define $f_1^{\kappa},\cdots,f_n^{\kappa}$ as:
$$
f_i^{\kappa}=\frac{v(q_i)}{v(\delta_i^{\kappa})},
$$
where $\delta_1^{\kappa},\cdots,\delta_n^{\kappa}$ are the solutions of the system of equations \eqref{eq:delta_kappa} given in Lemma \ref{lemma:extension_boundedness}.
Let $j\in\mathcal{C}_{\kappa}$ and denote by $f_{\overline{1}}^{\kappa}=\min_{i\in\mathcal{C}_\kappa}$ and $f_{\overline{n}}^{\kappa}=\max_{i\in\mathcal{C}_\kappa}f_i^{\kappa}$. We have:
\begin{align*}
f_j^{\kappa}&=\frac{v(q_j)}{v(\delta_j^{\kappa})} \\
&=\frac{v\left(\frac{1}{N}y_j^*\left(\frac{1}{n}\sum_{i\neq j,i\in\mathcal{C}_\kappa}f_i^{\kappa}v(\delta_i^\kappa)y_iy_i^*+\frac{1}{n}\sum_{i\in\mathcal{C}_\kappa^{c}}v(q_i)y_iy_i^*\right)^{-1}y_j\right)}{v(\delta_j^{\kappa})}\\
&\leq \frac{v\left(\frac{1}{N}y_j^*\left(\frac{1}{n}\sum_{i\neq j,i\in\mathcal{C}_{\kappa}}f_{\overline{n}}^{\kappa}v(\delta_i^{\kappa})y_iy_i^*+\frac{\psi_{\infty}}{nq_N^{-}}\sum_{i\in\mathcal{C}_{\kappa}^c}\frac{y_iy_i^*}{\tau_i+\alpha_N^{-}}\right)^{-1}y_j\right)}{v(\delta_j^{\kappa})}\\
&=\frac{v\left(\frac{1}{Nf_{\overline{n}}^{\kappa}}y_j^*\left(\frac{1}{n}\sum_{i\neq j,i\in\mathcal{C}_\kappa}v(\delta_i^{\kappa})y_iy_i^*+\right)^{-1}y_j\right)}{v(\delta_j^{\kappa})},
\end{align*}
where we used in the first inequality the fact that $q_i\geq q_N^{-}(\tau_i+\alpha_N^{-})$. Since $f_{\overline{n}}^{\kappa}=\frac{v(q_{\overline{n}})}{v(\delta_{\overline{n}}^{\kappa})}=\frac{\psi(q_{\overline{n}})}{\psi(\delta_{\overline{n}}^{\kappa})}\frac{\delta_{\overline{n}}^{\kappa}}{q_{\overline{n}}}$, we obtain:
$$
\frac{v(q_N^{+}(\tau_{\overline{n}}+\alpha_N^{+}))}{v(\delta_N^{-,\kappa_0}(\tau_{\overline{n}}+\eta_N^{-,\kappa_0}))} \leq f_{\overline{n}}^{\kappa}\leq \frac{\psi(q_N^{+}(\tau_{\overline{n}}+\alpha_N^{+}))\delta_N^{+,\kappa_0}(\tau_{\overline{n}}+\eta_N^{+,\kappa_0})}{\psi(\delta_N^{-,\kappa_0}(\tau_{\overline{n}}+\eta_N^{-,\kappa_0}))q_N^{-}(\tau_{\overline{n}}+\alpha_N^{-})}.
$$
The above inequalities imply that $f_{\overline{\kappa}}$ is almost surely bounded irrespective of $\kappa$ small enough. To see that, note that if $\lim\inf \tau_{\overline{n}}=0$, the left inequality ensures that $\lim\inf f_{\overline{n}}^{\kappa}>0$ while if $\lim\sup_{n}\tau_{\overline{n}}\tau_{\overline{n}}=\infty $, the second inequality ensures that $\lim\sup f_{\overline{n}}^{\kappa}<\infty$. As a consequence, we can assume that $f_{\overline{n}}^{\kappa}>f_{-}$ for all large $n$ and for all $\kappa$ small enough. From this observation, for all large $n$, a.s. we have:
\begin{align}
f_j^{\kappa}&\leq \frac{v\left(\frac{1}{Nf_{\overline{n}}^{\kappa}}y_j^*\left(\frac{1}{n}\sum_{i\neq j, i\in\mathcal{C}_\kappa}v(\delta_i^{\kappa})y_iy_i^*+\frac{\psi_{\infty}}{nq_{-}f_{-}}\sum_{i\in \mathcal{C}_{\kappa}^{c}}\frac{1}{\tau_i+\alpha_N^{-}}y_iy_i^*\right)^{-1}y_j\right)}{v(\delta_j^{\kappa})}\nonumber\\
&=\frac{v\left(\frac{\delta_j^{\kappa}}{\psi(\delta_j^{\kappa})f_{\overline{n}}^{\kappa}}\left[\frac{\psi(\delta_j^{\kappa})}{N\delta_j^{\kappa}}y_j^*\left(\frac{1}{n}\sum_{i\neq j,i\in\mathcal{C}_\kappa}v(\delta_i^{\kappa})y_iy_i^*\right)^{-1}y_j+w_{j,n}\right]\right)}{v(\delta_j^{\kappa})}, \label{eq:e_jc}
\end{align}
where we defined similarly to \cite{couillet-pascal-2013} $w_{j,n}$ as:
$$
w_{j,n}=\frac{\psi(\delta_j^{\kappa})}{N\delta_j^{\kappa}}y_j^*\left(D_{\kappa,j}+C_{\kappa}\right)^{-1}y_j-\frac{\psi(\delta_j^{\kappa})}{N\delta_j^{\kappa}}y_j^*D_{\kappa,j}^{-1}y_j
$$
with 
$$
D_{\kappa,j}\triangleq \frac{1}{n}\sum_{i\in\mathcal{C_\kappa,i\neq j}} v(\delta_i^{\kappa})y_iy_i^*, \hspace{0.2cm} C_{\kappa}=\frac{\psi_{\infty}}{nq_N^{-}f_{-}}\sum_{i\in\mathcal{C}_\kappa^{c}} \frac{1}{\tau_i+\alpha_N^{-}}y_iy_i^*
$$
Using the resolvent identity $D^{-1}-F^{-1}=D^{-1}(F-D)F^{-1}$ (for any invertible matrices $D$ and $F$) along with Cauchy-Schwartz inequality, we obtain:
$$
|w_{n,j}|\leq \sqrt{\frac{\psi(\delta_j^{\kappa})}{N\delta_j^{\kappa}}y_j^*(D_{\kappa,j}+C_{\kappa})^{-1}C_\kappa(D_{\kappa_j}+C_\kappa)^{-1}y_j}\sqrt{\frac{\psi(\delta_j^{\kappa})}{N\delta_j^{\kappa}}y_j^*D_{\kappa,j}^{-1}C_\kappa D_{\kappa,j}^{-1}y_j}.
$$
Note that for $\kappa$ small enough, matrix $D_{\kappa,j}$ is invertible. Besides, from assumption \ref{ass:measure}$-ii)$, for $\kappa$ small enough and for enough large $n$, $\nu_n(\left[m,M_\kappa\right]) \geq c_{+}$. Using Lemma \ref{lemma:smallest_eigenvalue} in Appendix \ref{app:technical}, we conclude that there exists $C_1$ such that $\min_j\lambda_1(D_{\kappa,j}) \geq C_1$. Since matrix $C_\kappa$ has a bounded spectral norm, Theorem \ref{app:convergence_quadratic} in Appendix \ref{app:technical} along with the rank-1 perturbation Lemma \cite[Lemma 2.6]{SilBai95} yields:
$$
\max_{j\in\mathcal{C}_\kappa}\left|\frac{\psi(\delta_j^{\kappa})}{N\delta_j^{\kappa}}y_j^*D_{\kappa,j}^{-1}C_\kappa D_{\kappa,j}^{-1}y_j-\frac{\psi(\delta_j^{\kappa})}{N\delta_j^{\kappa}}\tr (B_N+\tau_j I_N) D_\kappa^{-1}C_{\kappa}D_{\kappa}^{-1}\right|\asto 0,
$$
where $D_\kappa=D_{\kappa,j}+\frac{1}{n}v(\delta_n^{\kappa})y_jy_j^*$.
From $|\tr XY|\leq \|X\|\tr Y$ for positive definite $Y$, we have:
$$
\frac{\psi(\delta_j^{\kappa})}{N\delta_j^{\kappa}}\tr (B_N+\tau_j I_N) D_\kappa^{-1}C_{\kappa}D_{\kappa}^{-1}\leq \frac{\|B_N+\tau_j I_N\|}{C_1^2\delta_N^{+,\kappa_0}\psi_{\infty}(\tau_j+\eta_N^{-,\kappa_0})}\frac{1}{N}\sum_{i\in\mathcal{C}_\kappa^c} \frac{y_i^*y_i}{\tau_i+\alpha_N^{-}}
$$
where
$$
\frac{1}{n}\sum_{i\in\mathcal{C}_\kappa^c} \frac{1}{N}\frac{y_i^*y_i}{\tau_i+\alpha_N^{-}} -\frac{1}{n} \sum_{i\in\mathcal{C}_\kappa^c}\frac{1}{N}\tr\frac{B_N+\tau_i I_N}{\tau_i+\alpha_N^{-}}\asto 0.
$$
Since
$
\frac{1}{n}\sum_{i\in\mathcal{C}_\kappa^c} \frac{1}{N}\tr \frac{B_N+\tau_i I_N}{\tau_i+\alpha_N^{-}} \leq\frac{\|B_N\|+M_{\kappa_0}}{\alpha_N^{-}+M_{\kappa_0}}\frac{\nu_n(M_\kappa,\infty)}{n} 
$ for all $\kappa \leq \kappa_0$ and for all large $n$ a.s., we have:
$$
\max_{j\in\mathcal{C}_\kappa}\frac{\psi(\delta_j^{\kappa})}{N\delta_j^{\kappa}}y_j^*D_{\kappa,j}^{-1}C_\kappa D_{\kappa,j}^{-1}y_j \leq K_1 \nu_n(M_\kappa,+\infty),
$$
where $K_1$ is a constant that does not depend on $\kappa\leq \kappa_0$. In the same way, we can control the term $\frac{\psi(\delta_j^\kappa)}{N\delta_j^\kappa}y_j^*(D_{\kappa,j}+C_\kappa)^{-1}C_\kappa (D_{\kappa,j}+C_\kappa)^{-1}y_j$. Finally, we conclude that:
\begin{equation}
\max_{j\in\mathcal{C}_\kappa} |w_{j,n}|\leq K\nu_n(M_\kappa,+\infty)
\label{eq:max_w}
\end{equation}
for some constant $K$ independent of $\kappa\leq \kappa_0$.
Quantities $w_{j,n}$ being controlled for $j\in\mathcal{C}_\kappa$, we can now proceed in a similar way as in the case of the bounded $\tau_i$ case.   Lemma \ref{lemma:quadratic_form_total_kappe} implies that for any fixed $\kappa>0$ there exists a sequence $\epsilon_n^{\kappa}\downarrow 0$ such that a.s. for $n$ large enough,
\begin{equation}
\max_{j\in\mathcal{C}_\kappa} \left|\frac{\psi(\delta_j^{\kappa})}{N\delta_j^{\kappa}}y_j^*\left(\frac{1}{n}\sum_{i\in\mathcal{C}_\kappa}\frac{\psi(\delta_i^{\kappa})}{\delta_i^{\kappa}}y_iy_i^*\right)^{-1}y_j - {\psi(\delta_j^{\kappa})}\right|\leq \epsilon_n^{\kappa}.
\label{eq:max_quadratic}
\end{equation}
Combining \eqref{eq:e_jc}, \eqref{eq:max_w} and \eqref{eq:max_quadratic}, we then have for all large $n$ a.s. and for all $j\in\mathcal{C}_\kappa$,
$$
f_{j}^{\kappa} \leq \frac{v\left[\frac{\delta_j^{\kappa}}{\psi(\delta_j^{\kappa})f_{\overline{n}}^{\kappa}}\max\left(\psi(\delta_j^{\kappa})-\epsilon_n^{\kappa}-K\nu_n(M_\kappa,\infty),0\right)\right]}{v(\delta_j^{\kappa})}
$$
which for $j=\overline{n}$ becomes:
\begin{equation}
f_{\overline{n}}^{\kappa} \leq \frac{v\left[\frac{\delta_{\overline{n}}^{\kappa}}{\psi(\delta_{\overline{n}}^{\kappa})f_{\overline{n}}^{\kappa}}\max\left(\psi(\delta_{\overline{n}}^{\kappa})-\epsilon_n^{\kappa}-K\nu_n(M_\kappa,\infty),0\right)\right]}{v(\delta_{\overline{n}}^{\kappa})}.
\label{eq:unbounded_case}
\end{equation}
Assume that $\lim\sup f_{\overline{n}} >1+\ell$ for some $\ell >0$. Let us restrict the sequence $f_{\overline{n}}$ to those indexes for which $f_{\overline{n}}>1+\ell$. Similar to the case of bounded $\tau_i$, we can see that \eqref{eq:unbounded_case} implies that $\lim\inf \delta_{\overline{n}}^{\kappa} >\delta_{\rm min}$, a bound which can be chosen independent of $\kappa\leq \kappa_0$. In effect, from \eqref{eq:unbounded_case}. we have:
$$
f_{\overline{n}}^{\kappa} \leq \frac{v(0)}{v(\delta_{\overline{n}}^{\kappa})}
$$
which is equivalent to:
$$
v(\delta_{\overline{n}}^{\kappa}) \leq \frac{v(0)}{\ell+1},
$$
or also:
$$
\delta_{\overline{n}}^{\kappa} \geq v^{-1}\left(\frac{v(0)}{\ell+1}\right)
$$
Using the definition of $\psi$, \eqref{eq:unbounded_case} reads for $\kappa$ sufficiently small:
$$
\frac{\psi\left(\frac{\delta_{\overline{n}}^{\kappa}}{f_{\overline{n}}^{\kappa}}\left(1-\frac{\epsilon_n^{\kappa}}{\psi(\delta_{\overline{n}}^{\kappa})-\frac{K\nu_n(M_\kappa,\infty)}{\psi(\delta_{\overline{n}}^{\kappa})}}\right)\right)}{\psi(\delta_{\overline{n}}^{\kappa})\left(1-\frac{\epsilon_n^{\kappa}}{\psi(\delta_{\overline{n}}^{\kappa})}-\frac{K\nu_n(M_\kappa),\infty}{\psi(\delta_{\overline{n}}^{\kappa})}\right)} \geq 1
$$
or also for $n$ large enough:
$$
\frac{\psi(\delta_{\overline{n}}^{\kappa})\left(1-\frac{\epsilon_n^{\kappa}}{\psi(\delta_{\overline{n}}^{\kappa})}-\frac{K\nu_n(M_\kappa),\infty}{\psi(\delta_{\overline{n}}^{\kappa})}\right)}{\psi((f_{\overline{n}}^{\kappa})^{-1}\delta_{\overline{n}}^{\kappa})} \leq 1.
$$
Hence,
$$
\frac{\psi(\delta_{\overline{n}}^{\kappa})-\psi\left((f_{\overline{n}}^{\kappa})^{-1}\right)
}{\psi\left((f_{\overline{n}}^{\kappa})^{-1}\right)}\left(1-\frac{\epsilon_{n}^{\kappa}}{\psi(\delta_{\overline{n}^{\kappa}})}-\frac{K\nu_n(M_\kappa,\infty)}
{\psi(\delta_{\overline{n}}^{\kappa})}\right)\leq \frac{\epsilon_n^{\kappa}}{\psi(\delta_{\overline{n}}^{\kappa})} +\frac{K\nu_n(M_\kappa,\infty)}{\psi(\delta_{\overline{n}}^{\kappa})}.
$$
or equivalently:
$$
\frac{\psi(\delta_{\overline{n}}^{*,\kappa})-\psi\left(\left(f_{\overline{n}}^{*,\kappa}\right)^{-1}\delta_{\overline{n}}^{\kappa}\right)}{\epsilon_n^{\kappa}+K\nu_n(M_\kappa,+\infty)}\leq \frac{\psi\left(\left(f_{\overline{n}}^{\kappa}\right)^{-1}\delta_{\overline{n}}^{*,\kappa}\right)}{\psi(\delta_{\overline{n}}^{*,\kappa})\left(1-\frac{\epsilon_n^{\kappa}}{\psi(\delta_{\overline{n}}^{*,\kappa})}-\frac{K\nu_n(M_\kappa,+\infty)}{\psi(\delta_{\overline{n}}^{*,\kappa})}\right)}.
$$
Since $f_{\overline{n}}^{\kappa}>1+\ell$, $\psi\left((f_{\overline{n}}^{\kappa})^{-1}\delta_{\overline{n}}^{\kappa}\right) <\psi(\delta_{\overline{n}}^{\kappa})$. Therefore, for $\kappa$ chosen sufficiently small so that: $1-\frac{\epsilon_n^{\kappa}}{\psi(\delta_{\overline{n}}^{\kappa})}-\frac{K\nu_n(M_\kappa,\infty)}{\psi(\delta_{\overline{n}}^{\kappa})}\leq \frac{1}{2}$, we have:
\begin{equation}
\frac{\psi(\delta_{\overline{n}}^{*,\kappa})-\psi\left(\left(f_{\overline{n}}^{*,\kappa}\right)^{-1}\delta_{\overline{n}}^{\kappa}\right)}{\epsilon_n^{\kappa}+K\nu_n(M_\kappa,+\infty)} \leq 2.
\label{eq:limit_taken}
\end{equation}
Since $\delta_{\overline{n}}^{\kappa}$ belongs to the interval $\left[\delta_{\rm min},\lim\sup\delta_N^{+}(M_\kappa)+\eta_N^{+}\right]$ for $n$ large enough, taking the limit of \eqref{eq:limit_taken} over some  subsequences  over which $\delta_{\overline{n}}^{\kappa}\to x_{\kappa}\in\left[\delta_{\rm min},\lim\sup\delta_N^{+}(M_\kappa)+\eta_N^{+}\right] $, $c_N\to c$ and $\nu_n(M_\kappa,\infty)$ converges, we obtain:
\begin{equation}
\frac{\psi_c(x_\kappa)-\psi_c(\frac{x_\kappa}{1+\ell})}{\lim_n\nu_n(M_\kappa,\infty)} \leq 2,
\label{eq:contradiction_kappa}
\end{equation}
where $\psi_c=\lim_{c_N\to c}\psi_N$. We now operate on $\kappa$.  If $\lim\sup_{\kappa\to 0} x_\kappa <\infty$, the left-hand side of \eqref{eq:contradiction_kappa} goes to $+\infty$ as $\kappa\to 0$ so that starting from $\kappa$ sufficiently small and taking the limit over $n$ on the considered subsequence raises a contradiction. If instead $\lim\sup_{\kappa\to 0}x_{\kappa}=+\infty$, then since $x_\kappa \leq 2\lim\sup \delta_N^{+}M_\kappa$, we have:
$$
\frac{\psi_c(x_\kappa)-\psi_c(\frac{x_\kappa}{1+\ell})}{\lim_n\nu_n(M_\kappa,\infty)} \geq \frac{\psi_c(x_\kappa)-\psi_c(\frac{x_\kappa}{1+\ell})}{\lim_n\nu_n(\frac{x_\kappa}{\lim\sup \delta_N^{+}},\infty)}.
$$
Let $y_\kappa=g_c^{-1}(x_\kappa)$ with $g_c(x)=\frac{x}{1-c\phi(x)}$. Recall that $\psi_c=\frac{\phi\circ g_c^{-1}}{1-c\phi\circ g_c^{-1}}$.
Then,
\begin{align*}
\frac{\psi_c(x_\kappa)-\psi_c(\frac{x_\kappa}{\ell+1})}{\lim_n \nu_n\left(\frac{x_\kappa}{\lim\sup\delta_N^{+}},\infty\right)} &=\frac{\phi\circ g_c^{-1}(x_\kappa)-\phi\circ g_c^{-1}(\frac{x_\kappa}{\ell+1})}{(1-c\phi\circ g_c^{-1}(x_\kappa))(1-c\phi\circ g_c^{-1}(\frac{x_\kappa}{\ell+1}))\lim_n \nu_n\left(\frac{x_\kappa}{\lim\sup\delta_N^{+}},\infty\right)} \\
&\geq \frac{\phi(y_\kappa)-\phi\circ g_c^{-1}(\frac{g_c(y_\kappa)}{\ell+1})}{\lim_n \nu_n\left(\frac{x_\kappa}{\lim\sup\delta_N^{+}},\infty\right)}\\
&\geq \frac{\phi(y_\kappa)-\phi\circ g_c^{-1}(\frac{y_\kappa}{(\ell+1)(1-c_{+}\phi_{\infty})})}{\lim_n \nu_n\left(\frac{x_\kappa}{\lim\sup\delta_N^{+}},\infty\right)}.
\end{align*}
Since $y_\kappa\to\infty$ as $x_\kappa \to \infty$, from Assumption \ref{ass:unbounded}, the right-hand side must go to $\infty$ as $x_\kappa \to \infty$. Therefore, taking $\kappa$ sufficiently small and then consider the limist over $n$ on the subsequence under consideration raises a contradiction. Consequently, we must have $\lim\sup f_{\overline{n}}^{\kappa} \leq 1+\ell$ a.s. A similar reasoning allows to show that $\lim\inf f_{\overline{1}}^{\kappa}\geq 1-\ell$ a.s. for any given $\ell>0$. We conclude thus:
$$
\max_{j\in\mathcal{C}_\kappa}\left|f_{j}^{\kappa}-1\right|\asto 0.
$$
We will now deal with $f_j^{\kappa}$ for $j\in\mathcal{C}_{\kappa}^{c}$. Recall that $f_j$ is given by:
\begin{align*}
f_j&=\frac{v\left(\frac{1}{N}y_j^* \left(\frac{1}{n}\sum_{i\in\mathcal{C}_\kappa}f_j^{\kappa} v(\delta_j^{\kappa})y_iy_i^*+\frac{1}{n}\sum_{i\in\mathcal{C}^{c},i\neq j} v(q_i)y_iy_i^*\right)^{-1}y_j\right)}{v(\delta_j^{\kappa})}\\
&=\frac{v\left(\frac{\delta_j^{\kappa}}{\psi(\delta_j^{\kappa})}\frac{\psi(\delta_j^{\kappa})}{\delta_j^{\kappa}}\frac{1}{N}y_j^* \left(\frac{1}{n}\sum_{i\in\mathcal{C}_\kappa}f_j^{\kappa} v(\delta_j^{\kappa})y_iy_i^*+\frac{1}{n}\sum_{i\in\mathcal{C}^{c},i\neq j} v(q_i)y_iy_i^*\right)^{-1}y_j\right)}{v(\delta_j^{\kappa})}\\
&=\frac{v\left(\frac{\delta_j^{\kappa}}{\psi(\delta_j^{\kappa})}\left[\frac{\psi(\delta_j^{\kappa})}{\delta_j^{\kappa}}\frac{1}{N}y_j^* \left(\frac{1}{n}\sum_{i\in\mathcal{C}_\kappa}f_i^{\kappa} v(\delta_i^{\kappa})y_iy_i^*\right)^{-1}y_j+\tilde{w}_{j,n}\right]\right)}{v(\delta_j^{\kappa})},
\end{align*}
where
$$
\tilde{w}_{j,n}=\frac{\psi(\delta_j^{\kappa})}{N\delta_j^{\kappa}} y_j^{*}\left(\tilde{D}_{\kappa}+\tilde{C}_{\kappa,j}\right)^{-1}y_j - \frac{\psi(\delta_j^{\kappa})}{N\delta_j^{\kappa}} y_j^{*}\tilde{D}_{\kappa}^{-1} y_j
$$ 
with:
\begin{align*}
\tilde{D}_\kappa&=\frac{1}{n}\sum_{i\in\mathcal{C}_\kappa}v(\delta_i^{\kappa})y_iy_i^* \\
\tilde{C}_{\kappa,j}&=\frac{1}{n}\sum_{\substack{i\in\mathcal{C}_\kappa\\ i\neq j}}v(\delta_i^{\kappa})y_iy_i^*.
\end{align*}
Using the same reasnoning as with $w_{j,n}$ we can show that for $\kappa$ sufficiently small and $n$ large enough,
$$
\max_{j\in\mathcal{C}_\kappa^{c}} |\tilde{w}_{j,n}| \leq K \nu_n(M_\kappa,\infty) \leq K \kappa
$$
with $K$ independent of $\kappa\leq \kappa_0$. On the other hand, since $\max_{i\in \mathcal{C}_\kappa^{c}} |f_i^\kappa -1|\asto 0$, we have:
$$
\frac{\psi(\delta_j^{\kappa})}{\delta_{j}^{\kappa}}\frac{1}{N}y_j^{*}\left(\frac{1}{n}\sum_{i\in\mathcal{C}_\kappa}f_i^{\kappa}v(\delta_i^{\kappa})y_iy_i^*\right)^{-1}y_j-\frac{\psi(\delta_j^{\kappa})}{\delta_{j}^{\kappa}}\frac{1}{N}y_j^{*}\left(\frac{1}{n}\sum_{i\in\mathcal{C}_\kappa}v(\delta_i^{\kappa})y_iy_i^*\right)^{-1}y_j\stackrel{a.s.}{\to}0.
$$
As a consequence, for $\kappa$ sufficiently small and $n$ large enough:
$$
\max_{j\in\mathcal{C}_\kappa^{c}}\left|\frac{\psi(\delta_j^{\kappa})q_j}{\delta_j^{\kappa}}-\psi(\delta_j^{\kappa})\right| \leq \kappa^{'},
$$
where $\lim_{\kappa\to 0}\kappa^{'}=0$. Now, write $f_j^\kappa{}$ as:
$$
f_j=\frac{\psi\left(\frac{\delta_j^{\kappa}}{\psi(\delta_j^{\kappa})\frac{\psi(\delta_j^{\kappa})}{\delta_j^{\kappa}}q_j}\right)}{\frac{\psi(\delta_j^{\kappa})}{\delta_j^{\kappa}}q_j}.
$$
Then, one can easily note that:
$$
\lim_{\kappa\to 0}\lim\sup_n \max_{j\in\mathcal{C}_\kappa^{c}}\left\{\left|f_{j}^{\kappa}-1\right|\right\}\to 0.
$$
Combining the results for $j\in\mathcal{C}_\kappa$ and $j\in\mathcal{C}_{\kappa}^{c}$, we conclude that for each $\ell>0$, there exists $\kappa>0$ small enough such that a.s.,
$$
(1-\ell)\frac{1}{n}\sum_{i=1}^n \frac{\psi(\delta_i^{\kappa})}{\delta_i^{\kappa}}y_iy_i^{*} \preceq \frac{1}{n}\sum_{i=1}^n v(q_i)y_iy_i^{*}\preceq (1+\ell)\frac{1}{n}\sum_{i=1}^n \frac{\psi(\delta_i^{\kappa})}{\delta_i^{\kappa}}y_iy_i^*.$$
It remains thus to show that for each $\varepsilon >0$, there exists $\kappa_0$ such that for any $\kappa\leq \kappa_0$ and all large $n$, 
$$
\max_j\left|1-\frac{\delta_j}{\delta_j^{\kappa}}\right| \leq \varepsilon.
$$
Recall that $(\delta_1^{\kappa},\cdots,\delta_n^{\kappa})$ are given by:
$$
(\delta_1^{\kappa},\cdots,\delta_n^{\kappa})=\lim_{t\to\infty} (\delta_1^{\kappa}(t),\cdots,\delta_n^{\kappa}(t))
$$
with $\delta_1^{\kappa}(0),\cdots,\delta_n^{\kappa}(0)$ are arbitrary and:
$$
\delta_j^{\kappa}(t+1)=\frac{1}{N}\tr (B_N+\tau_j I_N)\left(\frac{1}{n}\sum_{\tau_i\leq M_\kappa}\frac{\phi\circ g_N^{-1}(\delta_i^{\kappa}(t))}{\delta_i^{\kappa}(t)}(B_N+\tau_i I_N)\right)^{-1},
$$
where we used the relation $\frac{\psi_N}{1+c_N\psi_N}=\phi\circ g_N^{-1}$. Set for $t=0$, $\delta_i^{\kappa}=\delta_i,i=1,\cdots,n$. We will prove by induction on $t$ that $\delta_j\leq \delta_j^{\kappa}(t)$ for all $j=1,\cdots,n$, thereby showing that $\delta \leq \delta_j^{\kappa}$. 
Obviously, the desired result holds for $t=0$. Assume now that for all $t\leq k$, $\delta_j^{\kappa}(t)\geq \delta_j$, and let us show that  $\delta_j^{\kappa}(k+1)\geq \delta_j$. Since $x\mapsto \frac{\phi\circ g_N^{-1}(x)}{x}$ is non-increasing and $\delta_i^{\kappa}(k)\geq \delta_i$, we have:
$$
 \frac{\phi\circ g_N^{-1}(\delta_i^{\kappa}(k))}{\delta_i^{\kappa}(k)} \leq \frac{\phi\circ g_N^{-1}(\delta_i)}{\delta_i}.
$$
Hence,
\begin{align*}
\delta_j^{\kappa}(k+1)&=\frac{1}{N}\tr (B_N+\tau_j I_N)\left(\frac{1}{n}\sum_{\tau_i\leq M_\kappa} \frac{\phi\circ g_N^{-1}(\delta_i^{\kappa}(k))}{\delta_i^{\kappa}(k)} (B_N+\tau_i I_N)\right)^{-1} \\
&\geq \frac{1}{N}\tr (B_N+\tau_j I_N)\left(\frac{1}{n}\sum_{i=1}^n \frac{\phi\circ g_N^{-1}(\delta_i)}{\delta_i} (B_N+\tau_i I_N)\right)^{-1}\\
&=\delta_j.
\end{align*}
We are now in position to control the convergence of $\max_{1\leq j\leq n} \left|1-\frac{\delta_j}{\delta_j^{\kappa}}\right|$ as $\kappa\to 0$. 
In particular, we recall that we need to prove that for each $\varepsilon >0$, there exists $\kappa_0$ such that: 
$$
\max_j\left|1-\frac{\delta_j}{\delta_j^{\kappa}}\right| \leq \varepsilon.
$$
To this end, define the maps $T_N$, $T_N^{M}$ as: 
%note first from the fixed equations satisfied by quantities $\delta_1,\cdots,\delta_n$ and $\delta_1^{\kappa},\cdots,\delta_n^{\kappa}$ that:
%\begin{equation}
%0=\frac{1}{N}\tr \frac{B_N+\tau_j I_N}{\delta_j}T_N(\delta_1,\cdots,\delta_n)-\frac{1}{N}\tr \frac{B_N+\tau_j I_N}{\delta_j^{\kappa}}T_N^{\kappa}(\delta_1^{\kappa},\cdots,\delta_n^{\kappa})
%\label{eq:zero_equation}
%\end{equation}
%where:
 $$T_N: (x_1,\cdots,x_n)\mapsto \left(\frac{1}{n}\sum_{i=1}^n \frac{\phi\circ g_N^{-1}(x_i)}{x_i}\left(B_N+\tau_i I_N\right)\right)^{-1}$$
and
$$T_N^{M}: (x_1,\cdots,x_n)\mapsto \left(\frac{1}{n}\sum_{\substack{i=1 \\\tau_i\leq M}}^n \frac{\phi\circ g_N^{-1}(x_i)}{x_i}\left(B_N+\tau_i I_N\right)\right)^{-1}.$$
From Lemma \ref{lemma:extension_boundedness} and \ref{lemma:control_boundedness}, it is easy to see that the spectral norms of $T_N(\delta_1,\cdots,\delta_n)$ and $T_N^{M_\kappa}(\delta_1^{\kappa},\cdots,\delta_n^{\kappa})$ are uniformly bounded. Note that:
$$
\left\|T_N(\delta_1,\cdots,\delta_n)-T_N^{M}(\delta_1,\cdots,\delta_n)\right\| \leq \left\|T_N^M(\delta_1,\cdots,\delta_n)\right\|^2\phi_{\infty}\frac{\limsup\|B_N\|}{\lim\inf\delta_N^{-}\lim\inf \eta_N^{-}}\nu_n(M,\infty).
$$
and $$
\left\|T_N^{M_\kappa}(\delta_1^{\kappa},\cdots,\delta_n^{\kappa})-T_N^{M}(\delta_1^{\kappa},\cdots,\delta_n^{\kappa})\right\| \leq \left\|T_N^M(\delta_1^{\kappa},\cdots,\delta_n^{\kappa})\right\|^2\phi_{\infty}\frac{\limsup\|B_N\|}{\lim\inf\delta_N^{-,\kappa}\lim\inf \eta_N^{-,\kappa}}\nu_n(M,M_\kappa)
$$
for any $M_\kappa \geq M$.
Setting $M$ large enough so that $\lim\inf\nu_n(m,M)>0$, we get:
$$
\max\left(\left\|T_N^{M}(\delta_1,\cdots,\delta_n)\right\|,\left\|T_N^{M}(\delta_1^{\kappa},\cdots,\delta_n^{\kappa})\right\|\right)\leq \frac{\max\left(\limsup(m+\eta_N^{+})\delta_N^{+},\limsup (m+\eta_N^{+,\kappa})\delta_N^{+,\kappa} \right)}{m\phi\circ g_N^{-1}(m)(1-\lim\sup\nu_n(M,\infty))}.
$$
Therefore, one can fix $M$ sufficiently large in such a way that:
\begin{equation}
\limsup_N\left\|T_N(\delta_1,\cdots,\delta_n)-T_N^{M}(\delta_1,\cdots,\delta_n)\right\| \leq \frac{\varepsilon}{3}
\label{eq:fromM_1}
\end{equation}
and
\begin{equation}
\limsup_N\left\|T_N^{M_\kappa}(\delta_1^{\kappa},\cdots,\delta_n^{\kappa})-T_N^{M}(\delta_1^{\kappa},\cdots,\delta_n^{\kappa})\right\| \leq \frac{\varepsilon}{3}.
\label{eq:fromM_2}
\end{equation}
With this value of $M$ at hand, we will now prove that:
$$
\lim_{\kappa \to 0} \lim\sup_N\left\|T_N^{M}(\delta_1^{\kappa},\cdots,\delta_n^{\kappa})- T_N^{M}(\delta_1,\cdots,\delta_n)\right\|= 0.
$$
To this end, we will work out the differences $\frac{\delta_j^{\kappa}-\delta_j}{\delta_j}$. We have:
\begin{align*}
&\frac{\delta_j^{\kappa}-\delta_j}{\delta_j^{\kappa}}=\\
&\frac{1}{N}\tr \frac{B_N+\tau_j I_N}{\delta_j^{\kappa}}T_N^{M_\kappa}(\delta_1^{\kappa},\cdots,\delta_n^{\kappa})\left[\frac{1}{n}\sum_{\tau_i\leq M_\kappa} \left(\frac{\phi\circ g_N^{-1}(\delta_i)}{\delta_i}-\frac{\phi\circ g_N^{-1}(\delta_i^{\kappa})}{\delta_i^{\kappa}}\right)(B_N+\tau_i I_N)\right]T_N(\delta_1,\cdots,\delta_n)\\
&+\frac{1}{N}\tr \frac{B_N+\tau_j I_N}{\delta_j^{\kappa}}T_N^{M_\kappa}(\delta_1^{\kappa},\cdots,\delta_n^{\kappa})\left[\frac{1}{n}\sum_{\tau_i\geq M_\kappa} \frac{\phi\circ g_N^{-1}(\delta_i)}{\delta_i}\right]T_N(\delta_1,\cdots,\delta_n)\\
&=\frac{1}{N}\tr \frac{B_N+\tau_j I_N}{\delta_j^{\kappa}}T_N^{M_\kappa}(\delta_1^{\kappa},\cdots,\delta_n^{\kappa})\left[\frac{1}{n}\sum_{\tau_i\leq M_\kappa} \left(\frac{\phi\circ g_N^{-1}(\delta_i)-\phi\circ g_N^{-1}(\delta_i^{\kappa})}{\delta_i^{\kappa}}\right)(B_N+\tau_i I_N)\right]T_N(\delta_1,\cdots,\delta_n)\\
&+\frac{1}{N}\tr \frac{B_N+\tau_j I_N}{\delta_j^{\kappa}}T_N^{M_\kappa}(\delta_1^{\kappa},\cdots,\delta_n^{\kappa})\left[\frac{1}{n}\sum_{\tau_i\leq M_\kappa} \phi\circ g_N^{-1}(\delta_i)\left[\frac{1}{\delta_i}-\frac{1}{\delta_i^{\kappa}}\right](B_N+\tau_i I_N)\right]T_N(\delta_1,\cdots,\delta_n)\\
&+\frac{1}{N}\tr \frac{B_N+\tau_j I_N}{\delta_j^{\kappa}}T_N^{M_\kappa}(\delta_1^{\kappa},\cdots,\delta_n^{\kappa})\left[\frac{1}{n}\sum_{\tau_i\geq M_\kappa} \frac{\phi\circ g_N^{-1}(\delta_i)}{\delta_i}(B_N+\tau_i I_N)\right]T_N(\delta_1,\cdots,\delta_n)\\
&\triangleq \alpha_{1,j}+\alpha_{2,j}+\alpha_{3,j}.
\end{align*}
Note that $\alpha_2$ can be bounded as:
$$
\alpha_{2,j} \leq \max_{i}\left|\frac{\delta_i^{\kappa}-\delta_i}{\delta_i^{\kappa}}\right|.
$$
Let $j_0$ be the index of the maximum element of $\frac{\delta_i^{\kappa}-\delta_i}{\delta_i^{\kappa}}, i=1,\cdots,n$. Therefore:
$$
-\alpha_{1,j_0} \leq \alpha_{3,j_0}
$$ 
or equivalently,
\begin{align*}
&\frac{1}{N}\tr \frac{B_N+\tau_{j_0} I_N}{\delta_{j_0}^{\kappa}}T_N^{M_\kappa}(\delta_1^{\kappa},\cdots,\delta_n^{\kappa})\left[\frac{1}{n}\sum_{\tau_i\leq M_\kappa} \left(\frac{\phi\circ g_N^{-1}(\delta_i^{\kappa})-\phi\circ g_N^{-1}(\delta_i)}{\delta_i^{\kappa}}\right)(B_N+\tau_i I_N)\right]T_N(\delta_1,\cdots,\delta_n)\\
&\leq 
\frac{1}{N}\tr \frac{B_N+\tau_{j_0} I_N}{\delta_{j_0}^{\kappa}}T_N^{M_\kappa}(\delta_1^{\kappa},\cdots,\delta_n^{\kappa})\left[\frac{1}{n}\sum_{\tau_i\geq M_\kappa} \frac{\phi\circ g_N^{-1}(\delta_i)}{\delta_i}(B_N+\tau_i I_N)\right]T_N(\delta_1,\cdots,\delta_n)
\end{align*}
Hence,% using the relation $\lambda_1(A)\frac{1}{N} \tr B\leq \frac{1}{N}\tr AB \leq\lambda_N(A)\frac{1}{N}\tr B$ for $A$ and $B$ positive definite, we obtain:
\begin{align*}
& \frac{1}{N}\tr \frac{B_N+\tau_{j_0} I_N}{\delta_{j_0}^{\kappa}}T_N^{M_\kappa}(\delta_1^{\kappa},\cdots,\delta_n^{\kappa})\left[\frac{1}{n}\sum_{\tau_i\leq M} \left(\frac{\phi\circ g_N^{-1}(\delta_i^{\kappa})-\phi\circ g_N^{-1}(\delta_i)}{\delta_i^{\kappa}}\right)(B_N+\tau_i I_N)\right]T_N(\delta_1,\cdots,\delta_n)\\
\\
&\leq \left\|T_N(\delta_1,\cdots,\delta_n)\right\| \frac{\phi_{\infty}\nu_n(M_\kappa,\infty)\lim\sup\|B_N\|}{\lim\inf \eta_N^{-}\delta_N^{-}}
\end{align*}
The right-hand side in the above inequality converges to zero as $\kappa\to0$. This is possible if and only if:
\begin{equation}
\frac{1}{n}\sum_{\tau_i\leq M} \left(\frac{\phi\circ g_N^{-1}(\delta_i^{\kappa})-\phi\circ g_N^{-1}(\delta_i)}{\delta_i^{\kappa}}\right)(B_N+\tau_i I_N)\xrightarrow[\kappa\to 0]{}0.
\label{eq:convergence_1}
\end{equation}
Function $x\mapsto \phi\circ g_N^{-1}(x)$ is continuously differentiable. Therefore, by the mean value theorem, 
$$
\frac{\phi\circ g_N^{-1}(\delta_i^{\kappa})-\phi\circ g_N^{-1}(\delta_i)}{\delta_i^{\kappa}}=\left(\phi \circ g_N^{-1}\right)'(\xi_{i}^{\kappa})\frac{\delta_i^{\kappa}-\delta_i}{\delta_i^{\kappa}},
$$ 
where $\left(\phi \circ g_N^{-1}\right)'$ denotes the derivative of $\phi\circ g_N^{-1}$ and $\xi_i^{\kappa}\in \left[\delta_i,\delta_i^{\kappa}\right]$. Now, since for $n$ large enough, $\min_{i,\tau_i\leq M}\delta_i \geq a\triangleq\lim\inf \delta_N^{-}\eta_N^{-}$ and $\max_{i,\tau_i\leq M} \delta_i^{\kappa} \leq b\triangleq\lim\sup \delta_N^{+,\kappa}(\eta_N^{+,\kappa}+M)$, we obtain:
$$
\frac{\phi\circ g_N^{-1}(\delta_i^{\kappa})-\phi\circ g_N^{-1}(\delta_i)}{\delta_i^{\kappa}}\geq \inf_{x\in[a,b]} \left(\phi \circ g_N^{-1}\right)'(x) \frac{\delta_i^{\kappa}-\delta_i}{\delta_i^{\kappa}}
$$
Since $\inf_{x\in[a,b]} (\phi\circ g_N^{-1})'>0$ by Assumption \ref{ass:u}-$ii)$, we get:
\begin{equation}
\frac{1}{n}\sum_{\tau_i \leq M} \frac{\delta_i^{\kappa}-\delta_i}{\delta_i^{\kappa}}(B_N+\tau_i I_N)\xrightarrow[\kappa\to 0]{}0.
\label{eq:convergence_2}
\end{equation}
Using the convergences \eqref{eq:convergence_1} and \eqref{eq:convergence_2}, we can prove that:
$$
\lim_{\kappa\to 0}\lim\sup_N\left\|T_N^{M}(\delta_1^{\kappa},\cdots,\delta_n^{\kappa})-T_N^{M}(\delta_1,\cdots,\delta_n)\right\|\to 0.
$$
This can be easily seen by noting that:
\begin{align*}
&T_N^{M}(\delta_1,\cdots,\delta_n)-T_N^{M}(\delta_1^{\kappa},\cdots,\delta_n^{\kappa})\\
&=T_N^{M}(\delta_1,\cdots,\delta_n)\frac{1}{n}\sum_{\tau_i\leq M}\left\{\frac{\phi\circ g_N^{-1}(\delta_i^{\kappa})-\phi\circ g_N^{-1}(\delta_i)}{\delta_i^{\kappa}}+\phi\circ g_N^{-1}(\delta_i)\left[\frac{1}{\delta_i^{\kappa}}-\frac{1}{\delta_i}\right]\right\}(B_N+\tau_i I_N)\\
&\times T_N^{M}(\delta_1^{\kappa},\cdots,\delta_n^{\kappa}).
\end{align*}
One can thus choose $\kappa_0$ in such a way that for all $\kappa \leq \kappa_0$
$$
\lim\sup\left\|T_N^{M}(\delta_1^{\kappa},\cdots,\delta_n^{\kappa})-T_N^{M}(\delta_1,\cdots,\delta_n)\right\|\leq \frac{\varepsilon}{3}.
$$
From \eqref{eq:fromM_1} and \eqref{eq:fromM_2}, we therefore get for all $\kappa\leq \kappa_0$
$$
\lim\sup_N \left\|T_N(\delta_1,\cdots,\delta_n)-T_N^{M_\kappa}(\delta_1^{\kappa},\cdots,\delta_n^{\kappa})\right\| \leq \epsilon.
$$
In an equivalent way, we therefore have, for each $\varepsilon >0$, there exists $\kappa_0$ such that for any $\kappa\leq \kappa_0$ and all large $n$,
$$
\max_{1\leq j\leq n}\left|1-\frac{\delta_j}{\delta_j^{\kappa}}\right| \leq \varepsilon.
$$
Using this result, we will show that for each $\ell>0$, there exist $\kappa>0$ small enough, such that a.s.,
$$
(1-\ell)\frac{1}{n}\sum_{i=1}^n \frac{\psi_N(\delta_i)}{\delta_i}y_iy_i^* \preceq \frac{1}{n}\sum_{i=1}^n \frac{\psi_N(\delta_i^{\kappa})}{\delta_i^{\kappa}}y_iy_i^*   \preceq(1+\ell)\frac{1}{n}\sum_{i=1}^n \frac{\psi_N(\delta_i)}{\delta_i}y_iy_i^*.
$$
To this end, it suffices to show that for each $\varepsilon >0$  and $\kappa$ small enough:
$$
\max_{1\leq i\leq n}\left|\psi_N(\delta_i)-\psi_N(\delta_i^{\kappa})\right|\leq \varepsilon.
$$
If this was not true, then one can find a sequence $(n)$ over which:
\begin{equation}
\max_{1\leq i\leq (n)}\left|\psi_{N(n)}(\delta_i)-\psi_{N(n)}(\delta_i^{\kappa})\right| \geq \epsilon.
\label{eq:contradiction_with}
\end{equation}
for any small $\kappa$. Since the sequence function $\psi_N$ converge uniformly, one can extract a subsequence $(p)$ from $(n)$ such that $c_{(p)}\to c^{*}$ and $\psi_{(p)}$ converge uniformly to $\psi^{*}$. On the other hand, we know that for any arbitrairly small $r$ there exists $\kappa_0$ such that for any $\kappa \leq \kappa_0$ and for all large $n$,
$$
\max_{1\leq j\leq n}\left|1-\frac{\delta_j}{\delta_j^{\kappa}}\right| \leq  r.
$$
or also, for all $j=1,\cdots,n$, 
$$
\delta_j \geq (1-r)\delta_j^{\kappa}.
$$
Let $x_0$ be such that $\psi^{*}(x_0(1-r))\geq \psi_{\infty}(1-\epsilon/3)$. Since $\psi^{*}$ is increasing and bounded at infinity by $\psi_{\infty}$, we have, for any $x,y\geq x_0(1-r)$, 
$$
\left|\psi^{*}(x)-\psi^{*}(y)\right| \leq \frac{\epsilon}{3}.
$$
Consider the indices $i$ such that $\delta_i^{\kappa}\geq x_0$, and thus $\delta_i\geq x_0(1-r)$. Take $n$ large enough such that:
$$
\|\psi_N-\psi^{*}\|\leq \frac{\epsilon}{3}.
$$
Then, for those indices, one can prove that:
\begin{align*}
\max_{\substack{1\leq j\leq (n)\\
\delta_j^{\kappa}\geq x_0 }}\left|\psi_{N(n)}(\delta_j)-\psi_{N(n)}(\delta_j^{\kappa})\right|&\leq \max_{\substack{1\leq j\leq (n)\\
\delta_j^{\kappa}\geq x_0 }} \left|\psi_{N(n)}(\delta_j)-\psi^{*}(\delta_j)\right|\\
&+\left|\psi^{*}(\delta_j)-\psi^{*}(\delta_j^{\kappa})\right|+\left|\psi^{*}(\delta_j^{\kappa})-\psi_{N(n)}(\delta_i^{\kappa})\right|\leq \epsilon.
\end{align*}
Consider now the indices $i$ such that   $\delta_i^{\kappa}\leq x_0$. For those indices, we have:
\begin{align*}
\psi(\delta_i^{\kappa})-\psi(\delta_i) &u(0)\left(\leq \delta_i^{\kappa}-\delta_i \right)\\
&=u(0) \frac{\delta_i^{\kappa}-\delta_i}{x_0}x_0\\
&\leq u(0)x_0 \frac{\delta_i^{\kappa}-\delta_i}{\delta_i^{\kappa}}.
\end{align*}
Taking $r\leq \frac{\epsilon}{x_0u(0)}$, we will get:
$$
\left|\psi(\delta_i^{\kappa})-\psi(\delta_i)\right| \leq \epsilon
$$
which is in contradiction with \eqref{eq:contradiction_with}. We therefore have for each $\ell>0$,
$$
(1-\ell)^2\frac{1}{n}\sum_{i=1}^n \frac{\psi_N(\delta_i)}{\delta_i}y_iy_i^* \preceq \hat{C}_N   \preceq(1+\ell)^2\frac{1}{n}\sum_{i=1}^n \frac{\psi_N(\delta_i)}{\delta_i}y_iy_i^*
$$
which therefore implies that $\|\hat{C}_N-\hat{S}_N\|\asto 0$. This completes the proof.

%Function $\phi\circ g_N^{-1}$ is increasing on $[0,\infty]$, therefore its derivative should be strictly positive on any interior closed 
%Since $\lambda_1(T_N(\delta_1,\cdots,\delta_n))\geq \frac{\lim\inf\delta_N^{-}\eta_N^{-}}{\phi_{\infty}\lim\sup \|B_N\|}$, we obtain:
%Decompose \eqref{eq:zero_equation} as:
%$$
%0=\alpha_1+\alpha_2
%$$
%where:
%\begin{align*}
%\alpha_1&=\frac{1}{N}\tr \frac{B_N+\tau_j I_N}{\delta_j}T_N(\delta_1,\cdots,\delta_n)-\frac{1}{N}\tr \frac{B_N+\tau_j I_N}{\delta_j}T_N^{\kappa}(\delta_1,\cdots,\delta_n) \\
%\alpha_2&=\frac{1}{N}\tr \frac{B_N+\tau_j I_N}{\delta_j}T_N^{\kappa}(\delta_1,\cdots,\delta_n) -\frac{1}{N}\tr \frac{B_N+\tau_j I_N}{\delta_j^{\kappa}}T_N^{\kappa}(\delta_1^{\kappa},\cdots,\delta_n^{\kappa}).
%\end{align*}
%One can easily see that $\alpha_1\leq 0$. Therefore, necessarily $\alpha_2\geq 0$ and $\alpha_2=|\alpha_1|$. Since $\delta_j^{\kappa}\leq \delta_j$ for all $j=1,\cdots, n$,
%$$
%\alpha_2\geq \frac{1}{N}\tr \frac{B_N+\tau_j I_N}{\delta_j}
%$$ 
\end{paragraph}
\end{paragraph}
\appendix
\section*{Appendix: Technical Lemmas}
\label{app:technical}
This appendix gathers some technical Lemmas that will help control quadratic forms of the type:
$$
z_j^*R_j\left(\frac{1}{n}\sum_{i=1,i\neq j}^n R_i z_iz_i^*R_i^*\right)^{-1}R_jz_j,
$$
where  $z_1,\cdots,z_n$ are independent random vectors with size $\overline{N}\times 1$  and $R_1,\cdots,R_n$ are $n$ matrices of size $\overline{N}\times N$ independent of $z_1,\cdots,z_n$  and whose eigenvalues are bounded above and below by constants independent of $n$ and $N$. 
The control of this quadratic form can be performed using the most well-known trace Lemma of Silverstein etal \cite[Lemma 2.7]{BaiSil98}, provided that we can guarantee that the infimum of the set $\mathcal{S}$ of   smallest eigenvalues of  matrices $\frac{1}{n}\sum_{i=1,i\neq j}^n R_i z_iz_i^*R_i, j=1,\cdots,n$ is above zero uniformly in $n$ and $N$ or more formally,
$$
\min_{1\leq j \leq n}\left\{\lambda_1\left(\frac{1}{n}\sum_{i=1,i\neq j}^n R_i z_iz_i^*R_i^*\right), j=1,\cdots,n\right\} >\epsilon,
$$
for some $\epsilon >0$ a.s. for $n$ large enough. 
 This is however a challenging task since  the question of the smallest eigenvalue of  matrices of the form  $\frac{1}{n}\sum_{i=1}^n R_i z_iz_i^*R_i^*$ being above zero almost surely  was implicitly   raised in \cite{wagner} where this fact was assumed because no immediate answer can be provided in general.  It was only recently that we have provided a rigorous proof thereof under the Gaussian setting \cite{smallest_eigenvalue}.
%has only recently been studied in our work \cite{smallest_eigenvalue} where we show that the smallest eigenvalue of  $\frac{1}{n} \sum_{i=1}^n R_i z_iz_i^*R_i^*$ is above zero, under the Gaussian setting.
 %Recently, we proved this result  under the Gaussian setting and when the smallest eigenvalue of $\frac{1}{n} \sum_{i=1}^n R_i z_iz_i^*R_i^*$ is above zero:
%$$
%\liminf \min_{1\leq i\leq n} \lambda_1(R_iR_i^*) >0.
%$$
 In this appendix, we extend this result to the random vector model of the present work, i.e. $z_i=\left[s_i,w_i\right]$ with $s_i\sim\mathcal{CN}(0,I_K)$ and $w_i$ zero-mean unitarily invariant satisfying $\|w_i\|^2=N$.  The control of the infimum of the set $\mathcal{S}$ will be shown along the same lines of the proof of  Lemma 1 in \cite{couillet-13a}.

In the sequel, we will start by bounding the maximum eigenvalue of  $\frac{1}{n}\sum_{i=1}^n R_i z_iz_i^*R_i^*$ when $z_1,\cdots,z_n$ are Gaussian random vectors. To this end, we will start by introducing the following concentration inequality, the proof of which is provided for sake of completeness.
	\begin{lemma}
\label{lemma:gamma_function}
		Let $\gamma_1,\cdots,\gamma_n$ be $n$ independent random variables having an exponential distribution with rate parameter $1$, and $\left(\alpha_i\right)_{i=1}^n$, $n$ be positive scalars. Then, there exists $C$ such that for any $t>0$:
		$$
		\mathbb{P}\left[\sum_{i=1}^n \alpha_i \gamma_i > t\right]\leq C\exp\left(-\min\left(\frac{t^2}{4\sum_{i=1}^n \alpha_i^2},\frac{t}{4\max_{1\leq i\leq n}\alpha_i}\right)\right).
		$$
	\end{lemma}
	\begin{proof}
	Let $s$ be a positive scalar such that $s < \frac{1}{2\alpha_i}$ for all $i=1,\cdots,n$. Then:
\begin{align*}
	\mathbb{P}\left[\sum_{i=1}^n \alpha_i \gamma_i> t\right]&=\mathbb{P}\left[\exp\left(s\sum_{i=1}^n \alpha_i \gamma_i\right)>\exp(ts)\right] \\																																																				   &\leq \frac{\exp(-ts)}{\prod_{j=1}^n (1-s\alpha_j)}.
\end{align*}
Now, using the inequality $-\log(1-x)\leq  x+x^2$ for $x\leq \frac{1}{2}$, we get:
$$
\frac{1}{1-s\alpha_j}\leq \exp(s\alpha_j)\exp(s^2\alpha_j^2)\leq \exp(\frac{1}{2})\exp(s^2\alpha_j^2),
$$
thereby yielding:
\begin{align*}
	\mathbb{P}\left[\sum_{i=1}^n \alpha_i \gamma_i > t\right]\leq \exp(\frac{1}{2})\exp(-ts+\sum_{i=1}^n {s^2\alpha_i^2}).
\end{align*}
Two cases have to be considered. If $\frac{t\max_{1\leq i\leq n}\alpha_i}{2\sum_{i=1}^n \alpha_i^2}\leq \frac{1}{2}$. Then, setting $s=\frac{t}{2\sum_{i=1}^n \alpha_i^2}$ yields:
\begin{equation}
	\mathbb{P}\left[\sum_{i=1}^n \alpha_i \gamma_i > t\right]\leq C\exp\left(-\frac{t^2}{4\sum_{i=1}^n \alpha_i^2}\right)
\label{eq:inequality_1}
\end{equation}
where $C=\exp(\frac{1}{2})$.
Otherwise, if $\frac{t\max_{1\leq i\leq n}\alpha_i}{2\sum_{i=1}^n \alpha_i^2}\geq{\frac{1}{2}}$. Then, $\sum_{i=1}^n \alpha_i^2 \leq t\max_{1\leq i\leq n}\alpha_i$. Set $s=\frac{1}{2\max_{1\leq i\leq n}\alpha_i}$, then we have:
\begin{equation}
	\mathbb{P}\left[\sum_{i=1}^n \alpha_i \gamma_i > t\right]\leq C\exp\left(-\frac{t}{4\max_{1\leq i\leq n}\alpha_i}\right)
\label{eq:inequality_2}
\end{equation}
Gathering \eqref{eq:inequality_1} and \eqref{eq:inequality_2} yields the desired result.
%Taking $c_2=s$ and $c_1=\exp\left(s^2 \overline{\alpha}^2\right)$ establishes the desired result.
	\end{proof}

With this Lemma at hand, we will now control the maximum eigenvalue of $\frac{1}{n}R_iz_iz_i^*R_i$. We have in particular, the following result:
\begin{lemma}
\label{lemma:bounded_norm}
Let $z_1,\cdots,z_n$ be $n$ independent $\overline{N}\times 1$ standard Gaussian vectors. Consider $\left(R_i\right)_{i=1}^n$ a family of $N\times \overline{N}$ matrices with uniformly bounded spectral norm, i.e,
$$
\lim\sup_{n} \max_{1\leq i\leq n}\left\|R_i\right\|<+\infty. 
$$
Let $\Sigma$ be given by:
$$
\Sigma=\frac{1}{n}\sum_{i=1}^n R_i z_iz_i^*R_i^*
$$
Then, there exists a constant $K_{\rm max}$ such that a.s. for $n$ large enough,
$$
\left\|\Sigma\right\|<K_{\rm max}.
$$
\end{lemma}
\begin{proof}
The proof relies on the observation that:
$$
\|\Sigma\| =\max_{\|a\|=1} a^*\Sigma a.
$$
Based on the result of Lemma \ref{lemma:gamma_function}, a concentration inequality involving the term $a^*\Sigma a$ can be established.
Define $u_i=\frac{R_i a}{\left\|R_i a\right\|}$ and expand $a^*\Sigma a$ as:
\begin{align*}
a^*\Sigma a&=\frac{1}{n}\sum_{i=1}^n a^* R_i z_iz_i^* R_i a \\
&=\frac{1}{n}\sum_{i=1}^n \left(a^*R_iR_i^* a\right) \left|u_i^*z_i\right|^2.
\end{align*}
Since $u_i$ is unitary, the random quantity $u_i^*z_i$ is a Gaussian random variable with zero mean and variance $1$. Hence, $\left|z_i^*u_i\right|^2,i=1,\cdots,n$ is a sequence of $n$ independent exponential distributed random variables with rate parameter $1$. Applying Lemma \ref{lemma:gamma_function}, we get:
\begin{align*}
\mathbb{P}\left[a^*\Sigma a> t\right]&=\mathbb{P}\left[\sum_{i=1}^n \left(a^*R_iR_i^*a\right)\left|u_i^* z_i\right|^2> nt\right] \\
&\leq C\exp\left(-\min\left(\frac{n^2t^2}{4\sum_{i=1}^n a^*R_iR_i^* a},\frac{nt}{4\max_{1\leq i\leq n}a^*R_iR_i^*a}\right)\right),\\
\end{align*}
where $C$ is some constant independent of $n$ and $N$. 
For $t\geq 1$, we therefore have:
\begin{equation}
\mathbb{P}\left[a^*\Sigma a> t\right]\leq C\exp\left(-\frac{nt}{4\max_{1\leq i\leq n}\left\|R_i\right\|^2}\right).
\label{eq:inequality_fundamental}
\end{equation}
With the above inequality at hand, we are now in position to control the behaviour of the spectral norm of $\Sigma$. For that, we will resort to the well-known $\epsilon-$net argument. Let $\mathcal{S}$ be an $\frac{1}{2}-$ net of the unit sphere of $\mathbb{C}^{N}$. Using Lemma 2.3.2 of \cite{tao}, we have:
$$
\mathbb{P}\left[\left\|\Sigma^{\frac{1}{2}}\right\|\geq t\right] \leq \mathbb{P}\left[\bigcup_{a\in\mathcal{S}}a^*\Sigma a> \frac{t^2}{4}\right],
$$
Using \eqref{eq:inequality_fundamental}, we obtain that each of the probabilities of $\mathbb{P}\left[a^*\Sigma a> \frac{t^2}{4}\right]$ is bounded by $C\exp(-\frac{ct^2n}{4})$ for $t$ and $n$ large enough with $c$ some constant independent of $n$.
On the other hand, the cardinality of $\mathcal{S}$ is of order $\left(\mathcal{O}(1)\right)^n$. By taking $t$ large enough, this term can be absorbed into the exponential gain of  $C\exp(-\frac{cnt}{4})$. 
For some $t$ large enough the event $\left\{\left\|\Sigma^{\frac{1}{2}}\right\|<t\right\}$ holds with overwhelming probability. Setting $K_{\rm max}\geq t$. We have thus a.s. for $n$ large enough,
$$
 \left\|\Sigma^{\frac{1}{2}}\right\|<K_{\rm max}.
$$
\end{proof}

All the above results are derived under the assumption of a Gaussian setting. Before going further into the proofs of the main lemmas of this appendix, we will show that the considered random model of the paper is equivalent to a Gaussian model. In particular, we have the following result: 
\begin{lemma}
\label{lemma:gaussian}
Let $z_1,\cdots,z_n \in \mathbb{C}^{\overline{N}}$ be $n$ independent and identically distributed vectors such that $z_i=\left[s_i^{\mbox{\tiny T}},w_i^{\mbox{\tiny T}}\right]^{\mbox{\tiny T}}$ where $s_i$ and $w_i$ are independent and distributed as:
\begin{itemize}
	\item $s_i\sim \mathcal{CN}(0,{ I}_K)$
	\item $w_i$ is unitarily invariant zero-mean vector such that $\|w_i\|=N$.
	\end{itemize}
Write $w_i$ as $w_i=\frac{\sqrt{N}\tilde{w}_i}{\|\tilde{w}_i\|}$ with $\tilde{w}_i\sim\mathcal{CN}(0,I_N)$, and denote by $\Sigma$ the $N\times N$ matrix given by:
$$
\Sigma=\frac{1}{n}\sum_{i=1}^n R_i z_iz_i^*R_i^*
$$
where $R_i=\left[R_{i,1}\hspace{0.1cm}R_{i,2}\right]$, $R_{i,1}\in\mathbb{C}^{N\times K}$ and $R_{i,2}\in\mathbb{C}^{N\times N}$ are some deterministic matrices with uniformly bounded norm.
Let $\tilde{z}_i=\left[s_i^{\mbox{\tiny T}},\tilde{w}_i^{\mbox{\tiny T}}\right]^{\mbox{\tiny T}}$, and $\tilde{\Sigma}$ be given by:
$$
\tilde{\Sigma}=\frac{1}{n}\sum_{i=1}^n R_i \tilde{z}_i\tilde{z}_i^*R_i^*
$$
Then, in the asymptotic regime, 
$$
\|\Sigma-\tilde{\Sigma}\|\stackrel{a.s}{\to} 0.
$$
\end{lemma}
\begin{proof}
Notice that $\Sigma$ can be written as:
\begin{align*}
\Sigma=\frac{1}{n}\sum_{i=1}^n R_i D_i \tilde{z}_i\tilde{z}_i^*D_i R_i^*
\end{align*}
where 
$$
D_i=\begin{bmatrix}
I_K & 0 \\
0 & \frac{\sqrt{N}}{\|\tilde{w}_i\|} I_N
\end{bmatrix}.
$$
Then, 
\begin{align*}
\Sigma-\tilde{\Sigma}&=\frac{1}{n}\sum_{i=1}^n R_i\left(D_i-I\right)\tilde{z}_i\tilde{z}_i^*\left(D_i-I\right)R_i^* +\frac{1}{n}\sum_{i=1}^n R_i \tilde{z}_i\tilde{z}_i^*\left(D_i-I\right)R_i^* \\
&+\frac{1}{n}\sum_{i=1}^n R_i\left(D_i-I\right)\tilde{z}_i\tilde{z}_i^*R_i^*\\
&=\Theta_1+\Theta_2+\Theta_3.
\end{align*}
In the sequel, we will prove that the spectral norms of $\Theta_i,i=1,2,3$ converge to zero almost surely. We will treat only the term $\Theta_2$ since the treatment of $\Theta_1$ and $\Theta_3$ relies on the same arguments. Expanding  $\Theta_2$ using $R_i=\left[R_{i,1}\hspace{0.1cm}R_{i,2}\right]$, we get:
\begin{align*}
\Theta_2&= \frac{1}{n}\sum_{i=1}^n R_{i,1}s_i\tilde{w}_i^*\left(\frac{\sqrt{N}}{\|\tilde{w}_i\|}-1\right)R_{i,2}^* + \frac{1}{n}\sum_{i=1}^n R_{i,2}\tilde{w}_i\tilde{w}_i^*R_{i,2}^* \left(\frac{\sqrt{N}}{\|\tilde{w}_i\|}-1\right)\\
&\triangleq \Theta_{1,2} +\Theta_{2,2}.
\end{align*}
Let us control $\Theta_{1,2}$. 
\begin{align*}
\|\Theta_{1,2}\|&=\sup_{\|a\|=1,\|b\|=1} \frac{1}{n}\left|\sum_{i=1}^n a^{*}R_{i,1}s_i\tilde{w}_i^* R_{i,2}^* b \left(\frac{\sqrt{N}}{\|\tilde{w}_i\|}-1\right) \right| \\
&\leq \max_i \left|\frac{\sqrt{N}}{\|\tilde{w}_i\|}-1\right|\sup_{\|a\|=1,\|b\|=1} \frac{1}{n}\sum_{i=1}^n \left|a^*R_{i,1}s_i\right|\left|\tilde{w}_i^*R_{i,2}b\right| \\
&\leq \max_i \left|\frac{\sqrt{N}}{\|\tilde{w}_i\|}-1\right| \sqrt{\sup_{\|a\|=1} \frac{1}{n}\sum_{i=1}^n a^* R_{i,1}s_is_i^*R_{i,1}^* a}\sqrt{\sup_{\|b=1\|}\frac{1}{n}\sum_{i=1}^n \tilde{w}_i^* R_{i,2}b b^* R_{i,2}^* \tilde{w}_{i}}\\
&\leq \max_i \left|\frac{\sqrt{N}}{\|\tilde{w}_i\|}-1\right| \sqrt{\|\frac{1}{n}\sum_{i=1}^n R_{i,1}s_{i}s_i^*R_{i,1}^*\|}\sqrt{\|\frac{1}{n}\sum_{i=1}^n R_{i,2}\tilde{w}_i\tilde{w}_i^*R_{i,2}^*\|}
\end{align*}
Using the facts that $\max_i \left|\frac{\sqrt{N}}{\|\tilde{w}_i\|}-1\right|$ converges almost surely to zero and $\|\frac{1}{n}\sum_{i=1}^n R_{i,1}s_{i}s_i^*R_{i,1}\|$ and $\|\frac{1}{n}\sum_{i=1}^n R_{i,2}\tilde{w}_i\tilde{w}_i^*R_{i,2}\|$ are almost surely bounded as a result of Lemma \ref{lemma:bounded_norm}, we have:
$$
\left\|\Theta_{1,2}\right\|\asto 0.
$$
Similarly, we can also prove that:
$$
\left\|\Theta_{2,2}\right\|\asto 0,
$$
thereby implying that:
$$
\left\|\Theta_2\right\|\asto 0.
$$
%If $K$ grow with $n$ and $N$,  the term $\xi_1$ could be treated similarly. 
\end{proof}
Using the same notations of Lemma \ref{lemma:gaussian}, consider $\Sigma_j$ and $\tilde{\Sigma}_j$ the $N\times N$ matrices given by:
\begin{align*}
\Sigma_j&=\frac{1}{n}\sum_{i=1,i\neq j}^n R_iz_iz_i^*R_i  \\
\tilde{\Sigma}_j&=\frac{1}{n}\sum_{i=1,i\neq j}^n R_i\tilde{z}_i\tilde{z}_i R_i.
\end{align*}
Arguing along the same lines as in the proof of Lemma \ref{lemma:gaussian}, we can show that:
\begin{equation}
\max_{1\leq j\leq n} \left\|\Sigma_j-\tilde{\Sigma}_j\right\|\asto 0.
\label{eq:convergence}
\end{equation}
This observation is essential to facilitate the proof of the first main result of this Appendix which is about showing that:
$$
\min_{1\leq j\leq n}\left\{\lambda_1\left(\Sigma_j\right),j=1,\cdots,n\right\} >\epsilon
$$
In fact, from the convergence inequality in \eqref{eq:convergence}, we can see that the proof can be reduced to showing this result when $\Sigma_j$ is replaced with $\tilde{\Sigma}_j$. The proof of the following Lemma will rely on this observation:
\begin{lemma}
Let $z_1,\cdots,z_n \in \mathbb{C}^{\overline{N}}$ be $n$ independent and identically distributed vectors such that $z_i=\left[s_i^{\mbox{\tiny T}},w_i^{\mbox{\tiny T}}\right]^{\mbox{\tiny T}}$ where $s_i$ and $w_i$ are independent and distributed as:
\begin{itemize}
	\item $s_i\sim \mathcal{CN}(0,{ I}_K)$
	\item $w_i$ is unitarily invariant zero-mean vector such that $\|w_i\|=N$.
	\end{itemize}
Define matrices  $\Sigma$ and $\Sigma_j$  as:
\begin{align*}
\Sigma&=\frac{1}{n}\sum_{i=1}^n R_iz_iz_i^*R_i \\
\Sigma_j&=\frac{1}{n}\sum_{i=1,i\neq j}R_i z_iz_i^*R_i,
\end{align*}
where $\left(R_i\right)_{i=1}^n$ are $N\times \overline{N}$ matrices satisfying:
$$
\lim\inf_N \min_{1\leq i\leq n} \lambda_1(R_iR_i^*) >0
$$
and 
$$
\lim\sup_N \max_{1\leq i\leq n} \lambda_N(R_iR_i^*) <+\infty,
$$
Consider the asymptotic regime of Assumption \ref{ass:regime}. Therefore, there exists $\epsilon >0$ such that for all large $n$  a.s.,
$$
\lambda_1(\Sigma)\geq \min_{1\leq j\leq n}\lambda_1(\Sigma_j) > \epsilon.
$$
\label{lemma:smallest_eigenvalue}
\end{lemma} 
\begin{proof}
It is clear from the discussion before the statement of the above lemma that we can assume  $z_1,\cdots,z_n$ to be standard Gaussian vectors.
Under the Gaussian setting, the fact that the smallest eigenvalue of $\Sigma$ is almost surely bounded above zero can be deduced from corollary 5 of our work in \cite{smallest_eigenvalue}.
It remains thus to treat that of $\Sigma_j$. To this end, we will resort to the same kind of the arguments as those used in the proof of \cite[Lemma 1]{couillet-13a}. Notice first that we can assume without loss of generality that $\lambda_1(\Sigma_j)\neq \lambda_1(\Sigma)$, for all $j=1,\cdots,n$. By definition, the eigenvalues of $\Sigma_j$ are solutions in $\lambda$ of the following equation:
$$
\det(\Sigma_j-\lambda I_N) = 0.
$$
Developing the above equation, we obtain:
\begin{align*}
\det\left(\Sigma_j-\lambda I_N\right)&=\det(\Sigma-\frac{1}{n}R_jz_jz_j^*R_j^*-\lambda I_N) \\
&=\det Q(\lambda)\det\left(I_N-Q(\lambda)^{\frac{1}{2}}\frac{1}{n}R_jz_jz_j^*R_j^* Q(\lambda)^{\frac{1}{2}}\right) \\
&=\det  Q(\lambda)\left(1-\frac{1}{n}z_j^*R_j^*Q(\lambda)R_j z_j\right),
\end{align*}
where $Q(\lambda)=\left(\Sigma-\lambda I_N\right)^{-1}$.
If $\lambda$ is an eigenvalue of $\Sigma_j$ different from that of $\Sigma$, then necessarily:
$$
\frac{1}{n}z_j^*R_j^*Q(\lambda)R_jz_j=1.
$$
Building on the ideas of \cite{couillet-13a}, we propose to study the behaviour of function:
$$
f_{n,j}(x)=\frac{1}{n}z_j^*R_j^*Q(x)R_j z_j,
$$
in a neighborhood of zero. The result of the lemma follows if we prove that there exists $\xi>0$ such that $f_{n,j}(x)<1$ for all $x\in\left[0,\xi\right]$ and $j=1,\cdots,n$ a.s. for $n$ large enough.
From our recent result in \cite{smallest_eigenvalue}, we know that there exists $\eta>0$ such that a.s. for $n$ large enough, $\lambda_1(\Sigma) > \eta$. Functions $x\mapsto f_{n,j}(x)$ being increasing in the interval $\left[0,\eta\right]$, it suffices thus to show that there exists $\xi$ in $\left[0,\eta\right]$ such that $f_{n,j}(\xi)<1$ a.s. for $n$ large enough. 
 
Let us start by analyzing the behaviour of $f_{n,j}(x)$ for $x<0$. Define $Q_j(x)$ as $Q_j(x)=(\Sigma_j-xI_N)^{-1}$. Using the matrix inversion relation: $a^*(A+aa^*)^{-1}a=\frac{a^*A^{-1}a}{1+a^*A^{-1}a}$ for  $a\in\mathbb{C}^{N\times 1}$ and ${A}$ any $N\times N$ invertible matrix, we obtain:
$$
f_{n,j}(x)=\frac{\frac{1}{n}z_j^*R_j^*Q_j(z)R_j z_j}{1+\frac{1}{n}z_j^*R_j^*Q_j(x)R_j z_j}=1-\frac{1}{1+\frac{1}{n}z_j^*R_j^*Q_j(x)R_j z_j}.
$$
Now, for $x<0$, using the trace lemma of Silverstein etal \cite[Lemma 2.7]{BaiSil98} in conjunction with the rank-one perturbation Lemma \cite[Lemma 2.6]{SilBai95}, we can prove that:
$$
\max_{1\leq j\leq n}\left|\frac{1}{n}z_j^*R_j^*Q_j(x)R_j z_j-\frac{1}{n}\tr Q(x)R_jR_j^*\right|\asto 0.
$$
and thus:
$$
\max_{1\leq j\leq n}\left|f_{n,j}(x)-1+\frac{1}{1+\frac{1}{n}\tr Q(x)R_jR_j^*}\right|\asto 0.
$$
Therefore, for $\epsilon<\min_{1\leq j\leq n}\frac{1/2}{1+\frac{\lambda_N(R_jR_j^*)}{\eta}}$ and $x<0$, we have for $n$ large enough a.s.,
\begin{equation}
\forall j=1,\cdots,n \hspace{0.2cm} f_{n,j}(x) \leq 1-\frac{1}{1+\frac{1}{n}\tr  QR_jR_j^*} +\epsilon.
\label{eq:ff}
\end{equation}
On the other hand, since the smallest eigenvalue of $\Sigma$ is greater than $\eta$,  we have for $n$ large enough, a.s.
\begin{equation}
\frac{1}{n}\tr R_jR_j^*Q \leq \frac{\lambda_N(R_jR_j^*)}{\eta}
\label{eq:rjj}
\end{equation}
Plugging \eqref{eq:rjj} into \eqref{eq:ff}, we obtain that for each $x<0$ we have for $n$ large enough, a.s.
\begin{equation}
\forall j=1,\cdots,n \hspace{0.2cm} f_{n,j}(x) \leq 1-\epsilon.
\label{eq:frjj}
\end{equation}
Now, we will consider the analysis of functions $f_{n,j}(x)$ on the open interval $U=(-\frac{\eta}{2},\frac{\eta}{2})$. Note that on this interval, functions $x\mapsto f_{n,i}(x),i=1,\cdots,n$ are well-defined and continuously differentiable. Moreover, for each $x\in U$, we have:
$$
f_{n,j}^{'}(x)=\frac{1}{n}z_j^*R_j^* Q^2(x)R_j z_j \leq \frac{\frac{1}{n}z_j^*R_j^*R_jz_j}{(\lambda_1(\Sigma)-\frac{\eta}{2})^2}. 
$$
Moreover, we have:
$$
\max_{1\leq j\leq n}\left|\frac{1}{n}z_j^*R_j^*R_jz_j-\frac{1}{n}\tr R_jR_j^*\right|\asto 0,
$$
The above convergence along with the fact that $\lambda_1(\Sigma)>\eta$ for $n$ large enough a.s. yields:
$$
0< f_{n,j}^{'}(x) < \frac{2\lim\sup_n \max_{1\leq j\leq n} \frac{1}{n}\tr R_jR_j^*}{\frac{\eta^2}{4}} \triangleq{K}.
$$
By bounding the derivatives of functions $f_{n,j}$ over $U$, we have by the mean value theorem:
$$
\forall x\in [0,\frac{\eta}{2}],\hspace{0.2cm} f_{n,j}(x) < f_{n,j}(-x) +2xK^{'}
$$
Set $\xi=\min(\frac{\eta}{2},\frac{\epsilon}{4K})$. Then,  we know from \eqref{eq:frjj} that:
$$
 f_{n,j}(-\xi) \leq 1-\frac{1/2}{1+\frac{\lambda_N(R_jR_j^*)}{\eta}}
$$
Combining this inequality with the fact that $f_{n,j}(\xi) <f_{n,j}(-\xi) +2xK^{'}$, yields:
$$
f_{n,j}(\xi) \leq 1-\frac{\epsilon}{4K},
$$
thereby finishing the proof.
\end{proof}
We are now in position to state the following key results of this appendix:
\begin{lemma}
Let $z_1,\cdots,z_n \in \mathbb{C}^{\overline{N}}$ be $n$ independent and identically distributed vectors such that $z_i=\left[s_i^{\mbox{\tiny T}},w_i^{\mbox{\tiny T}}\right]^{\mbox{\tiny T}}$ where $s_i$ and $w_i$ are independent and distributed as:
\begin{itemize}
	\item $s_i\sim \mathcal{CN}(0,{ I}_K)$
	\item $w_i$ is unitarily invariant zero-mean vector such that $\|w_i\|=N$.
	\end{itemize}
Let $\left(A_{N,j}\right)_{j=1}^n$ be random matrices independent of $z_1,\cdots,z_n$ and $\kappa$ be a positive constant. Then,
$$
\max_{1\leq j\leq n} 1_{\left\|A_j\right\|\leq \kappa}\left|\frac{1}{N}z_j^*A_{N,j}z_j-\frac{1}{N}\tr A_{N,j}\right|\asto 0.
$$
\label{app:convergence_quadratic}
\end{lemma}
\begin{proof}
Write $w_i=\frac{\sqrt{N}\tilde{w}_i}{\|\tilde{w}_i\|}$ with $\tilde{w}_i\in\mathbb{C}^{\overline{N}-K}$. 
Let $D_j$ be the diagonal matrix given by:
$$
D_j=\begin{bmatrix}
I_K & 0 \\
0 & \frac{\sqrt{N}}{\|\tilde{w}_i\|}
\end{bmatrix}
$$
Let $\tilde{z}_i=\left[s_i^{\mbox{\tiny T}},\tilde{w}_i^{\mbox{\tiny T}}\right]^{\mbox{\tiny T}}$. Then:
\begin{align}
 {1}_{\|A_j\|\leq \kappa}\frac{1}{N}z_j^*A_{N,j}z_j &={1}_{\|A_j\|\leq \kappa}\frac{1}{N}\tilde{z}_j^*D_j A_{N,j} D_j \tilde{z}_j \nonumber \\
&={1}_{\|A_j\|\leq \kappa}\frac{1}{N}\tilde{z}_j^*(D_j-I_{\overline{N}})A_{N,j}(D_j -I_{\overline{N}})\tilde{z}_j +{1}_{\|A_j\|\leq \kappa}\frac{1}{N}\tilde{z}_j^* A_{N,j}(D_j-\overline{N})\tilde{z_j} \nonumber\\
&+{1}_{\|A_j\|\leq \kappa}\frac{1}{N}\tilde{z}_j^*(D_j-I_{\overline{N}})A_{N,j}\tilde{z}_j +{1}_{\|A_j\|\leq \kappa}\frac{1}{N}\tilde{z}_j^*A_{N,j}\tilde{z}_j. \label{eq:above}
\end{align}
Since $\max_j \left|1-\frac{\sqrt{N}}{\|\tilde{w}_i\|}\right|\asto 0$,
$$
\max_{1\leq j\leq n} \|D_j-I_{\overline{N}}\|\asto 0.
$$
Therefore, it is easy to see that the maximum over $j$ of the first three terms in \eqref{eq:above} converge to zero almost surely. The problem unfolds thus to the control of $ {1}_{\|A_{N,j}\|\leq \kappa}\frac{1}{N}\tilde{z}_j^*A_{N,j}\tilde{z}_j$. Let $\mathbb{E}_{\tilde{z}_j}$ denote the expectation with respect to $\tilde{z_j}$.
From the trace Lemma of Silverstein etal \cite[Lemma 2.7]{BaiSil98} applied for $p>2$, we obtain:
$$
\mathbb{E}_{\tilde{z}_j}\left[1_{\|A_{N,j}\|\leq \kappa}\left|\frac{1}{N}\tilde{z}_j^*A_{N,j}\tilde{z}_j-\frac{1}{N}\tr A_{N,j}\right|^p\right] \leq \frac{1_{\|A_j\|\leq \kappa}K_p}{N^{\frac{p}{2}}}\left[\left(\frac{\zeta_4}{N}\tr A_{N,j}^{2}\right)^{\frac{p}{2}}+\frac{\zeta_{2p}}{N^{\frac{p}{2}}}\tr A_{N,j}^{p}\right]
$$
where $\zeta$ is any upper-bound on $\mathbb{E}\left[\left|\tilde{z}_{i,j}\right|^\ell\right]$ and $K_p$ a constant dependent only in $p$. Since $1_{\|A_{N,j}\|\leq \kappa}\|A_{N,j}\|\leq \kappa$, we have:
$$
\mathbb{E}_{\tilde{z}_j}\left[1_{\|A_{N,j}\|\leq \kappa}\left|\frac{1}{N}\tilde{z}_j^*A_{N,j}\tilde{z}_j-\frac{1}{N}\tr A_{N,j}\right|\right]\leq \frac{K_p\kappa^p}{N^{\frac{p}{2}}}\left(\zeta_4^{\frac{p}{2}}+\frac{\zeta_{2p}}{N^{\frac{p}{2}-1}}\right)
$$
This bound being independent of $A_{N,j}$, we can take the expectation with respect to $A_{N,j}$ to obtain:
$$
\mathbb{E}\left[1_{\|A_{N,j}\|\leq \kappa}\left|\frac{1}{N}\tilde{z}_j^*A_{N,j}\tilde{z}_j-\frac{1}{N}\tr A_{N,j}\right|^p\right] =\mathcal{O}\left(\frac{1}{N^{\frac{p}{2}}}\right).
$$
Therefore,
$$
\max_{1\leq j\leq n} 1_{\|A_{N,j}\|\leq \kappa}\left|\frac{1}{N}\tilde{z}_j^*A_{N,j}\tilde{z}_j-\frac{1}{N}\tr A_{N,j}\right| \xrightarrow[]{a.s} 0.
$$
\end{proof}
\begin{lemma}
\label{lemma:convergence_deterministic_equivalent}
Let $z_1,\cdots,z_n\in\mathbb{C}^{\overline{N}}$ be $n$ independent and identically distributed vectors such that $z_i=\left[s_i^{\mbox{\tiny T}},w_i^{\mbox{\tiny T}}\right]^{\mbox{\tiny T}}$ where $s_i\in\mathbb{C}^{K}$ and $w_i\in\mathbb{C}^{\overline{N}-K}$ are independent and distributed as:
\begin{itemize}
\item $s_i\sim \mathcal{CN}(0,I_K)$
\item $w_i$ is unitarily invariant zero-mean vector such that $\|w_i\|=N$ where $N \leq \overline{N}$. 
\end{itemize}
Denote by $\left(R_i\right)$ a family of $N\times \overline{N}$ deterministic matrix that satisfy
$$
\lim\sup_n \max_{1\leq i\leq n} \lambda_N(R_iR_i^*) <+\infty
$$
Consider the asymptotic regime of Assumption \ref{ass:regime}. 
Let $\Sigma_j$ be given as:
$$
\Sigma_j=\frac{1}{n}\sum_{\substack{i=1\\ i\neq j}}^n R_i z_iz_i^*R_i^*
$$
Assume that there exists $\epsilon>0$ such that for all large $n$ a.s.,
$$
\lambda_1(\Sigma) \geq \displaystyle{\min_{1\leq j\leq n}}\lambda_1(\Sigma_j) >\epsilon
$$
Then, for any $\Theta_N\in\mathbb{C}^{N\times N}$ with bounded spectral norm, 
\begin{equation}
\max_{1\leq j\leq n} \left|\frac{1}{n}\tr \Theta\Sigma_j^{-1}-\frac{1}{n}\tr \Theta \left(\frac{1}{n}\sum_{j=1}^n \frac{R_jR_j^*}{1+e_j}\right)^{-1}\right| \xrightarrow[]{a.s.} 0
\label{eq:e_j}
\end{equation}
where $e_1,\cdots,e_n$ are the unique solutions to the following system of equations: 
$$
e_k=\frac{1}{n}\tr R_kR_k^*\left(\frac{1}{n}\sum_{j=1}^n \frac{R_jR_j^*}{1+e_j}\right)^{-1}.
$$

\label{lemma:deter}
\end{lemma}
\begin{proof}
To prove lemma \ref{lemma:deter}, we show first that:
$$
\max_{1\leq j\leq n} \left|\frac{1}{n}\tr \Theta\Sigma_j^{-1}-\frac{1}{n}\tr R_jR_j^*\Sigma^{-1}\right|\xrightarrow[]{a.s.}0
$$
From the resolvent identity, we have: 
\begin{align*}
\max_{1\leq j\leq n}\left|\frac{1}{n}\tr \Theta\Sigma_j^{-1}-\frac{1}{n}\tr \Theta \Sigma^{-1}\right|&=\max_{1\leq j\leq n}\left|\frac{1}{n}\tr \Theta\Sigma_j^{-1}\left(\Sigma-\Sigma_j\right)\Sigma_j^{-1}\right|\\
&=\max_{1\leq j\leq n}\left|\frac{1}{n}\tr \Theta \Sigma^{-1}R_j\frac{z_jz_j^*}{n}R_j^*\Sigma_j^{-1}\right| \\
&\leq \max_{1\leq j\leq n} c_N \frac{\|\Theta\|\|R_j\|^2}{\lambda_1(\Sigma)\displaystyle{\min_{1\leq j\leq n}}\lambda_1(\Sigma_j^{-1})} \frac{\displaystyle{\max_{1\leq j\leq n}}z_j^*z_j}{n^2}\\
&= \max_{1\leq j\leq n}c_N \frac{\|\Theta\|\|R_j\|^2}{\lambda_1(\Sigma)\displaystyle{\min_{1\leq j\leq n}}\lambda_1(\Sigma_j^{-1})} \frac{\displaystyle{N+\max_{1\leq j\leq n}}s_j^*s_j}{n^2}
\end{align*}
Since there exists $\epsilon >0$ such that for all large $n$ a.s.
$$
\lambda_1(\Sigma) \geq \displaystyle{\min_{1\leq j\leq n} \lambda_1(\Sigma_j)} >\epsilon. 
$$
and  $\max_{1\leq j \leq n}\frac{1}{n}s_j^*s_j$ is almost surely bounded, we have:
\begin{equation}
\max_{1\leq j\leq n}\left|\frac{1}{n}\tr \Theta\Sigma_j^{-1}-\frac{1}{n}\tr \Theta\Sigma^{-1}\right|\xrightarrow[]{a.s.}0.
\label{eq:tobe_plugged}
\end{equation}
Similarly to the previous Lemma, write $w_i=\frac{\sqrt{N}\tilde{w}_i}{\|\tilde{w}_i\|}$, and let $\tilde{z}_i=\left[s_i^{\mbox{\tiny T}},\tilde{w}_i^{\mbox{\tiny T}}\right]^{\mbox{\tiny T}}$. Denote by $\tilde{\Sigma}=\frac{1}{n}\sum_{i=1}^n R_i z_iz_i^*R_i^*$. Then from Lemma \ref{lemma:gaussian}, 
$$
\|\Sigma-\tilde{\Sigma}\|\xrightarrow[]{a.s}0.
$$
Therefore,
\begin{align}
\displaystyle{\max_{1\leq j\leq n}} \left|\frac{1}{n}\tr \Theta\Sigma^{-1}-\frac{1}{n}\tr \Theta\tilde{\Sigma}^{-1}\right|&=\displaystyle{\max_{1\leq j\leq n}}\left|\frac{1}{n}\tr \Theta\Sigma^{-1}\left(\tilde{\Sigma}-\Sigma\right)\tilde{\Sigma}^{-1} \right|.\nonumber\\
&\leq c_N\|\tilde{\Sigma}-\Sigma\| \displaystyle{\max_{1\leq j\leq n}}\|\Theta\| \|\Sigma^{-1}\|\|\tilde{\Sigma}^{-1}\| \xrightarrow[]{a.s.} 0.\label{eq:above}
\end{align}
Hence, plugging \eqref{eq:above} into \eqref{eq:tobe_plugged}, we get:
$$
\max_{1\leq j\leq n} \left|\frac{1}{n}\tr \Theta \Sigma_j^{-1}-\frac{1}{n}\tr \Theta\tilde{\Sigma}^{-1}\right|\xrightarrow[]{a.s.}0.
$$
The asymptotic convergence of $\frac{1}{n}\tr \Theta (\tilde{\Sigma}-zI_N)^{-1}$ has been studed in \cite{wagner} for $z\in\mathbb{C}_{+}$. Since the 
smallest eigenvalue of $\tilde{\Sigma}$ is almost surely away from zero, we can extend the convergence results for $z=0$ by using the same arguments as those presented in \cite[footnote in page 20]{couillet-pascal-2013}.
%Using similar arguments as those presented in \cite{couillet-pascal-2013}, these 
%From (114) in \cite{wagner}, we have:
%$$
%\max_{1\leq j\leq n}\left|\frac{1}{n}\tr R_jR_j^{*}\tilde{\Sigma}^{-1}-e_j\right|\xrightarrow[]{a.s.}0
%$$
%thereby implying  \eqref{eq:e_j}.
\end{proof}

% For that, we will need the following well-known concentration inequality, whose proof is provided for sake of completeness:
%Let $z_1,\cdots,z_n$ be $n$ random vectors with size $N\times1$ and ${R}_1,\cdots,R_n$ be $n$ matrices of size $N\times N$ independent of $z_1,\cdots,z_n$.  

\bibliographystyle{plain}
\bibliography{./math}

\def\cprime{$'$} \def\cdprime{$''$} \def\cprime{$'$} \def\cprime{$'$}
\begin{thebibliography}{10}

\bibitem{smallest_eigenvalue}
S.~Alouini A.~Kammoun.
\newblock {On the Smallest Eigenvalue of General Gaussian Matrices}.
\newblock {\em Submitted to IEEE Transactions on Information Theory}, 2014.

\bibitem{BaiSil98}
Z.~D. Bai and J.~W. Silverstein.
\newblock No eigenvalues outside the support of the limiting spectral
  distribution of large-dimensional sample covariance matrices.
\newblock {\em Ann. Probab.}, 26(1):316--345, 1998.

\bibitem{BIA09}
P.~Bianchi, J.~Najim, M.~Maida, and M.~Debbah.
\newblock {Performance analysis of some eigen-based hypothesis tests for
  collaborative sensing}.
\newblock In {\em {IEEE 15th Workshop on Statistical Signal Processing
  (SSP'09)}}, pages 5--8, Cardiff, Wales, September 2009.

\bibitem{CAR09}
L.~S. Cardoso, M.~Debbah, P.~Bianchi, and J.~Najim.
\newblock {Cooperative spectrum sensing using random matrix theory}.
\newblock In {\em {International Symposium on Wireless Pervasive Computing
  (ISWPC'08)}}, pages 334--338, Santorini, Greece, May 2008.

\bibitem{chitour-14}
Y.~Chitour, R.~Couillet, and F.~Pascal.
\newblock {On the Convergence of Maronna's M-Estimators of Scatter}.
\newblock {\em IEEE Signal Processing Letters}, 2014.

\bibitem{couillet-13}
R.~Couillet and M.~McKay.
\newblock {Large Dimensional Analysis and Optimization of Robust Shrinkage
  Covariance Matrix Estimators}.
\newblock {\em {To appear in Joint of Multivariate Analysis}}, 2013.

\bibitem{couillet-13a}
R.~Couillet, F.~Pascal, and J.~W. Silverstein.
\newblock {Robust M-Estimation for Array Processing: A Random Matrix Approach}.
\newblock {\em IEEE Transactions on Information Theory}, 61(16):4141--4148,
  2013.

\bibitem{couillet-pascal-2013}
R.~Couillet, F.~Pascal, and J.~W. Silverstein.
\newblock {The Random Matrix Regime of Maronna's M-Estimator With Elliptically
  Distributed Samples}.
\newblock {\em submitted}, 2013.

\bibitem{huber-81}
P.~J. Huber.
\newblock {\em {Robust Statistics}}.
\newblock {Wiley Series in Probability and Statistics John Wiley \& Sons},
  1981.

\bibitem{maronna-76}
R.~A. Maronna.
\newblock {Robust M-Estimators of Multivariate Location and Scatter}.
\newblock {\em {The Annals of Statistics}}, pages 51--67, 1976.

\bibitem{mestre-08}
X.~Mestre.
\newblock {Improved Estimation of Eigenvalues of Covariance Matrices and their
  Associated Subspaces using their Sample Estimates}.
\newblock {\em IEEE Transactions on Information Theory}, 54(11):5113--5129,
  2008.

\bibitem{murihead}
R.~J. Murihead.
\newblock {\em {Aspects of Multivariate Statistical Analysis}}.
\newblock John Wiley \&Sons, Inc., 1982.

\bibitem{nadler-10}
B.~Nadler.
\newblock {Nonparametric Detection of Signals by Information Theoretic
  Criteria: Performance Analyis and Improved Estimator}.
\newblock {\em IEEE Transactions on Signal Processing}, 58(5):2764--2756, 2010.

\bibitem{vallet-10}
{P. Vallet and P. Loubaton and X. Mestre}.
\newblock {Improved Subspace Estimation for Multivariate Observations of High
  Dimension: The Deterministic Signals Case}.
\newblock {\em {IEEE Transactions on Information Theory}}, 58(2), February
  2012.

\bibitem{SilBai95}
J.~W. Silverstein and Z.~D. Bai.
\newblock On the empirical distribution of eigenvalues of a class of
  large-dimensional random matrices.
\newblock {\em J. Multivariate Anal.}, 54(2):175--192, 1995.

\bibitem{tao}
T.~Tao.
\newblock {\em Topics in Random Matrix Theory}, volume 132.
\newblock American Mathematical Society, Graduate Studies in Mathematics, 2012.

\bibitem{valaee}
S.~Valaee, B.~Champagne, and P.~Kabal.
\newblock {Parametric Localization of Distributed Sources}.
\newblock {\em IEEE Transactions on Signal Processing}, 43(9), September 1995.

\bibitem{hachem-13}
{W. Hachem and P. Loubaton and X. Mestre and J. Najim and P. Vallet}.
\newblock {A Subspace Estimator for Fixed Rank Perturbations of Large Random
  Matrices}.
\newblock {\em Journal of Multivariate Analysis}, 114:427--447, February 2013.

\bibitem{wagner}
S.~Wagner, R.~Couillet, M.~Debbah, and D.~T.~M. Slock.
\newblock {Large System Analysis of Linear Precoding in Correlated MISO
  Broadcast Channels Under Limited Feedback}.
\newblock {\em IEEE Transactions on Information Theory}, 58(7), July 2012.

\bibitem{yates}
R.~D. Yates.
\newblock {A Framework for Uplink Power Control in Cellular Radio Systems}.
\newblock {\em IEEE Journal on Selected Areas in Communications},
  13(7):1341--1347, 1995.

\end{thebibliography}
\end{document}